\documentclass[11pt]{article}
\newcommand{\vecuse}[1]{{\boldsymbol{#1}}}
\newcommand{\matuse}[1]{{{\boldsymbol{#1}}}} \newcommand{\matusesigma}[1]{{{\boldsymbol{#1}}}}

\usepackage[margin=1in]{geometry}
\usepackage{setspace}
\usepackage{graphicx} 
\usepackage{bm} 
\usepackage{multirow} 
\usepackage{amsfonts}
\usepackage{amsmath}
\usepackage{xcolor}
\usepackage{hyperref}
\usepackage{soul}
\usepackage{amsthm}
\usepackage[authoryear,round]{natbib}
\hypersetup{
    colorlinks,
    linkcolor={blue!50!black},
    citecolor={blue!50!black},
    urlcolor={blue!50!black}
}
\usepackage{bbm}

\def\I {\mathbbm i}

\title{Connecting reflective asymmetries in multivariate spatial and spatio-temporal covariances}
\author{Drew Yarger}

\newcommand{\mb}[1]{\boldsymbol{#1}}

\usepackage[noend, boxed, linesnumbered]{algorithm2e}
\SetKwFor{For}{for}{}{end for}%
\newtheorem{remark}{Remark}
\newtheorem{example}{Example}
\newtheorem{theorem}{Theorem}
\newtheorem{lemma}{Lemma}
\newtheorem{proposition}{Proposition}
\newtheorem{definition}{Definition}
\usepackage{amsmath,amssymb}
\graphicspath{{images/}}

\allowdisplaybreaks

\begin{document}

\title{Connecting reflective asymmetries in multivariate spatial \\ and spatio-temporal covariances}

\author{Drew Yarger\footnote{Department of Statistics, Purdue University, 150 N. University St, West Lafayette, IN 47907, \url{anyarger@purdue.edu}}}







\maketitle

\begin{abstract}
    In the analysis of multivariate spatial and univariate spatio-temporal data, it is commonly recognized that asymmetric dependence may exist, which can be addressed using an asymmetric (matrix or space-time, respectively) covariance function within a Gaussian process framework. 
This paper introduces a new paradigm for constructing asymmetric space-time covariances, which we refer to as ``reflective asymmetric,'' by leveraging recently-introduced models for multivariate spatial data. 
We first provide new results for reflective asymmetric multivariate spatial models that extends their applicability.
We then propose their asymmetric space-time extension, which come from a substantially different perspective than Lagrangian asymmetric space-time covariances. 
There are fewer parameters in the new models, one controls both the spatial and temporal marginal covariances, and the standard separable model is a special case. 
In simulation studies and analysis of the frequently-studied Irish wind data, these new models also improve model fit and prediction performance, and they can be easier to estimate. 
These features indicate broad applicability for improved analysis in environmental and other space-time data.

\end{abstract}



\section{Introduction}\label{sec:intro}

Accurately representing spatial dependence structures, which may be quantified through a covariance function, has broad applicability in a variety of environmental applications. 
Two particular extended settings are multivariate spatial data \citep{genton15}, where multiple spatially-indexed variables are modeled using a matrix-valued covariance function, and spatio-temporal data, where a variable indexed by both space and time is modeled using a space-time covariance function \citep{chen2021space, porcu202130}. 
In these situations, complex dependencies may arise that do not appear in the univariate spatial setting. 
In particular, there may be asymmetric dependence between variables \citep{li_approach_2011} or in space-time \citep{gneiting2006geostatistical}. 
Such asymmetries are physically realistic. 
For example, the common model for modeling asymmetries in space-time data is the Lagrangian model \citep{porcu202130, salvana2023spatio}, which can be interpreted as a process evolving with directionality through space over time, like advection through the atmosphere or the ocean following the direction of currents.

In this work, we introduce a new perspective we refer to as ``reflective'' asymmetries.
These asymmetries are enabled by asymmetric cross-covariances in the purely spatial multivariate setting recently introduced in \cite{yarger2023multivariate}.
These cross-covariances (describing the covariance between two spatial variables) result from multiplication in the spectral domain of $-\I \textrm{sign}(\langle \vecuse{x}, \tilde{\vecuse{x}}\rangle)$ to the spectral density $f(\vecuse{x})$ of a spatial covariance, where $\vecuse{x} \in \mathbb{R}^d$ is the spatial frequency ($d \in \{1, 2, 3, \dots\}$), $\tilde{\vecuse{x}} \in \mathbb{R}^d$ is a unit-vector directional parameter, $\I$ is the imaginary unit, $\textrm{sign}(z)$ is the sign function, and $\langle \cdot, \cdot\rangle$ is the Euclidean inner product in $\mathbb{R}^d$. 
The resulting term in the cross-covariance is an odd function, reflected across the origin, which motivates the ``reflective'' terminology. 
We extend closed forms to any dimension $d$ for squared-exponential and Cauchy covariances. 
The main focus of this work, however, is applying such constructions to the univariate space-time setting. 

The new space-time models are defined by multiplying $\textrm{sign}(\langle \vecuse{x}, \tilde{\vecuse{x}}\rangle)\textrm{sign}(\eta)$ onto a spectral density $f^*(\lVert \vecuse{x} \rVert, |\eta |)$ of a symmetric space-time covariance, where $\eta \in \mathbb{R}$ is the temporal frequency. 
The choice of sign functions ensures the validity of the covariance when $f^*(\lVert \vecuse{x} \rVert, |\eta |)$ is valid. 
In many cases, the space-time asymmetries are constructed from asymmetric cross-covariances. 
This construction results in a two-term decomposition of the space-time covariance function into symmetric and asymmetric terms. 
These models include some asymmetric versions of Gneiting-type nonseparable space-time covariances \citep{gneiting_nonseparable_2002}. 

These models have some conceptual advantages over the Lagrangian model. 
The nontrivial symmetric and separable space-time model is a special case of the reflective asymmetric model. 
One also maintains control of both the marginal spatial and temporal covariance functions. 
The reflective asymmetric model also has fewer parameters than the full Lagrangian model, providing model parsimony.

We also propose approaches for unconditional simulation, nonstationarity, and computation for large space-time data through Vecchia's approximation \citep{GpGp}. 
In multivariate spatial and univariate space-time simulation studies, we find that using asymmetric models can improve model fit when the true covariance is asymmetric. 
In addition to symmetry, we also evaluate the misspecification of the covariance function.
We find that new models are substantially different than the Lagrangian model, where neither can perform as well when the other is the true model. 

These models are also evaluated in an analysis of the Irish wind data of \cite{haslett1989space}, which has been frequently used as a benchmark dataset \citep{gneiting2006geostatistical, ma2025asymmetric}. 
We consider five variations of the new asymmetric models with different marginal covariances. 
In general, the asymmetric models outperform the symmetric models in terms of model selection criteria. 
Depending on the covariance function used, the reflective asymmetric covariances can provide improved model fit and faster likelihood estimation than the Lagrangian model. 
We also find similar results on PM$_{2.5}$ data.

These developments are presented in the following.
In Section \ref{sec:background}, we introduce relevant background in more detail. 
Sections \ref{sec:spatial_cc} and \ref{sec:univariate_st} discuss asymmetries in multivariate spatial and univariate space-time covariances, respectively. 
In Section \ref{sec:practical}, we discuss computation, while Section \ref{sec:simulation} presents the simulations.
Section \ref{sec:data_analysis} presents the data analyses.  
We conclude in Section \ref{sec:discussion}.

\section{Background and literature}\label{sec:background}

The basic univariate spatial setting considers a random field $\{Y(\vecuse{s}), \vecuse{s}\in \mathbb{R}^d\}$ for $d \in \{1, 2, 3, \dots\}$. 
Throughout much of the exposition, we will assume second-order stationarity, where the mean and covariance do not depend on $\vecuse{s}$: $\mathbb{E}\{Y(\vecuse{s})\} = \mu$ and $\textrm{Cov}\left\{Y(\vecuse{s}), Y(\vecuse{s}+\vecuse{h})\right\} = \mathbb{E}\left\{(Y(\vecuse{s}) - \mu)(Y(\vecuse{s}+\vecuse{h}) - \mu)\right\} = C(\vecuse{h})$. Here $\vecuse{h}\in \mathbb{R}^d$ is a lag; in the case $d=1$, scalar operations (for example, $\lVert \vecuse{h}\rVert^2 = h^2$) may be used. 
Our primary interest is the covariance function $C(\vecuse{h})$, which requires $C(\vecuse{0})\geq 0$ and symmetry $C(\vecuse{h}) = C(-\vecuse{h})$, including further conditions for validity. 
We will assume that $\mu = 0$ and similarly for a multivariate mean below.
Throughout, covariance functions depend on parameters $\mb{\theta}$, but we suppress this notation for simplicity.

We will next provide background on multivariate and space-time models.

\subsection{Multivariate spatial models}

Let $\{\mb{Y}(\vecuse{s}) \in \mathbb{R}^p, \vecuse{s}\in \mathbb{R}^d\}$ be a multivariate random field which has a matrix-valued covariance function $
    \matusesigma{C}(\vecuse{h})=\mathbb{E}\left\{(\mb{Y}(\vecuse{s}) - \mb{\mu}) (\mb{Y}(\vecuse{s} + \vecuse{h}) - \mb{\mu})^\top \right\}  = \left[C_{jk}(\vecuse{h})\right]_{j,k=1}^p$, and $\mb{\mu} = \mathbb{E}\{\mb{Y}(\vecuse{s})\}$.
The diagonals of $\mb{C}(\vecuse{h})$ are the covariance functions of the processes, while the off-diagonal are cross-covariances \citep{genton15}.

\begin{theorem}[Simplified multivariate Bochner]\normalfont\label{thm:mult_boch}
The matrix-valued stationary covariance $\matusesigma{C}(h)$ is valid (positive-definite) if we have spectral representation \begin{align*}
    \matusesigma{C}(\vecuse{h}) &= \int_{\mathbb{R}^d}{\rm \textrm{exp}}\left(\I \langle \vecuse{h}, \vecuse{x}\rangle\right)\matusesigma{f}(\vecuse{x}) d\vecuse{x} = \left[\int_{\mathbb{R}^d}\textrm{exp}\left(\I \langle \vecuse{h}, \vecuse{x}\rangle\right)f_{jk}(\vecuse{x}) d\vecuse{x} \right]_{j,k=1}^p,
\end{align*}and the $p\times p$ matrix $\matusesigma{f}(\vecuse{x})$ is Hermitian (satisfying $f_{jk}(\vecuse{x}) = \overline{f_{jk}(-\vecuse{x})} = \overline{f_{kj}(\vecuse{x})}$, where $\overline{z}$ represents the conjugate of $z$) and positive-definite for all $\vecuse{x}$. 

\end{theorem}
See \cite{gneiting2010matern}, or \cite{yarger2023multivariate} for a more technical statement, where the careful use of spectral measures makes the statement if and only if.

Many of the cross-covariances proposed have been proportional to a covariance function. 
This has been especially considered for Mat\'ern spatial models \citep{gneiting2010matern, emery2022new}. 
At times, ensuring the validity of such multivariate models is nontrivial \citep{gneiting2010matern, emery2022new}. 
For a stationary covariance, we must have that \begin{align*}
    C_{jk}(\vecuse{h}) &= \textrm{Cov}(Y_{j}(\vecuse{s}), Y_k(\vecuse{s} + \vecuse{h})) = \textrm{Cov}(Y_k(\vecuse{s} + \vecuse{h}), Y_{j}(\vecuse{s})) =\textrm{Cov}(Y_k(\vecuse{s}^*), Y_{j}(\vecuse{s}^* - \vecuse{h})) \\
    &= C_{kj}(-\vecuse{h}),
\end{align*}where $\vecuse{s}^* = \vecuse{s} + \vecuse{h}$ momentarily. This implies that $\vecuse{C}(\vecuse{h}) = \vecuse{C}(-\vecuse{h})^\top$ is a requirement for a valid covariance. 
However, cross-covariances may be asymmetric, where $C_{jk}(\vecuse{h}) \neq C_{jk}(-\vecuse{h})$ for $j\neq k$. 
We define symmetry in the following way. 

\begin{definition}[Multivariate symmetry]\normalfont
A stationary multivariate covariance $\mb{C}(\vecuse{h})$ is symmetric if $
    \matusesigma{C}(\vecuse{h}) = \matusesigma{C}(-\vecuse{h})$ for all $\vecuse{h}\in \mathbb{R}^d$.
\end{definition}
In many cases, such a symmetric cross-covariance may be limiting and inflexible. 
Cross-covariances that are proportional to a covariance function are inherently symmetric. 
Most asymmetric cross-covariances $C_{jk}(\vecuse{h})$ are based on shifts, where $C_{jk}(\vecuse{h}) = \rho_{jk}C^*(\vecuse{h} - \vecuse{v}_{jk})$ for a covariance $C^*(\vecuse{h})$, a shift vector $\vecuse{v}_{jk} \in \mathbb{R}^d$, and a correlation parameter $\rho_{jk} \in [-1, 1]$ representing a correlation at lag $\mb{h} = \mb{v}_{jk}$ \citep{li_approach_2011, vu_modeling_2023, mu2024gaussian}.
This shift corresponds to a multiplication of $\textrm{exp}\{-\I \langle \vecuse{v}_{jk}, \vecuse{x}\rangle\}$ on the entry $f_{jk}(\vecuse{x})$ in the spectral domain. 
Recently, \cite{yarger2023multivariate} introduced other asymmetric multivariate covariances, which is the main topic of Section \ref{sec:spatial_cc}. 

\subsection{Univariate spatio-temporal models}

We now review some key definitions and covariance models in spatio-temporal statistics. 
We assume here that a covariance function $C(\vecuse{h},u)$ of a space-time process $\{Y(\vecuse{s},t) \in \mathbb{R}; \vecuse{s}\in \mathbb{R}^d, t \in \mathbb{R}\}$ is finite and covariance stationary, in that the covariance $
    C(\vecuse{h},u) = \textrm{Cov}\left\{Y(\vecuse{s},t), Y(\vecuse{s} + \vecuse{h}, t+ u)\right\} < \infty$ does not depend on $\vecuse{s}$ or $t$. 
We focus on a few properties of space-time covariances, where more information can be found in \cite{gneiting2006geostatistical, porcu202130, chen2021space}.

\begin{theorem}[Simplified Bochner, space-time covariances]\normalfont\label{thm:st_bochner}
Define a spectral density $f(\vecuse{x}, \eta)$ that is integrable and satisfies $f(-\vecuse{x}, -\eta) = f(\vecuse{x}, \eta)$ and nonnegativity $f(\vecuse{x}, \eta ) \geq 0$ for all $\vecuse{x}\in \mathbb{R}^d$ and $\eta \in \mathbb{R}$.
Then the following is a valid space-time covariance on $\mathbb{R}^d \times \mathbb{R}$: $
    C(\vecuse{h},u) = \int_{\mathbb{R}} \int_{\mathbb{R}^d} \textrm{exp}\left\{\I(\langle \vecuse{h},  \vecuse{x}\rangle + u\eta)\right\} f(\vecuse{x}, \eta) d\vecuse{x} d\eta$.
\end{theorem}
We propose asymmetric covariance models through the spectral density $f(\vecuse{x}, \eta)$.
While multivariate Bochner's theorem in Theorem \ref{thm:mult_boch} extends the output dimension using a matrix-valued spectral density rather than a univariate spectral density, this space-time Bochner result in Theorem \ref{thm:st_bochner} instead adds time and separates to these input dimensions into two parts, while $f(\vecuse{x}, \eta)$ remains scalar-valued. 
While $f(-\vecuse{x}, -\eta) = f(\vecuse{x}, \eta)$ is required, $f(-\vecuse{x}, \eta) = f(\vecuse{x}, \eta)$, for example, is not. 
One simplified construction of space-time models is through separability. 
\begin{definition}[Separability]\normalfont\label{def:separability}
    A space-time stationary covariance $C(\vecuse{h},u)$ is called separable if a decomposition $C(\vecuse{h},u) = C_{space}(\vecuse{h}) C_{time}(u)$ exists for all $\vecuse{h}\in \mathbb{R}^d$ and $u \in \mathbb{R}$. 
\end{definition}

In general, separability is often considered to be restrictive \citep{chen2021space,de_iaco_positive_2013}. 
In particular, there is only limited multiplicative interaction between the spatial and temporal dimensions. 
Such a structure is often violated in data applications \citep{gneiting2006geostatistical}.
A common nonseparability measure is  $R(\vecuse{h},u) = C(\vecuse{h},u) C(\vecuse{0},0) - C(\vecuse{h},0)C(\vecuse{0},u)$, and when Definition \ref{def:separability} holds, for all $\vecuse{h}$ and $u$, we have, we have: 
\begin{align*} R(\vecuse{h},u) &= C_{space}(\vecuse{h}) C_{time}(u)C_{space}(\vecuse{0}) C_{time}(0)- C_{space}(\vecuse{h}) C_{time}(0)C_{space}(\vecuse{0}) C_{time}(u) \\
    &=0.
\end{align*}
Positive (negative) nonseparability occurs if $R(\vecuse{h},u) > 0$ ($R(\vecuse{h},u) < 0$) for some $\vecuse{h}$ and $u$ \citep{de_iaco_positive_2013}.  
Previous literature \cite{de_iaco_positive_2013} finds positive nonseparability or a mix of positive and negative nonnonseparability more common in established models. 
For positive nonseparability, the covariance for $\vecuse{h} \neq \vecuse{0}$ and $u \neq 0$ is larger than the covariance of a corresponding separable model. 
As a separable model may be used to initially model data, an empirical nonseparability measure provides a natural next step for data analysis and model choice if a nonseparable model is desired \citep{de_iaco_positive_2013}. 
An alternate but related nonseparability measure is $C(\vecuse{h},u)/ \{C(\vecuse{0}, u)C(\vecuse{h}, 0)\}$. 
We refer to $C(\vecuse{h},0)$ and $C(\vecuse{0}, u)$ as the marginal spatial and temporal covariances, respectively. 

We now introduce a celebrated nonseparable and symmetric class of space-time covariances called the Gneiting class \citep{gneiting_nonseparable_2002}. 
We use the form of the ``extended'' class \citep{allard2025modeling} with a separability parameter $b\in[0,1]$, where $b=0$ corresponds to separability.
With the additional parameter $\delta > 0$, consider
\begin{align}
    C(\vecuse{h},u) &= \frac{\sigma}{\left(\gamma(u)+1\right)^{\frac{b}{2} + \delta}}\varphi\left\{\frac{\lVert \vecuse{h} \rVert^2}{\left(\gamma(u) + 1\right)^b}\right\},\label{eq:gne_general}
\end{align}where $\gamma(u)$ is a continuous variogram on $\mathbb{R}$ and $\varphi(r)$ is a completely monotone function on $[0, \infty)$. 
Gneiting covariances only introduce positive nonseparability \citep{de_iaco_positive_2013}.
We next discuss symmetry and asymmetry in space-time covariances, the principal characteristic motivating this paper. 
\begin{definition}[Full symmetry]\normalfont
    A space-time covariance $C(\vecuse{h},u)$ is fully symmetric, if, for all $\vecuse{h} \in \mathbb{R}^d$ and $u \in \mathbb{R}$, $
        C(\vecuse{h},u) = C(-\vecuse{h}, u) = C(\vecuse{h}, -u) = C(-\vecuse{h},-u)$.
\end{definition}

If a covariance is not fully symmetric, then we say it is asymmetric, though $C(\vecuse{h},u)= C(-\vecuse{h},-u)$ must be satisfied.
Intuitively, the covariance between the process at a fixed location and time will have increased dependence in one direction compared to others at a later time. 
Thus, if a process is spatio-temporally evolving through space and time in a particular direction, asymmetry will be present. 
Symmetric space-time covariances $C(\vecuse{h}, u)$ have maximum value for fixed $u$ at $\vecuse{h} = \vecuse{0}$ and vice-versa, so such dynamic processes are not well-represented. 
While multivariate spatial asymmetry manifests through dependence between processes, this space-time asymmetry manifests through dependence between the same process as different times. 
In some senses, one can consider a space-time process observed at two time points $t_1$ and $t_2$ to be analogous to the multivariate case with $p = 2$.
All covariances that are separable are also fully symmetric \citep{gneiting2006geostatistical}.
\cite{park2008testing} propose statistical tests for symmetry in space-time data and define  axial symmetry as below. 

\begin{definition}[Axial symmetry]\normalfont\label{def:axial}
    A space-time covariance is axially symmetric in space with respect to the $k$-th dimension if $
        C(\vecuse{h}, u) = C(\vecuse{h}^*, u)$ for any pairs $(\vecuse{h}, \vecuse{h}^*)$ whose entries satisfy $h_j = h_j^*$ for all $j\neq k$ and $h_k = -h_k^*$.
\end{definition}

Axial symmetry gives that symmetry holds in the $k$-th spatial dimension while holding others fixed. 
Thus, an asymmetric covariance may be asymmetric with respect to one spatial dimension only, while it is axial symmetric with respect to others. 
Alternately, when the asymmetry manifests through all dimensions of $\vecuse{h}$, axial symmetry will not hold for any of the dimensions. 
Such interpretations can be helpful, for example, when the asymmetry is directly East-West, while the space-time process is axially symmetric with respect to the North-South dimension.

The most-commonly referred-to class of asymmetric space-time covariances are Lagrangian covariances \citep{gneiting2006geostatistical, salvana2023spatio, porcu202130, chen2021space, fonseca2011general, ezzat2018spatio, park2006new}. 
One takes 
    $C(\vecuse{h},u) = \mathbb{E}_\vecuse{V}\left\{C_{space}(\vecuse{h} - u\vecuse{V})\right\}$, where $C_{space}(\vecuse{h})$ is a purely spatial covariance function and $\vecuse{V}$ is a random vector in $\mathbb{R}^d$. 
The expectation can be simplified when using the squared-exponential covariance $C_{space}(\vecuse{h}) =\sigma \textrm{exp}(-a_{\vecuse{s}}^2\lVert \vecuse{h}\rVert^2)$, where $\sigma > 0$ is a variance parameter and $a_{\vecuse{s}} > 0$ is an inverse range parameter. 
If $\vecuse{V} \sim \mathcal{N}(\vecuse{\mu}_{\vecuse{V}}, \matuse{\Sigma}_{\vecuse{V}})$, where $\mathcal{N}(\cdot, \cdot)$ is the multivariate normal distribution, the Lagrangian covariance is \begin{align}
C(\vecuse{h},u) &= \sigma \frac{ \textrm{exp}\left\{-a_{\vecuse{s}}^2\left(\vecuse{h} - u\vecuse{\mu}_{\vecuse{V}}\right)^\top\left(I_d + 2a_{\vecuse{s}}^2u^2\matuse{\Sigma}_{\vecuse{V}}\right)^{-1} \left(\vecuse{h} - u\vecuse{\mu}_{\vecuse{V}}\right)\right\}}{\left|I_d + 2a_{\vecuse{s}}^2u^2\matuse{\Sigma}_{\vecuse{V}}\right|^\frac{1}{2}},\label{eq:lagrangian}
\end{align}where $I_d$ is the $d\times d$ identity matrix and $|\matuse{\Sigma}|$ is the determinant of $\matuse{\Sigma}$  \citep{schlather2010some}. 
The covariance is asymmetric if $\vecuse{\mu}_{\vecuse{V}} \neq \vecuse{0}$ \citep{salvana_daouia_lagrangian_2021}. 
Similar representations exist for other scale-mixture covariances \citep{noauthor_corrections_2024, ma2025asymmetric}. 
The marginal covariances are $C(\vecuse{h}, 0) = \sigma\textrm{exp}\left(-a_{\vecuse{s}}^2 \lVert \vecuse{h}\rVert^2\right)$ and $C(\vecuse{0}, u) = \sigma \left|I_d + 2a_{\vecuse{s}}^2u^2\matuse{\Sigma}_{\vecuse{V}}\right|^{-1/2}\textrm{exp}\left\{a_{\vecuse{s}}^2u^2\vecuse{\mu}_{\vecuse{V}}^\top \left(I_d + 2a_{\vecuse{s}}^2u^2\matuse{\Sigma}_{\vecuse{V}}\right)^{-1} \vecuse{\mu}_{\vecuse{V}}\right\},$ with the temporal marginal covariance affected by the spatial covariance function. 
Asymptotic study of parameter estimation and prediction results typically requires appropriate control of both marginal spectral densities \citep[Section 5.2,][]{faouzi_spacetime_2025}. 
The only separable Lagrangian covariance takes $\vecuse{V}=\vecuse{0}$, a trivial case where $C(\vecuse{0},u)$ is constant in $u$.
The Lagrangian model has $1 + 1 + d + d(d+1)/2$ parameters through $(\sigma, a_{\vecuse{s}}, \vecuse{\mu}_{\vecuse{V}}, \matuse{\Sigma}_\vecuse{V})$, respectively. 
Using a convex combination of Lagrangian and Gneiting covariances has also been considered \citep{gneiting2006geostatistical, chen2021space}. 
However, there is not much harmony between them, as often the Lagrangian covariance is compact unlike the Gneiting covariance.

In Fig.~\ref{fig:cov_comparison}, we give some example space-time covariances. 
The separable model has the spatial covariance take the same shape as $u$ moves away from $0$. 
The Gneiting covariance has a similar form, but the spatial range increases more as $u$ moves away from $0$, indicating space-time interaction. 
Both the separable and Gneiting covariances are symmetric. 
The reflective asymmetric covariances has negative correlation for some spatial and temporal lags and indicate asymmetry as $C(\vecuse{h}, u)$ is not necessarily equal to $ C(-\vecuse{h}, u)$ for $u\neq 0$. The Lagrangian covariances are also asymmetric, with areas of high covariance moving in a particular direction as $u$ increases. 
The reflective asymmetric 1 and Lagrangian 1 covariances are axially symmetric with respect to the second spatial dimension $h_2$, as $C((h_1, h_2)^\top, u) = C((h_1, -h_2)^\top, u)$ for these covariances throughout. 
When $\vecuse{V}$ is nonrandom as in Lagrangian 1, the strength of dependence does not decrease as $u$ moves away from $0$, which is sometimes referred to as the frozen-field model \citep{salvana2023spatio}. 
More plausibly, the strength of dependence decreases over time as in Lagrangian 2. 

\begin{figure}
    \centering
    \includegraphics[width=0.98\linewidth]{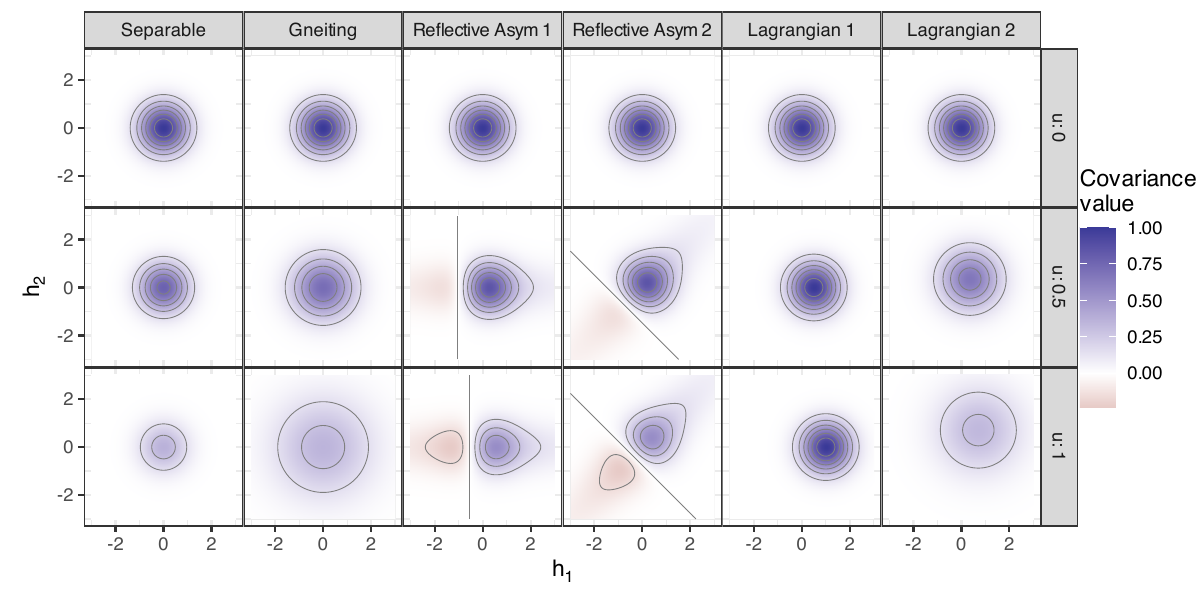}
    \caption{Covariance function examples. All covariances have marginal covariance $C(\vecuse{h},u) = \exp(-\lVert \vecuse{h}\rVert^2) \exp(-u^2)$. The Gneiting covariance takes $C(\vecuse{h}, u) = \exp(-\lVert \vecuse{h}\rVert^2 / \{1 + u^2\}^2)/ (1 + u^2)^{3/2}$. The reflective asymmetric are defined in Section \ref{sec:univariate_st}, with $\xi = 0.9$ and $\tilde{\vecuse{x}} = (1,0)^\top$ for 1 and $\tilde{\vecuse{x}} = (1,1)^\top /\sqrt{2}$ for 2. 
    The Lagrangian covariances take $\vecuse{V} = (1,0)^\top$ for 1 and $\vecuse{V} \sim \mathcal{N}(\{1,1\}^\top, I_2)$ for $2$. }
    \label{fig:cov_comparison}
\end{figure}

Other asymmetric space-time covariances have also been introduced. 
The approach in Section 5 of \cite{stein2005space} is most similar to ours by using the following spectral density of the form $
    f(\vecuse{x}, \eta) = f^*(\lVert \vecuse{x} \rVert, |\eta|) \left(1 + \xi\langle \vecuse{x}, \tilde{\vecuse{x}}\rangle \eta + c_1\lVert \vecuse{x}\rVert^2 + c_2\eta^2\right)$, where $f^*(\lVert \vecuse{x}\rVert, |\eta|)$ is the spectral density of a symmetric space-time covariance, $\tilde{\vecuse{x}} \in \mathbb{R}^d$, and $\xi^2 < 4c_1c_2$. 
However, additional requirements on $f^*(\lVert \vecuse{x} \rVert, |\eta|)$ are needed, and there are a number of new parameters $\xi$, $c_1$, $c_2$, and $\tilde{\vecuse{x}}$.
\cite{jin_bag_2024} use directed acyclic graphs to introduce flexible directional asymmetry with relatively high model complexity. 
\cite{horrell2017half} define a half-spectral approach, where one integral in Bochner's representation has been evaluated. 
Finally, \cite{wu2021non} and \cite{zhang_non-fully_nodate} describe an approach for asymmetry through scale-mixtures, though covariances are only available through infinite series. 
While the Lagrangian covariances maintain positive correlation, reflective asymmetric covariances tend to induce higher positive correlation in one direction and negative correlation in the opposite direction.

\section{Reflective asymmetric cross-covariances}\label{sec:spatial_cc}

We next provide background and new results on asymmetric cross-covariances.
In \cite{yarger2023multivariate}, cross-covariance functions were introduced of the form \begin{align*}
    C_{jk}(\vecuse{h}) &= \int_{\mathbb{R}^d}\textrm{exp}\left(\I \langle \vecuse{h}, \vecuse{x}\rangle\right)\left\{\Re(\sigma_{jk}) - \I \textrm{sign}(\langle \vecuse{x}, \tilde{\vecuse{x}}\rangle)\Im(\sigma_{jk})\right\} \sqrt{f_j(\vecuse{x}) f_k(\vecuse{x})} d\vecuse{x},
\end{align*}where $\{f_j(\vecuse{x})\}$ are the individual spectral densities of each process, $\Re(z)$ and $\Im(z)$ the real and imaginary parts of $z$, respectively, $\matusesigma{\sigma} = [\sigma_{jk}]_{j,k=1}^p$ is a positive definite, Hermitian matrix with potentially complex entries on the off-diagonal, and $\tilde{\vecuse{x}}$ is a unit vector of length $d$. 
We will also use $\Re$ and $\Im$ as superscripts to denote functions being symmetric or asymmetric. 
For $d=1$, we generally take $\tilde{\vecuse{x}} = 1$, and we can also represent the cross-covariance $
    C_{jk}(h) = 
    2\int_{0}^\infty \left\{\cos(hx) \Re(\sigma_{jk}) + \sin(hx) \Im(\sigma_{jk})\right\}\sqrt{f_j(x) f_k(x)} dx$.
For general $d$, this leads to \begin{align}
    C_{jk}(\vecuse{h}) = \Re(\sigma_{jk})C_{jk}^\Re(\vecuse{h}) + \Im(\sigma_{jk})C_{jk}^\Im(\vecuse{h}),\label{eq:spatial_asym_decomp}
\end{align}where $ C_{jk}^\Re(\vecuse{h})$ is symmetric, and $C_{jk}^\Im(\vecuse{h})$ introduces asymmetry when $\Im(\sigma_{jk}) \neq 0$. In particular, as mentioned later, $C_{jk}^\Im(\vecuse{h})$ will be an odd function satisfying $C_{jk}^{\Im}(\vecuse{h})=-C_{jk}^{\Im}(-\vecuse{h})$, and thus provides a natural pair to the even function $C_{jk}^{\Re}$.

The choice of $-\I \textrm{sign}(\langle \mb{x}, \tilde{\mb{x}}\rangle)$ has a couple of motivations. 
First, we must satisfy the Hermitian property of Theorem \ref{thm:mult_boch} with $-\I\textrm{sign}(\langle (-\mb{x}), \tilde{\mb{x}})\rangle) = +\I\textrm{sign}(\langle \mb{x}, \tilde{\mb{x}})\rangle) = \overline{-\I\textrm{sign}(\langle \mb{x}, \tilde{\mb{x}})\rangle)}$. 
Other potential odd functions could be considered like $-\I \langle\mb{x}, \tilde{\mb{x}}\rangle$, $-\I \langle\mb{x}$, $\tilde{\mb{x}}\rangle / \lVert \mb{x}\rVert$, or $-\I\sin(\langle \mb{x}, \tilde{\mb{x}}\rangle)$. 
However, using $-\I \langle\mb{x}, \tilde{\mb{x}}\rangle$ requires additional validity requirements on the spectral density to ensure integrability, finding closed form cross-covariances may be more challenging, and the sign function is, in some senses, the closest odd function to a constant.
This construction was originally motivated through forms of operator fractional Brownian motion \citep{didier_integral_2011}, which has a similar indicator-based form to characterize asymmetry.
The directional parameter $\tilde{\vecuse{x}}$ should be interpreted as the primary direction of asymmetry. 
In particular, when $\Im(\sigma_{jk}) >0$, one will obtain $C_{jk}^{\Im}(\tilde{\vecuse{x
}}) > C_{jk}^{\Im}(-\tilde{\vecuse{x
}})$ and thus $C_{jk}(\tilde{\vecuse{x
}}) > C_{jk}(-\tilde{\vecuse{x
}})$, so that processes $j$ and $k$ have increased cross-covariance in direction $\tilde{\vecuse{x
}}$ compared to $-\tilde{\vecuse{x
}}$. Thus, this can be an appropriate model when one process is more positively correlated with a shifted version (in direction $\tilde{\vecuse{x}}$) of another. 

When $d=1$, $C_{jk}^{\Im}(h)$ will be the Hilbert transform of $C_{jk}^{\Re}(h)$ \citep{king2009hilbert}, which corresponds to the Fourier multiplier $-\I\textrm{sign}(x)$. 
Thus, various results and properties of the Hilbert transform can be employed to derive cross-covariance functions and their characteristics. For example,
if $C_{jk}^{\Re}(h)$ is an even function, then $C_{jk}^{\Im}(h)$ is odd. 
Many closed forms for $C_{jk}^{\Im}(h)$ can be found in \cite{king2009hilbert}. 
For $d >1$, the Fourier multiplier $-\I\textrm{sign}(\langle \vecuse{x}, \tilde{\vecuse{x}}\rangle)$ has been mentioned in some forms as directional or partial Hilbert transforms (Volume 2, 61 of \citealt{king2009hilbert}; 487, 571 of \citealt{garcia1985weighted}; 172 of \citealt{granlund2013signal}; \citealt{felsberg2002monogenic}). 
While the sign function has a discontinuity at $0$, this point represents a set of measure $0$ in the spectral integral.

Like the Lagrangian space-time  covariance, there are possible physical interpretations of this asymmetry. 
In the foundations of signal processing \citep{gabor1946theory}, a real-valued temporal signal $F(t)$ can be transformed into a complex ``analytic signal,'' with real part $F(t)$ and imaginary part its Hilbert transform.
Consider an example of a simple harmonic oscillator whose position at time $t$ is defined by $A\cos(\omega t)$ for an amplitude $A > 0$ and frequency $\omega$. 
Its analytic signal through the Hilbert transform is $A\cos(\omega t) + \I A\sin(\omega t)$ which represents it physical position in the real part and its (negative) velocity in the imaginary part. 
Meanwhile, its instantaneous amplitude is constant at $A$, representing conservation of energy as the sum of squares of its potential energy (related to the real part) and kinetic energy (related to the imaginary part). 
A similar idea appears in optics, where a function and its Hilbert transform can together describe (a) how much light bends when passing through a material (b) how much energy is absorbed by the material \citep{lucarini2005kramers}.

We likewise define the instantaneous amplitude $\sqrt{\Re(\sigma_{jk})^2C_{jk}^{\Re}(h)^2 + \Im(\sigma_{jk})^2C_{jk}^{\Im}(h)^2}$ that can represent dependence near $h$ which can be larger even if $\Re(\sigma_{jk})^2C_{jk}^{\Re}(h)^2$ is small. 
The instantaneous phase $\tan^{-1}[\Im(\sigma_{jk})C_{jk}^{\Im}(h)/ \{\Re(\sigma_{jk})C_{jk}^{\Re}(h)\}]$ gives an angular interpretation between the two components. 
These represent an overall measure of dependence at lag $h$ and how much of that dependence is symmetric or asymmetric, respectively. 
We regard $C_{jk}^{\Re}(h)$ and $C_{jk}^{\Im}(h)$ as an opposite and complementary pair that together describe cross-covariance with flexible symmetry or asymmetry.  





We next concretely present $C_{jk}^{\Im}(\vecuse{h})$ for some covariances. 
We delineate which are new and which are presented in \cite{yarger2023multivariate}. 
In Section \ref{sec:sqexp}, the form in Proposition \ref{prop:sq_exp_spat} was recognized in the $d=1$ case only in \cite{yarger2023multivariate}, with Proposition \ref{prop:sq_exp_spat} extending that expression to the general $d$ case. 
In addition, all Mat\'ern forms in Section 
\ref{sec:matern_cov} were established in \cite{yarger2023multivariate}. All other forms are newly established: Propositions \ref{prop:sq_exp_spat}, \ref{prop:schoen_asymmetric}, \ref{prop:asymm_cauchy_spatial}, and associated examples and remarks. 
The most impactful advance over \cite{yarger2023multivariate} is the first establishment of such models for $d > 1$. 
In the following forms, we implicitly assume that $f_{j}(\vecuse{x}) = f_k(\vecuse{x})$, but the full versions are given in Appendix \ref{app:vary_params}. 
Throughout, proofs for propositions are also presented in Appendix \ref{app:proofs}.

\subsection{Squared-exponential covariance}\label{sec:sqexp}

As mentioned in Section 7.6 of \cite{yarger2023multivariate}, in $d=1$ if $C_{jk}^{\Re}(h) = \textrm{exp}\left(-a^2h^2 \right)$ where $a > 0$, then $
    C_{jk}^{\Im}(h) = \left(2/\sqrt{\pi}\right)
    D_+\left(a h\right)$, where $D_+(z)$ is Dawson's function defined as the following: $
    D_+(z) = 
\left(\sqrt{\pi}/2\right)\textrm{exp}\left(-z^2\right) \textrm{erfi}(z)$, and $\textrm{erfi}(z)$ is the imaginary error function. 
See Section 7 (7.2.5, 7.5.1) of \cite{NIST:DLMF}. 
Dawson's function is related to the one-sided sine transform of the squared-exponential function: $(2/\sqrt{\pi})D_+(z) = (1/\sqrt{\pi}) \int_0^\infty \sin(zt) \exp(-t^2/4) dt$, forming a natural pair with $\exp(-z^2) = (1/\sqrt{\pi})\int_0^\infty \cos(zt) \exp(-t^2/4) dt$ that are used for the odd and even parts of the cross-covariance, respectively. 
This result is also Hilbert transform numbered (3.5) in \cite{king2009hilbert}. 
We next introduce the novel extended result for general $d$, which thus increases applicability to a much wider variety of spatial processes in two or three dimensions and beyond. 
Furthermore, since the squared-exponential covariance plays a key role in Schoenberg's theorem discussed below, this establishment will extend such asymmetric cross-covariances to wider varieties of covariance functions. 
\begin{proposition}\normalfont\label{prop:sq_exp_spat}
Consider the cross-covariance $C_{jk}^{\Re}(\vecuse{h}) = \textrm{exp}\left(-a^2 \lVert \vecuse{h}\rVert^2\right)$ where $\vecuse{h} \in \mathbb{R}^d$. The asymmetric part corresponding to the $-\I\textrm{sign}(\langle \vecuse{x}, \tilde{\vecuse{x}}\rangle)$ Fourier multiplier is $
    C^{\Im}_{jk}(\vecuse{h}) = \textrm{exp}\left(-a^2\lVert \vecuse{h}\rVert^2\right)  \textrm{erfi}\left(a\langle \vecuse{h}, \tilde{\vecuse{x}}\rangle\right)$.
\end{proposition}

\begin{figure}[t]
    \centering
    \includegraphics[width=0.95\textwidth]{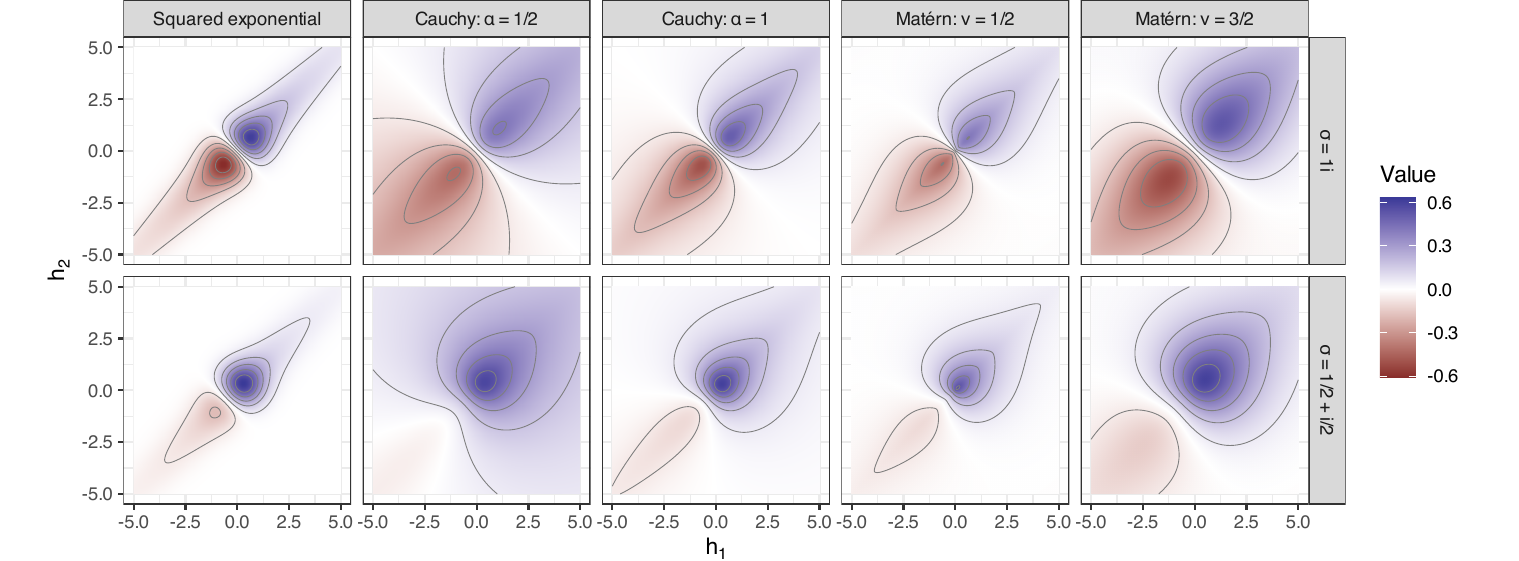}
    \caption{Squared-exponential, Cauchy, and Mat\'ern cross-covariances $\Re(\sigma_{jk})C_{jk}^{\Re}(\vecuse{h}) + \Im(\sigma_{jk})C_{jk}^{\Im}(\vecuse{h})$ in $d=2$. In all cases, we take $a=1$, $\tilde{\vecuse{x}} = (1,1)^\top / \sqrt{2}$. On the top row, we plot $C_{jk}^{\Im}(\vecuse{h})$, while the bottom row takes $\Re(\sigma_{jk}) = \Im(\sigma_{jk}) = 1/2$.}
    \label{fig:spatial_comparison}
\end{figure}

In Fig.~\ref{fig:spatial_comparison}, we plot an example of $C_{jk}^{\Im}(\vecuse{h})$ for $d=2$, and in Fig.~\ref{fig:1ds}, we plot an examples of $\Re(\sigma_{jk})C_{jk}^{\Re}(\vecuse{h}) +\Im(\sigma_{jk})C_{jk}^{\Im}(\vecuse{h})$ for $d=1$ with varying amounts of asymmetry. The form still includes the squared-exponential form, but multiplies by the additional term introducing the asymmetry. We next discuss the tail behavior of this form.

\begin{remark}[Tail behavior, squared-exponential]
    We detail the cross-covariance's tail behavior, using $f(u) \overset{u\to \infty}{\sim} g(u)$ to mean $\lim_{u\to \infty} f(u)/g(u) = 1$. Consider a fixed direction defined by $\vecuse{h}$, so that $\langle \vecuse{h}, \tilde{\vecuse{x}} \rangle = \lVert \vecuse{h} \rVert \cos(\gamma)$ for an angle $\gamma$ between $\vecuse{h}$ and $\tilde{\vecuse{x}}$, 
    and consider the asymptotic form as $\lVert \vecuse{h}\rVert$ increases.
    When $\cos(\gamma)=0$ (for example, $\gamma = \pi/2$), $C_{jk}^{\Im}(\vecuse{h}) = 0$ also as $\textrm{erfi}(0) = 0$. 
    Otherwise, when $\cos(\gamma) \neq0$, the tail behavior as $\lVert \vecuse{h}\rVert \to \infty$ is \citep[e.g.][]{erfi} \begin{align*}
         C^{\Im}_{jk}(\vecuse{h}) 
         &\overset{\lVert \vecuse{h}\rVert \to \infty}\sim \textrm{exp}\left[-a^2\lVert \vecuse{h}\rVert^2 \left\{1 - \cos(\gamma)^2\right\}\right] \frac{1}{\sqrt{\pi} a\lVert \vecuse{h}\rVert \cos(\gamma)}.
    \end{align*}In the direction of asymmetry where $\cos(\gamma)^2 = 1$, $C_{jk}^{\Im}(\vecuse{h}) \sim (\sqrt{\pi} a\lVert \vecuse{h}\rVert)^{-1} $ as $n\to \infty$. For any other direction, however, the exponential term dominates with faster tail decay. 
    The slow tail decay when $\cos(\gamma)^2 = 1$ is perhaps surprising. However, this is a general property of the Hilbert transform, as Schwartz class functions with nonzero integral (which includes $h(z) = \exp(-z^2)$) have Hilbert transform tail decay of this magnitude \cite{tao2004lecture2}. See also Chapter 8 of \cite{king2009hilbert} which examines the asymptotic expansion of the Hilbert transform. 
    This tail behavior is also apparent in Fig.~\ref{fig:1ds}, with heavier tail decay whenever $\Im(\sigma_{jk}) \neq 0$.
\end{remark}

The asymmetric representation for the squared-exponential covariance provides a basis to find $C_{jk}^{\Im}(\vecuse{h})$ for other covariance functions through Schoenberg's theorem \citep{allard2025modeling}. 
Schoenberg's theorem establishes valid covariance functions in \emph{any} dimension $d$ that are thus generally applicable, and many commonly-used covariance functions can be represented through it. 
Equivalently, it establishes which radial functions are completely monotone \citep{gneiting_nonseparable_2002}. 
We outline its idea here. 
Suppose $
    C(\vecuse{h}) = \sigma \int_0^\infty \textrm{exp}\left(-v\lVert \vecuse{h}\rVert^2\right) g(v) dv$ where $g(v)$ is a nonnegative mixing density.
Then $C(\vecuse{h})$ is a valid covariance function for any $d$, and $C(\vecuse{h})$ has a normal scale-mixture representation \citep{allard2025modeling}. 
One may obtain new covariances by proposing appropriate $g(v)$. 
We next consider an asymmetric version.

\begin{proposition}\normalfont\label{prop:schoen_asymmetric}
        Suppose that one can represent $
            C_{jk}^{\Re}(\vecuse{h})=\int_0^\infty \textrm{exp}\left(- v \lVert \vecuse{h}\rVert^2\right) g(v) dv
       $ for $\vecuse{h} \in \mathbb{R}^d$ and $g(v)$ is a nonnegative integrable density so that $\int_0^\infty g(v) dv < \infty$.
    Suppose also that we are interested in asymmetry along $\tilde{\vecuse{x}}$ and $\int_{\langle \vecuse{h}, \tilde{\vecuse{x}}\rangle \in \pm \lVert\vecuse{h}\rVert} |C_{jk}^{\Re}(\vecuse{h})|^p d\lVert \vecuse{h} \rVert < \infty$ for some $1 < p < \infty$, representing a one-dimensional integral along the direction of asymmetry $\tilde{\vecuse{x}}$.
       Then a corresponding asymmetric part corresponding to a 1-dimensional Hilbert transform with respect to $\tilde{\vecuse{x}}$ can be represented $
            C_{jk}^{\Im}(\vecuse{h}) = \int_0^\infty \textrm{exp}\left(-v\lVert \vecuse{h}\rVert^2 \right)\textrm{erfi}\left(\sqrt{v}\langle \vecuse{h}, \tilde{\vecuse{x}}\rangle\right) g(v) dv$. 
            If the spectral density of $C_{jk}^{\Re}(\vecuse{h}
            )$ exists, this tranformation corresponds to the Fourier multiplier $-\I\textrm{sign}(\langle \vecuse{x}, \tilde{\vecuse{x}}\rangle)$. 
\end{proposition}

This generally follows by exchanging the order of integration. 
In general, the one-dimensional Hilbert transform of $x(h)$ exists if $\int_\mathbb{R} |x(h)|^p dh < \infty $ for some $1 < p < \infty$ \citep{king2009hilbert}.
We will expect essentially all Schoenberg covariance functions of interest to satisfy the condition, as generally $|C_{jk}^{\Re}(\vecuse{h})| \leq 1$ and any function with polynomially-decaying (or faster-decaying) tails will satisfy this.

From this result, asymmetric forms for a large class of multivariate covariances may be represented in this way, offering a path for establishing $C_{jk}^{\Im}(\vecuse{h})$ for any dimension $d$ for a variety of covariance functions.
This gives three ways to find a $C_{jk}^{\Im}(\vecuse{h})$ in closed form: through the Fourier transform with multiplier $-\I\textrm{sign}(\langle \vecuse{x}, \tilde{\vecuse{x}}\rangle)$, through the scale-mixture representation $g(v)$ in Proposition \ref{prop:schoen_asymmetric}, and through the Hilbert transform of $C_{jk}^{\Re}(\vecuse{h})$.

\subsection{Cauchy covariance}

We next apply Proposition \ref{prop:schoen_asymmetric} to the Cauchy covariance, defined as $
    C(\vecuse{h}) = \sigma\left(1 + a^2\lVert \vecuse{h}\rVert^2\right)^{-\alpha}$ for $\alpha > 0$ and $a >0$. 
This covariance has been useful to model heavier tail decay \citep{lim_gaussian_2009, chiles2012geostatistics}. 
We present the general complicated form, followed by simpler special cases, with some plotted in Figs.~\ref{fig:spatial_comparison} and \ref{fig:1ds}.

    \begin{proposition}\normalfont\label{prop:asymm_cauchy_spatial}
    The function 
        $C_{jk}^{\Re}(\vecuse{h})= \left(1 + a^2 \lVert \vecuse{h}\rVert^2\right)^{-\alpha}$
 has an asymmetric part corresponding to directional Hilbert transform with respect to $\tilde{\vecuse{x}}$ is \begin{align*}
        C_{jk}^{\Im}(\vecuse{h})&= 
        \frac{1}{(a^2\lVert \vecuse{h}\rVert^2 + 1)^\alpha} \frac{\Gamma(\alpha + \frac{1}{2})}{\sqrt{\pi}\Gamma(\alpha)}\frac{2a\langle \vecuse{h}, \tilde{\vecuse{x}}\rangle}{(a^2\lVert \vecuse{h}\rVert^2 + 1)^{ \frac{1}{2}}}  {}_2F_1\left(\frac{1}{2}, \frac{1}{2} + \alpha; \frac{3}{2}; \frac{a^2\langle \vecuse{h}, \tilde{\vecuse{x}}\rangle^2}{a^2\lVert \vecuse{h}\rVert^2 + 1}\right),
    \end{align*}
    where ${}_2F_1(a, b; c;z)$ is the hypergeometric function \citep{abramowitz1948handbook}. 
    \end{proposition}

This relatively complicated form introduces asymmetry in cross-covariances for covariances with heavy tail decay, which is controllable in a range through the parameter $\alpha$, even though the spectral density of $C_{jk}^{\Re}(\vecuse{h})$ does not exist for all $\alpha >0$ and $d$ \citep{lim_gaussian_2009}.
The hypergeometric function, while a complicated special function in full generality, reduces to simpler cases in several instances which follow.

\begin{example}[Special case $d=1$]\normalfont
    When $d=1$, from 15.8.1	of \cite{NIST:DLMF} we have: $
    C_{jk}^\Im(h)=[2ah\Gamma(\alpha + 1/2)/\{\sqrt{\pi}\Gamma(\alpha)\}] {}_2F_1\left(\alpha + 1/2, 1; 3/2; -a^2h^2\right)$, which corresponds to the Hilbert transform (2.56) of \cite{king2009hilbert}.
\end{example}

\begin{example}[Special case $\alpha = 1$]\normalfont\label{ex:cauchy_spatial_a1}
Consider $\alpha = 1$ where $C_{jk}^{\Re}(\vecuse{h}) = \left(a^2 \lVert \vecuse{h}\rVert^2 + 1\right)^{-1}$. Applying 15.1.8 of \cite{abramowitz1948handbook} results in \begin{align*}
    C_{jk}^{\Im}(\vecuse{h})
    &=\frac{1}{a^2\lVert \vecuse{h}\rVert^2 + 1}\frac{a\langle \vecuse{h}, \tilde{\vecuse{x}}\rangle}{\sqrt{a^2(\lVert \vecuse{h}\rVert^2 - \langle \vecuse{h}, \tilde{\vecuse{x}}\rangle^2) + 1}}.  
        \end{align*}
        When $d=1$, this matches with the Hilbert transform entry (2.4) in \cite{king2009hilbert}. 

        \end{example}

When $d=1$ or $\vecuse{h}$ points in the same direction of $\tilde{\vecuse{x}}$ so that $\langle \vecuse{h}, \tilde{\vecuse{x}}\rangle \in \{\pm \lVert \vecuse{h}\rVert\}$ in Example \ref{ex:cauchy_spatial_a1}, the asymmetric part is simply a linear adjustment to the even part $C_{jk}^{\Re}(\vecuse{h})= (a^2\lVert \vecuse{h}\rVert^2 + 1)^{-1}$: $C_{jk}^{\Im}(\vecuse{h})
    =C_{jk}^{\Re}(\vecuse{h}) a\lVert \vecuse{h}\rVert \textrm{sign}(\langle \vecuse{h}, \tilde{\vecuse{x}}\rangle) $. 

\begin{example}[Special case $\alpha = 1/2$]\normalfont\label{ex:cauchy_spatial_a12}

Consider $\alpha = 1/2$ where $C_{jk}^{\Re}(\vecuse{h}) = \left(a^2\lVert \vecuse{h}\rVert^2 + 1\right)^{-1/2}$. Then, 15.1.4 and 4.6.22 of \cite{abramowitz1948handbook} imply \begin{align*}
    C_{jk}^{\Im}(\vecuse{h})
    &= \frac{1}{\sqrt{a^2\lVert \vecuse{h}\rVert^2 + 1}} \frac{2}{\pi}\tanh^{-1}\left(\frac{a\langle \vecuse{h}, \tilde{\vecuse{x}}\rangle}{\sqrt{a^2\lVert \vecuse{h}\rVert^2 + 1}}\right).
\end{align*}
When $d=1$, one has $    C_{jk}^{\Im}(h)= (a^2h^2 + 1)^{-1/2} (2/\pi)\sinh^{-1}\left(ah\right),$ 
again matching Hilbert transform entry (2.54) in \cite{king2009hilbert}.  

\end{example}

Example \ref{ex:cauchy_spatial_a12} provides an asymmetric part that has heavier tails than when $\alpha = 1$ while relying on a more common function $\tanh^{-1}(z)$. 

\begin{figure}
    \centering
    \includegraphics[width=0.98\linewidth]{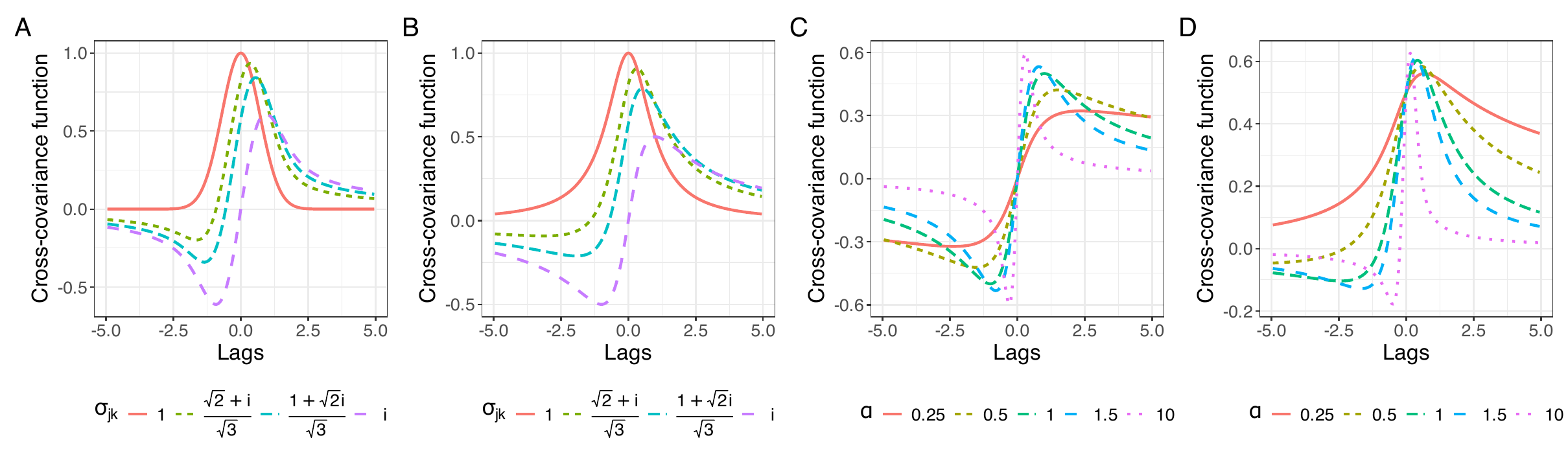}
    \caption{Example cross-covariance functions for squared-exponential and Cauchy covariances in $d=1$. (A) Squared-exponential with $a = 1$ and $\sigma_{jk}$ varying. (B) Cauchy  with $\alpha = 1$, $a = 1$, and $\sigma_{jk}$ varying. (C) Cauchy with $a=1$, $\Re(\sigma_{jk})=0$, $\Im(\sigma_{jk}) = 1$, and $\alpha$ varying. (D) Cauchy with $a=1$, $\Re(\sigma_{jk}) = \Im(\sigma_{jk}) = 1/2$, and $\alpha$ varying. 
    }
    \label{fig:1ds}
\end{figure}

We plot example Cauchy covariances in Fig. \ref{fig:1ds}. 
The cross-covariances are odd when $\sigma_{jk} = \I$ and even when $\sigma_{jk} = 0$. 
The parameter $a$ continues to control the inverse range of the cross-covariance, while $\alpha$ is related to the cross-covariance's tail behavior. 
The cross-covariances with $\alpha = 1/4$ decay much more slowly than larger $\alpha$. 
In particular, we next present tail decay behavior of the asymmetric part of the covariance. 

\begin{remark}[Tail decay, Cauchy]\normalfont
An attractive property of the Cauchy covariance is its controllable tail decay through $\alpha $: it has asymptotic expansion $C_{jk}^{\Re}(\vecuse{h}) = (1 + a^2 \lVert \vecuse{h}\rVert^2)^{-\alpha} \overset{\lVert \vecuse{h}\rVert \to \infty}{\sim} (a \lVert \vecuse{h}\rVert)^{-2\alpha}$. 
For the asymmetric part, when $\langle \vecuse{h}, \tilde{\vecuse{x}}\rangle = 0$, $C_{jk}^{\Im}(\vecuse{h}) = 0$, and we focus on the perpendicular direction when $\langle \vecuse{h}, \tilde{\vecuse{x}}\rangle = \lVert \vecuse{h}\rVert$.
In this case, we obtain the asymptotic expansion \begin{align*}
        C_{jk}^\Im(\vecuse{h}) \overset{\langle \vecuse{h}, \tilde{\vecuse{x}}\rangle  = \lVert \vecuse{h}\rVert \to \infty}{\sim} \left\{\begin{array}{cccc}    
(a\lVert \vecuse{h}\rVert)^{-2\alpha} &  \times&\frac{\Gamma(\alpha + \frac{1}{2})\Gamma(\frac{1}{2} - \alpha)}{\Gamma(\alpha)\Gamma(1 - \alpha)}, & \alpha < \frac{1}{2}\\
            (a\lVert \vecuse{h}\rVert)^{-1}\log(a\lVert \vecuse{h}\rVert) & \times& \frac{2}{\pi},  & \alpha = \frac{1}{2}\\
            (a\lVert \vecuse{h}\rVert)^{-1} & \times & \frac{\Gamma(\alpha - \frac{1}{2})}{\sqrt{\pi}\Gamma(\alpha)}, & \alpha > \frac{1}{2} 
        \end{array}\right. .
\end{align*}
These expressions follow from 15.4.20, 15.4.21, and 15.4.23 of \cite{NIST:DLMF}, respectively.
Thus, when $\alpha < 1/2$, $C_{jk}^{\Im}(\lVert \vecuse{h}\rVert)$ matches the order of the $(a\lVert \vecuse{h}\rVert)^{-2\alpha}$ tail behavior of $C_{jk}^{\Re}(\vecuse{h})$ in that direction up to a constant. 
When $\alpha = 1/2$, it matches the order of $C_{jk}^{\Re}(\vecuse{h})$ up to a logarithmic term. 
When $\alpha > 1/2$, $C_{jk}^{\Im}(\vecuse{h})$ has slower tail decay than $C_{jk}^{\Re}(\vecuse{h})$ in that direction. 
When $\alpha \leq 1/2$, it is possible that $C_{jk}(\vecuse{h})$ has positive tail behavior in both directions of $\langle \vecuse{h}, \tilde{\vecuse{x}}\rangle  =\pm \lVert \vecuse{h}\rVert \to \pm \infty$ depending on the sizes of $\Re(\sigma_{jk})$ versus $\Im(\sigma_{jk})$. 
Otherwise, the $\textrm{sign}(C_{jk}(\vecuse{h})) = -\textrm{sign}(C_{jk}(-\vecuse{h}))$ when $\Im(\sigma_{jk}) \neq 0$ for large enough $\lVert\vecuse{h}\rVert$ in this direction.

\end{remark}

\subsection{Mat\'ern covariance}\label{sec:matern_cov}

The Mat\'ern covariance, defined as $
    C_{jk}^{\Re}(\vecuse{h})= 2^{1-\nu}\Gamma(\nu)^{-1}(a\lVert \vecuse{h}\rVert)^{\nu} K_\nu(a\lVert \vecuse{h}\rVert)$,
is a key model used in spatial statistics and related fields \citep{porcu2023mat}, where $K_\nu(z)$ is a modified Bessel function. 
While closed form $C_{jk}^{\Im}(\vecuse{h})$ studied in \cite{yarger2023multivariate} remain limited to $d=1$, $C_{jk}^{\Im}(\vecuse{h})$ can be computed for $d=2$ through fast Fourier transforms \citep{yarger2023multivariate}, as in Fig.~\ref{fig:spatial_comparison}.    
We present $\nu \notin \{1/2, 3/2, \dots\}$ and $\nu = 1/2$ separately. 

The forms again rely on non-standard functions and must be broken down into cases when $\nu \in \{1/2, 3/2, \dots\}$ or not, which mirrors the Mat\'ern's reduction to an exponential times a polynomial for these values.

\begin{example}[Mat\'ern]\normalfont\label{ex:matern_nu}
    Suppose that $\nu \notin \{1/2, 3/2, \dots\}$ and $d=1$. Then, based on (8I.2) of \cite{king2009hilbert} and 3.771.1 of \cite{gradshteyn2014table} we have
    \begin{align*}
    C_{jk}^{\Im}(h)&=\frac{\pi \textrm{sign}(h)}{2^\nu\Gamma(\nu)\cos(\pi\nu)}
     (a|h|)^{\nu}\left\{I_{\nu}(a|h|) - L_{-\nu}(a|h|)\right\},
\end{align*}where $I_{\nu}(z)$ is a modified Bessel and $L_{\nu}(z)$ is a modified Struve function \citep{NIST:DLMF}.
\end{example}

This expression has similar terms as the Mat\'ern covariance, with 
$2^{-\nu}(a|h|)^\nu / \Gamma(\nu)$ and the arguments of the special functions. 
We next turn to the $\nu = 1/2$ exponential case. 

\begin{example}[Exponential]\normalfont\label{ex:exp_asymmetric}
Take  $\nu = 1/2$ and $d=1$ with $C_{jk}^{\Re}(h) = \textrm{exp}\left(-a|h|\right)$. Then we have $C_{jk}^{\Im}(h) = \left\{\textrm{sign}(h)/\pi\right\}\{\textrm{exp}(a|h|)E_1(a|h|)+ \textrm{exp}(-a|h|)\textrm{Ei}(a|h|)\},$
where $E_1(z)$ and $\textrm{Ei}(z)$ are commonly-defined exponential integrals \citep[Chapter 6 of][]{NIST:DLMF}. This follows from the Hilbert transform (3.3) in the tables of \cite{king2009hilbert}.
\end{example}

These results were previously presented in \cite{yarger2023multivariate} and are limited to $d=1$, but provide forms that could potentially be extended in the future. 
For $d > 1$, the general Mat\'ern desired form implied by Proposition \ref{prop:schoen_asymmetric} is \begin{align*}
    C_{jk}^{\Im}(\vecuse{h}) &= \frac{1}{2^{2\nu}\Gamma(\nu)}\int_0^\infty \textrm{exp}\left(-w\lVert \vecuse{h}\rVert^2 \right)\textrm{erfi}\left(\sqrt{w}\langle \vecuse{h}, \tilde{\vecuse{x}}\rangle\right) w^{\nu-1} \textrm{exp}\left(-\frac{1}{4w}\right)dw,
\end{align*} 
whose expression in terms of special functions, if it exists, does not immediately follow.
However, the integral's value has been determined in the case where $\langle \vecuse{h}, \tilde{\vecuse{x}}\rangle = \lVert \vecuse{h}\rVert$ to match the $d=1$ case in Example \ref{ex:matern_nu} for $\nu \notin \{1/2, 3/2, \dots\}$ \citep[for example, I.17.71 of][]{oberhettinger2012tables}.

\subsection{Other covariance functions and multivariate flexiblity}

We briefly discuss prospects for asymmetric covariance functions for a few other covariance functions.
In \cite{yarger2025multivariate}, fast Fourier transforms computed $C_{jk}^{\Im}(\vecuse{h})$ for the Confluent Hypergeometric covariance, which offers control over both the smoothness like the Ma\'tern covariance through $\nu$, as well as tail decay like the Cauchy covariance. For the Confluent Hypergeometric covariance, $g(v)$ is Beta-prime \citep[Theorem 1,][]{ma2022beyond}, and the closed forms based on Proposition \ref{prop:schoen_asymmetric} are still unclear. Such a result would provide closed form asymmetric cross covariances with large control over origin and tail behavior of cross-covariances. 

Compact covariance functions, where $C_{jk}^{\Re}(\vecuse{h}) = 0$ for all $\lVert\vecuse{h}\rVert > c $ for some $c >0$, have also been proposed \citep[e.g.][]{bevilacqua2025parsimonious, faouzi_spacetime_2025}. 
By introducing zero entries in the covariance matrix, calculations can be improved by reducing the total number of computations. 
However, a compactly supported $C^{\Re}_{jk}(\vecuse{h})$ may not generally have a compactly supported $C^{\Im}_{jk}(\vecuse{h})$.
For example, the triangular covariance function in $d=1$ \citep{chiles2012geostatistics} has Hilbert transform (9.6) of \cite{king2009hilbert} that satisfies $C_{jk}^{\Im}(h) = 0$ only when $h=0$.
Thus, attractive sparsity properties of compact covariances do not appear to extend to cross-covariances with this asymmetric formulation. 

We also discuss the introduction of flexibility in these forms with regard to the marginal covariance functions, so that such models can describe heterogeneous processes which maintaining validity. Such forms are necessary for most practical applications, as scales, tail behavior, or other properties may vary by process. 
Above, we assumed throughout that $f_{j}(\vecuse{x}) = f_{k}(\vecuse{x})$. 
In Appendix \ref{app:vary_params}, we present results when the covariances have different parameters. 
As presented above and in the Appendix, we need only check if $\matusesigma{\sigma}$ is Hermitian and positive definite for validity.
However, if we suppose that $\tilde{\vecuse{x}}$ may vary across cross-covariances, denoted $\tilde{\vecuse{x}}_{jk} = \tilde{\vecuse{x}}_{kj}$ for $j\neq k$. 
When $p \geq 3$, one needs the positive-definiteness for all $\vecuse{x}\in \mathbb{R}^d$ of $\left[\Re(\sigma_{jk}) - \I\textrm{sign}(\langle \vecuse{x}, \tilde{\vecuse{x}}_{jk}\rangle)\Im(\sigma_{jk})\right]_{j,k=1}^p$. 
In preliminary checks, care may need to be taken to specify the matrix $[\sigma_{jk}]_{j,k=1}^p$ to achieve this. 

\section{Reflective asymmetric space-time covariances}\label{sec:univariate_st}

We next discuss the univariate space-time setting, which will use some results in the previous section.  
Denote $a_{\vecuse{s}}$ ($a_t$) as an inverse spatial (temporal) range parameter where necessary to distinguish, as well as marginal covariances $C_{ space}(\vecuse{h})$ and $C_{time}(u)$. 
For univariate space-time covariances, the asymmetry must be introduced through the combination of spatial and temporal dimensions, which motivates a form that combines the form $-\I\textrm{sign}(\langle \vecuse{x}, \tilde{\vecuse{x}}\rangle)$ with a similar form for the temporal domain. 

We propose this construction for asymmetric models in the spectral domain: \begin{align}\label{eq:space_time_constr_density}
    f(\vecuse{x}, \eta) &= \sigma f^*(\lVert \vecuse{x}\rVert , |\eta|) \left\{1 + \xi \textrm{sign}(\langle \vecuse{x}, \tilde{\vecuse{x}}\rangle)\textrm{sign}(\eta)\right\},
\end{align}where $f^*(\lVert \vecuse{x}\rVert , |\eta|)$ is the symmetric spectral density of a space-time covariance, $\sigma > 0$ is a variance parameter, $\xi \in (-1, 1)$ is an asymmetry parameter with $\xi = 0$ corresponding to symmetry, and $\tilde{\vecuse{x}}$ is again a fixed unit vector of length $d$. 
When $\xi = 0$, we revert to the previous space-time covariance defined by $ \sigma f^*(\lVert \vecuse{x}\rVert , |\eta|)$, while there is more asymmetry when $|\xi| $ is closer to $1$. 
The parameter $\tilde{\vecuse{x}}$ now represents a direction of spatial asymmetry that is coupled with asymmetry in temporal dimension. 
In particular, when $\xi >0$, we will have a resulting covariance $C(\vecuse{h},u)$ such that $C(\tilde{\vecuse{x}}, u) > C(-\tilde{\vecuse{x}}, u)$ for $u > 0$. 
This suggests the process is more dependent in direction $\tilde{\vecuse{x
}}$ at a later time compared to $-\tilde{\vecuse{x
}}$. 
This can be appropriate for dynamic processes evolving through space and time, like wind or weather patterns transporting East over time.

We demonstrate the applicability of Bochner's Theorem as presented in Theorem \ref{thm:st_bochner}.  \begin{enumerate}
    \item[(a)]
    The density is symmetric with respect to both arguments at the same time: $f(-\vecuse{x}, -\eta) = \sigma f^*(\lVert - \vecuse{x}\rVert , \lvert-\eta\rvert) \left\{1 + \xi (-1)(-1)\textrm{sign}(\langle \vecuse{x}, \tilde{\vecuse{x}}\rangle)\textrm{sign}(\eta)\right\} = f(\vecuse{x}, \eta)$.
    \item[(b)] The positivity and integrability of $f(\vecuse{x},\eta)$ is ensured when $f^*(\lVert \vecuse{x}\rVert, |\eta|)$ is positive and integrable. 
    As $0 < 1 - |\xi| \leq  1 + \xi \textrm{sign}(\langle \vecuse{x}, \tilde{\vecuse{x}}\rangle)\textrm{sign}(\eta) \leq 1 + |\xi| < 2$ for $\xi \in (-1, 1)$, one has $(1- |\xi|)\sigma f^*(\lVert\vecuse{x} \rVert, |\eta|) \leq f(\vecuse{x},\eta) \leq (1 + |\xi|) \sigma f^*(\lVert \vecuse{x} \rVert, |\eta|)$.
\end{enumerate}
Thus, maintaining validity of the covariance while introducing asymmetry is straightforward.
We outline specific forms of these models, at first when $f^*(\lVert \vecuse{x}\rVert, |\eta|)$ is separable.

\subsection{Separable-type models}\label{sec:separable_type}
We first consider building space-time models of the form \eqref{eq:space_time_constr_density} based on separable covariance models, where $f^*(\lVert \vecuse{x}\rVert, |\eta|) = f_{space}(\lVert \vecuse{x}\rVert) f_{time}(|\eta|)$, and the integrals separate:\begin{align}
        C(\vecuse{h}, u) &= \sigma\int_{\mathbb{R}^d} \textrm{exp}\left(\I\langle\vecuse{h},  \vecuse{x}\rangle\right)f_{ space}(\lVert \vecuse{x} \rVert) d\vecuse{x}  \int_{\mathbb{R}}  \textrm{exp}\left(\I u\eta\right) f_{time}(|\eta|) d\eta \nonumber \\
        &+\sigma \xi \int_{\mathbb{R}^d} \textrm{exp}\left(\I \langle \vecuse{h},  \vecuse{x}\rangle\right)\textrm{sign}(\langle \vecuse{x}, \tilde{\vecuse{x}}\rangle) f_{space}(\lVert \vecuse{x} \rVert)d\vecuse{x}  \int_{\mathbb{R}}  \textrm{exp}\left(\I u\eta\right) \textrm{sign}(\eta)f_{time}(|\eta|) d\eta \nonumber \\
        &=: \sigma \left\{C_{ space}^{\Re}(\vecuse{h}) C_{space}^{\Re}(u) +\xi C_{ space}^\Im(\vecuse{h}) C_{time}^\Im(u) \right\}.\label{eq:time_separable}
\end{align}
This covariance is a sum of two separable parts, one symmetric and the other asymmetric, and is no longer separable when $\xi \neq 0$.
We assume one has chosen $C_{space}^{\Re}(\vecuse{0}) = C_{time}^{\Re}(0) = 1$ without loss of generality, so that $\sigma = C(\vecuse{0},0)$ represents the variance.
In this form of \eqref{eq:time_separable}, the key observation is that the spatial and temporal parts correspond to cross-covariances in Section \ref{sec:spatial_cc}, as $\textrm{sign}(\langle \vecuse{x} , \tilde{\vecuse{x}}\rangle)\textrm{sign}(\eta) = -\left\{-\I \textrm{sign}(\langle \vecuse{x} , \tilde{\vecuse{x}}\rangle)\right\}\left\{-\I \textrm{sign}(\eta)\right\}$.
We present an example from the large variety of combinations of covariances that may be chosen. 

\begin{example}\normalfont\label{ex:cauchyhalf}
Suppose we take a squared-exponential spatial covariance and a Cauchy temporal covariance with $\alpha = 1/2$: $
         C_{space}^{\Re}(\vecuse{h}) = \textrm{exp}\left(-a_{\vecuse{s}}^2\lVert \vecuse{h}\rVert^2\right)$ and $ C_{time}^{\Re}(u) =(a_t^2u^2 + 1)^{-1/2}$. 
The corresponding asymmetric separable-type space-time covariance is \begin{align*}
     C(\vecuse{h}, u) 
     &= \sigma \frac{\textrm{exp}\left(-a_{\vecuse{s}}^2 \lVert \vecuse{h}\rVert^2\right)}{\sqrt{a_t^2 u^2 + 1}}\left\{1 + \xi \textrm{erfi}\left(a_{\vecuse{s}}\langle \vecuse{h}, \tilde{\vecuse{x}}\rangle\right)\frac{2}{\pi}\sinh^{-1}(a_tu)\right\}.
\end{align*}
    
\end{example}

We next plot Example \ref{ex:cauchyhalf} with $d=1$ in Fig.~\ref{fig:separable_asymmetric}. 
For $\xi > 0$, there is increased covariance for $\textrm{sign}(h) = \textrm{sign}(u)$ and decreased covariance for $\textrm{sign}(h) = -\textrm{sign}(u)$ compared to the separable, symmetric model. 
The sign of $\xi$ has control over the direction of asymmetry, while $|\xi|$ controls the asymmetry strength, with stronger asymmetric features when $\xi = -0.9$ compared to $\xi = 0.4$. 
Interestingly, the asymmetric part can introduce negative covariance for some spatial and temporal lags, which occurred previously in \cite{mateu_recent_2008} and \cite{horrell2017half}. 

\begin{figure}[!t]
    \centering
        \includegraphics[width=0.9\linewidth]{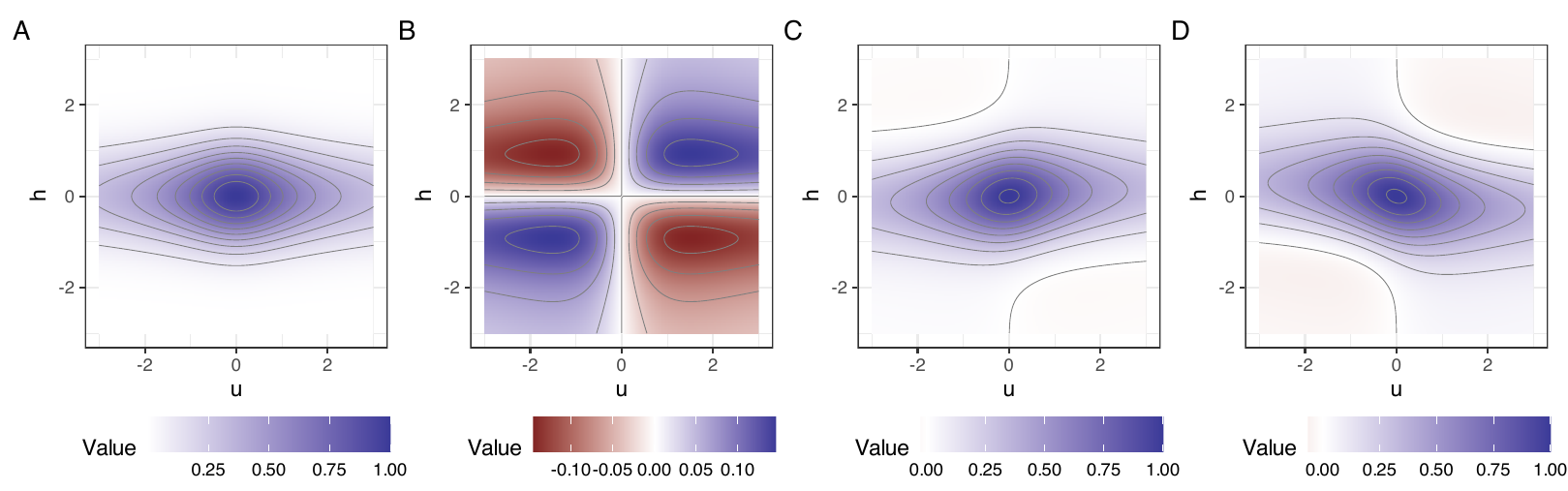}
    \caption{Separable-type covariances for $d=1$ as in Example \ref{ex:cauchy_spatial_a1} with $a_s = a_t = 1$ and $\sigma = 1$. (A) The symmetric case $\xi = 0$, (B) the asymmetric part only when taking $\xi = 1$ for illustration, (C) $\xi = 0.4$, (D) $\xi = -0.9$.}
    \label{fig:separable_asymmetric}
\end{figure}
We discuss a few additional properties of the general model.

\begin{remark}[Marginal covariances]\normalfont\label{rem:marginal}
    The marginal covariances are retained in the new model: $
        C(\vecuse{h}, 0) = \sigma C_{space}^{\Re}(\vecuse{h})$ and $
        C(\vecuse{0}, u) = \sigma C_{time}^{\Re}(u)$, as $C_{space}^{\Im}(\vecuse{0}) = C_{time}^{\Im}(0) = 0$. This second property is demonstrated in panel (B) of Fig.~\ref{fig:separable_asymmetric} with values of $0$ on the axes. 
        
\end{remark}

\begin{remark}[Separability]\normalfont
    When we take \eqref{eq:time_separable} and $
        C(\vecuse{h}, 0) = \sigma C_{space}^{\Re}(\vecuse{h})$ and $
        C(\vecuse{0}, u) = \sigma C_{ time}^{\Re}(u)$, the separability measure becomes \begin{align*}
        R(\vecuse{h},u)  &= C(\vecuse{h}, u) C(\vecuse{0}, 0) - C(\vecuse{h}, 0) C(\vecuse{0}, u) \\
        &= \left\{\sigma C_{space}^{\Re}(\vecuse{h}) C_{space}^{\Re}(u) +\sigma\xi C_{space}^\Im(\vecuse{h}) C_{time}^\Im(u) \right\} \left\{\sigma\right\}  \\
        &~~~~~~~ - \left\{\sigma C_{space}^{\Re}(\vecuse{h})\right\}
        \left\{\sigma C_{time}^{\Re}(u)\right\} \\
        &= \left[\sigma\right]^2  \xi C_{space}^{\Im}(\vecuse{h}) C_{time}^{\Im}(u),
    \end{align*}and thus introduces both positive and negative nonseparability (as in Fig.~\ref{fig:separable_asymmetric}). Introducing both negative and positive asymmetry appears to be a necessary property of an asymmetric space-time covariance. 
\end{remark}

\begin{remark}[Axial symmetry]\normalfont
If $\xi \neq 0$, the covariance function is spatially axially symmetric \citep[as in Definition \ref{def:axial} and][]{park2008testing} for the $k$-th dimension if and only if the $k$-th entry of $\tilde{\vecuse{x}}$ is $0$.
\end{remark}

We give an example in Fig.~\ref{fig:2dvary_u_h} with $d=2$. 
Asymmetry is introduced along the direction $\tilde{\vecuse{x}} = (1,1)^\top/\sqrt{2}$.
On the other hand, the asymmetric covariance matches the symmetric covariance along $\vecuse{h} = (-1, 1)^\top / \sqrt{2}$.
For directions between these two extremes, more moderate asymmetry is present. 
We also plot in Fig.~\ref{fig:2dvary_u_h} (B) an axially symmetric model in space with respect to this second spatial dimension as $C( (h_1, h_2)^\top, u) = C( (h_1, -h_2)^\top, u)$ for each $h_1$, $h_2$, and $u$, when $\tilde{\vecuse{x}} = (1, 0)^\top$. 
The spatial contours when $u \neq 0$ are not circular or elliptical but more raindrop-shaped, with a much quicker decay of covariance in one direction than the other. 

These models represent a simple construction to introduce asymmetry in common separable space-time covariances. 
There are only $1 + (d-1)$ additional parameters needed to describe the covariance through $\xi$ (the asymmetry strength) and $\tilde{\vecuse{x}}$ (the asymmetry direction for $d > 1$), compared to the $O(d^2)$ parameters in the full Lagrangian model. 
The parameters $(\xi, \tilde{\vecuse{x}})$ represent the same model as $(-\xi, -\tilde{\vecuse{x}})$, so one may reparameterize so that $0 \leq \xi < 1$. 
Similarly, estimating $\tilde{\vecuse{x}}$ will be challenging with $\xi \in (-1, 1)$ when $\xi \approx 0$. 
In the case of an estimated $\xi \approx 0$, one should conclude that there is little asymmetry estimated in the data, and the estimate of $\tilde{\vecuse{x}}$ is not particularly meaningful or impactful.

\begin{figure}[th]
    \centering
    \includegraphics[width=0.98\linewidth]{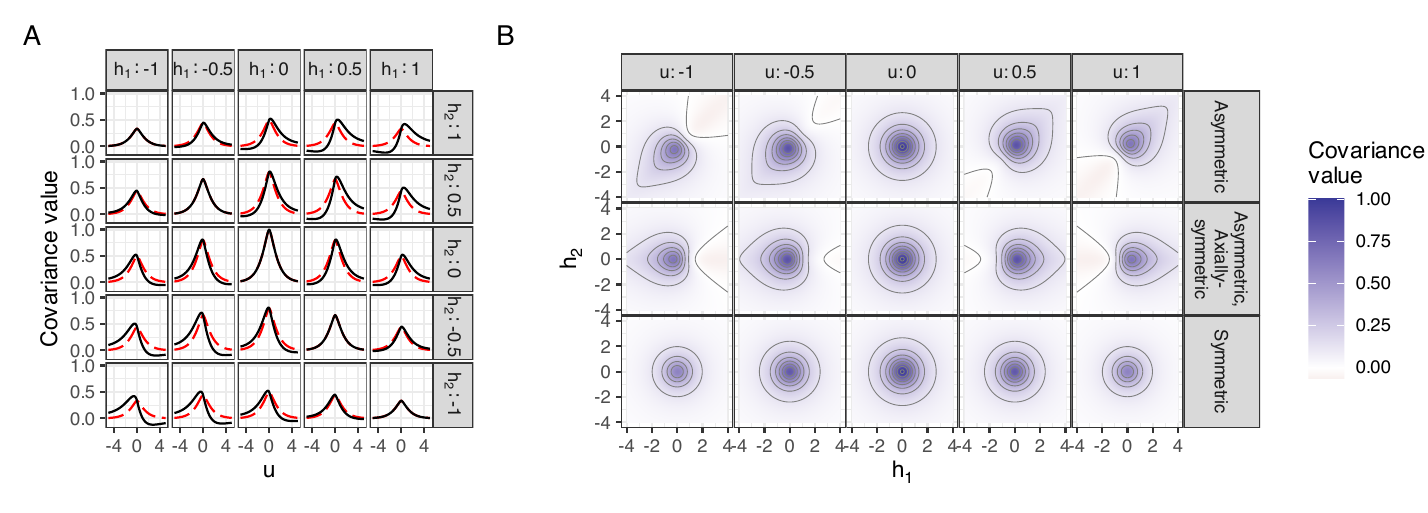}
    \caption{Cross-sections of the separable-type covariances for $d=2$ using the Cauchy (spatial) and Mat\'ern (temporal) model with $a_{\vecuse{s}} = a_t = 1$, $\alpha = 1$, $\nu = 1$, $\sigma = 1$, $\tilde{\vecuse{x}} = (1,1)^\top / \sqrt{2}$ and $\xi = 0.9$. (A) Covariance as a function of $u$ for varying $\vecuse{h} = (h_1, h_2)^\top$, with the symmetric model $\xi = 0$ also plotted in dashed red; (B) covariance as a function of $\vecuse{h}$ for varying $u$, with the symmetric model $\xi =0$ and the axially-symmetric model with $\xi = 0.9$ and $\tilde{\vecuse{x}} = (1,0)^\top $ also plotted.}
    \label{fig:2dvary_u_h}
\end{figure}

\subsection{Asymmetric Gneiting models}

We next establish asymmetric versions of Gneiting-type covariances \citep{gneiting_nonseparable_2002}, which are nonseparable even when $\xi = 0$. 
Now $\sigma f^*(\lVert \vecuse{x}\rVert , |\eta|)$ represents the spectral density of a form of \eqref{eq:gne_general}.
The closed forms are limited to $d=1$. 
We use $\gamma(u) = a_t^2u^2$ in particular here.

\begin{example}[$d=1$ squared-exponential]\normalfont\label{ex:sq_exp_gneiting_separability}
Let $b\in[0,1]$, $\delta > 0$, and $\tau = b/2 + \delta $. 
We have the valid space-time covariance \cite[2.13.1.5,][]{korotkov2020integrals} \begin{align*}
    C(h,u) &= \frac{\sigma}{\left(a_t^2 u^2+1\right)^\tau}\textrm{exp}\left\{- \frac{a_s^2h^2}{\left(a_t^2u^2 + 1\right)^b}\right\} \left[1 + \xi\textrm{erf}\left\{\frac{a_sh}{\left(a_t^2u^2 + 1\right)^{\frac{b}{2}}}a_tu\right\}\right].
\end{align*}
\end{example}

This form involves the regular error function $\textrm{erf}(z)$. 
Both the symmetric and asymmetric parts are nonseparable. 
If we consider the time-adjusted spatial distance $h_u = h / (a_t^2 u^2 + 1)^{b/2}$, the covariance takes the form $C(h,u) = \sigma (a_t^2 u^2 + 1)^{-\tau} \exp(-a_s^2 h_u^2)\{1 + \xi \textrm{erf}(a_s h_u a_t u)\}$. Thus the time-adjusted spatial distance and the temporal lag are multiplied together within the error function. 
As the error function is an odd function, the sign of the additional component is $\textrm{sign}(\xi h u)$, the same as the separable-type models when $d=1$. 
However, as $|\textrm{erf}(z)| \leq 1$ for real $z$, $C(h,u) \geq 0$ for all $h$ and $u$ and negative values are not attained.

We plot an overview and comparison of this model in Fig.~\ref{fig:gneiting_separability_par}. 
While the parameter $\xi$ controls the asymmetry in the two dimensions, $b$ controls the level on nonseparability of the symmetric covariance. 
The parameter $b$ changes the shape of the covariance functions, with larger $b$ values corresponding to rounder contours. 
The study of Gneiting covariances has uncovered the dimple effect \cite{kent2011dimple, cuevas_contours_2017, alegria_dimple_2017}, where $C(h,0) < C(h,u)$ for some $h$ and $u \neq 0$. 
This effect can be seen in the right panel of Fig.~\ref{fig:gneiting_separability_par} when $h=  1.5$ for the symmetric Gneiting model. 
The asymmetric models essentially introduce a ``half-dimple'' in one direction, with a dimple-like effect for one sign of $u$ and no such effect for the other sign of $u$. 
This general asymmetry is perhaps more natural than the dimple effect that \cite{kent2011dimple} labels as counterintuitive but \cite{alegria_dimple_2017} defends.
In Appendix \ref{app:gneiting}, we also present Gneiting-type asymmetric Cauchy covariances. 

\begin{figure}[!t]
    \centering
    \includegraphics[width=0.98\linewidth]{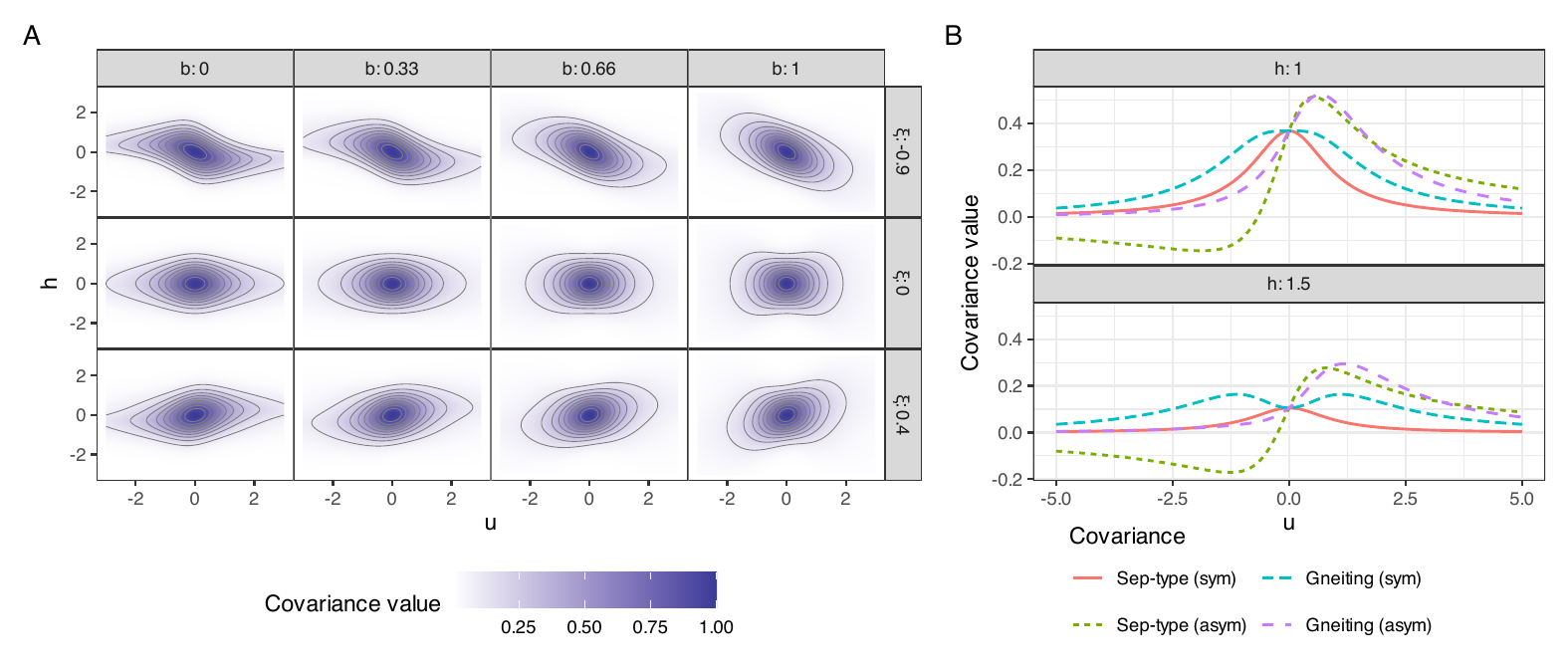}
    \caption{(A) Gneiting-type covariances as in Example \ref{ex:sq_exp_gneiting_separability} with separability parameter $b$ and asymmetry parameter $\xi$ for $d=1$ with $a_s = a_t = 1$, $\sigma = 1$, and $\delta = 1$. The reduction $\xi = b = 0$ corresponds to separability and $\xi =0 $ to symmetry. (B) Comparison of separable-type and Gneiting symmetric and asymmetric covariances with the same marginal covariances $C_{space}(h) = \exp(-a_s^2 h^2)$ and $C_{time}(u) = (a_t^2 u^2 + 1)^{-1}$. In particular, $a_s = a_t = 1$, $\xi = 0.9$ for asymmetric models, and $b=1$ and $\delta = 1/2$ for Gneiting models.}

    \label{fig:gneiting_separability_par}
\end{figure}

\section{Practical considerations}\label{sec:practical}

We next discuss simulation, nonstationarity and anisotropy, and computation. 

\subsection{Efficient spectral simulation}

We propose the following algorithm for simulating random fields with separable-type asymmetric covariance. 
For $\ell = 1, 2, \dots, L$, simulate $\phi_\ell \sim \textrm{Uniform}[0, 2\pi]$, $\psi_\ell \sim \textrm{Uniform}[0, 2\pi]$, $\vecuse{x}_\ell \sim g_{space}$, and $\eta_\ell \sim g_{time}$, where all are independent from each other. Here $g_{space}$ and $g_{space}$ are spatial and temporal proposal densities, respectively.
Following \cite{emery2016improved} for the spatial case, we propose to simulate using the following. 
Define \begin{align*}
    X_\ell^{\Re}(\vecuse{s}) &= \cos(\langle \vecuse{s}, \vecuse{x}_\ell\rangle  + \phi_\ell )\allowbreak \sqrt{\frac{f_{ space}(\lVert \vecuse{x}_\ell\rVert)}{ g_{ space}(\lVert \vecuse{x}_\ell\rVert)}} & W_\ell^\Re(t) & = \cos(t\eta_\ell + \psi_\ell  )\sqrt{\frac{f_{time}(| \eta_\ell |)}{g_{time}(| \eta_\ell |)}}\\
    X_{\ell}^{\Im}(\vecuse{s})&=  \sin(\langle \vecuse{s}, \vecuse{x}_\ell\rangle  + \phi_\ell )\allowbreak \sqrt{\frac{f_{ space}(\lVert \vecuse{x}_\ell\rVert)}{ g_{ space}(\lVert \vecuse{x}_\ell\rVert)}} & W_\ell^{\Im}(t) & = \sin(t\eta_\ell + \psi_\ell  )\sqrt{\frac{f_{time}(| \eta_\ell |)}{g_{time}(| \eta_\ell |)}},
\end{align*}
and construct \begin{align*}
    Y(\vecuse{s},t) = \sqrt{\frac{\sigma}{L}} \sum_{\ell = 1}^L c_1 X_\ell^{\Re}(\vecuse{s}) W_\ell^{\Re}(t) + c_2 X_\ell^{\Im}(\vecuse{s}) W_\ell^{\Im}(t)\textrm{sign}(\langle \vecuse{x}_\ell, \tilde{\vecuse{x}}\rangle)\textrm{sign}(\eta_\ell),
\end{align*}where $c_1 = \sqrt{1 +\xi} + \sqrt{1- \xi}$ and $c_2 = \sqrt{1+\xi} - \sqrt{1- \xi}$.
Suggestions for choosing the proposal densities $g_{space}$ and $g_{time}$ can be reviewed in \cite{arroyo_algorithm_2021}. In analogy to importance sampling, it is suggested that the proposal have heavier tails than the target.
Extensions to Gneiting-type models are also possible \citep{allard_simulating_2020}. 
We provide a derivation that we target the correct covariance in Appendix \ref{app:spec_sim}, in that $\textrm{Cov}(Y(\vecuse{s}_1, t_1), Y(\vecuse{s}_2, t_2)) = C(\vecuse{s}_1 - \vecuse{s}_2, t_1 - t_2)$ as in \eqref{eq:time_separable}.

We evaluate the spectral simulation with a squared-exponential covariance, with $\sigma=1$, $a_s = 1$, $a_t = 1/2$, $\xi = 0.9$, $\tilde{\vecuse{x}} = (1,0)^\top$, simulated on a $26 \times 26 \times 12$-dimensional grid on $[0,5] \times [0,5] \times [0,11]$. 
The proposals were Mat\'ern spectral densities with $\nu = 0.5$ and $a = 1$, and we take $L = 20{,}000$.
We form empirical covariance estimates for each lag $\hat{C}(\vecuse{h},u) = (1/n_{\vecuse{h}, u}) \sum_{\vecuse{s}_1 - \vecuse{s_2} = \vecuse{h}, t_1 - t_2 = u} Y(\vecuse{s}_1, t_1) Y(\vecuse{s}_2, t_2)$ assuming zero mean, where $n_{\vecuse{h}, u}$ is the number of observed location combinations with those lags. 
The results are presented in Fig.~\ref{fig:bigsim}. 
The empirical estimate for one simulation picks up the general left-to-right direction of the asymmetry, but there is also a substantial amount of noise in the estimate. 
The average across $200$ independent simulations matches quite close to the true covariance. 
Evaluation of the standard errors are given in Appendix \ref{app:spec_sim}.

\begin{figure}[th]
    \centering

    \includegraphics[width=0.98\linewidth]{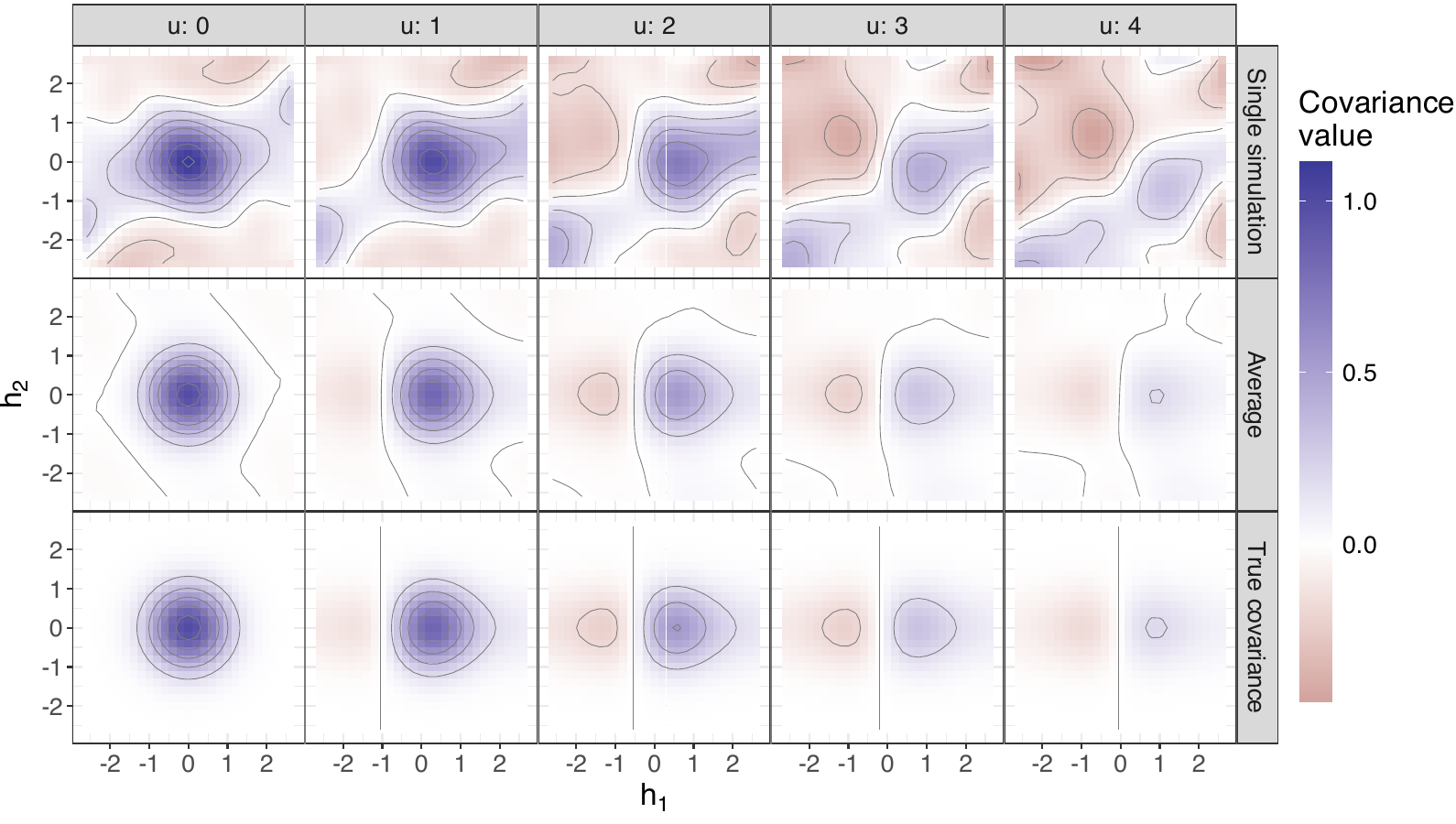}
    \caption{Evaluation of spectral simulation, using squared-exponential covariances. The top row gives the empirical estimate $\hat{C}(\vecuse{h}, u)$ for one simulation. The middle row gives the average of empirical estimates over $200$ independent simulations. 
    The bottom row is the true targeted covariance.}
    \label{fig:bigsim}
\end{figure}

\subsection{Nonstationary and anisotropic models}

We may construct versions where the covariance depends on the spatial locations through more than $\vecuse{h}$ (nonstationarity) and will satisfy $C(\vecuse{h}_1, 0) \neq C(\vecuse{h}_2, 0)$ for two $\vecuse{h}_1, \vecuse{h}_2 \in \mathbb{R}^d$ such that $\lVert \vecuse{h}_1\rVert = \lVert \vecuse{h}_2 \rVert$ (anisotropy) \citep{paciorek_spatial_2006}.
In particular, we specify a $d\times d$ spatially-varying positive-definite matrix $\matuse{\Sigma}_{\vecuse{s}}$. 
Consider replacing the squared-exponential spectral density $f_{\vecuse{s}}(\lVert \vecuse{x}\rVert)$ with $
     f_{\vecuse{s}_1, \vecuse{s}_2}(\vecuse{x})  = \left(2\sqrt{\pi}\right)^{-d} \left|\matuse{\Sigma}_{\vecuse{s}_1}\right|^{1/4}\left|\matuse{\Sigma}_{\vecuse{s}_2}\right|^{1/4}\textrm{exp}\left(- \vecuse{x}^\top\matuse{\Sigma}_{\vecuse{s}_1, \vecuse{s}_2} \vecuse{x}/4\right),$ where $\matuse{\Sigma}_{\vecuse{s}_1, \vecuse{s}_2} = (\matuse{\Sigma}_{\vecuse{s}_1} + \matuse{\Sigma}_{\vecuse{s}_2})/2$.
Following arguments from Proposition \ref{prop:sq_exp_spat}, the covariance between locations $\vecuse{s}_1$ and $\vecuse{s}_2$ is $
C_{\vecuse{s}_1, \vecuse{s}_2}(\vecuse{h},u) = \left|\matuse{\Sigma}_{\vecuse{s}_1}\right|^{1/4}\left|\matuse{\Sigma}_{\vecuse{s}_2} \right|^{1/4}\left|\matuse{\Sigma}_{\vecuse{s}_1, \vecuse{s}_2}\right|^{-1/2}C\left(\matuse{\Sigma}_{\vecuse{s}_1, \vecuse{s}_2}^{-1/2}\vecuse{h}, u \right),$ where $C(\vecuse{h}, u)$ is a separable-type model with squared-exponential spatial covariance.
The direction of asymmetry is determined by $\matuse{\Sigma}_{\vecuse{s}_1, \vecuse{s}_2}^{-1/2}\tilde{\vecuse{x}}$ instead of $\tilde{\vecuse{x}}$, inducing a spatially-varying asymmetry direction. 
Here, the parameter $a_{\vecuse{s}}$ is absorbed into $\matuse{\Sigma}_{\vecuse{s}}$, and forms may be constructed for other classes. 
In Fig.~\ref{fig:anisotropic_space_time}, we plot stationary $(\matuse{\Sigma}_\vecuse{s} = \matuse{\Sigma})$ examples. 

\begin{figure}[!t]
    \centering
    \includegraphics[width=0.9\linewidth]{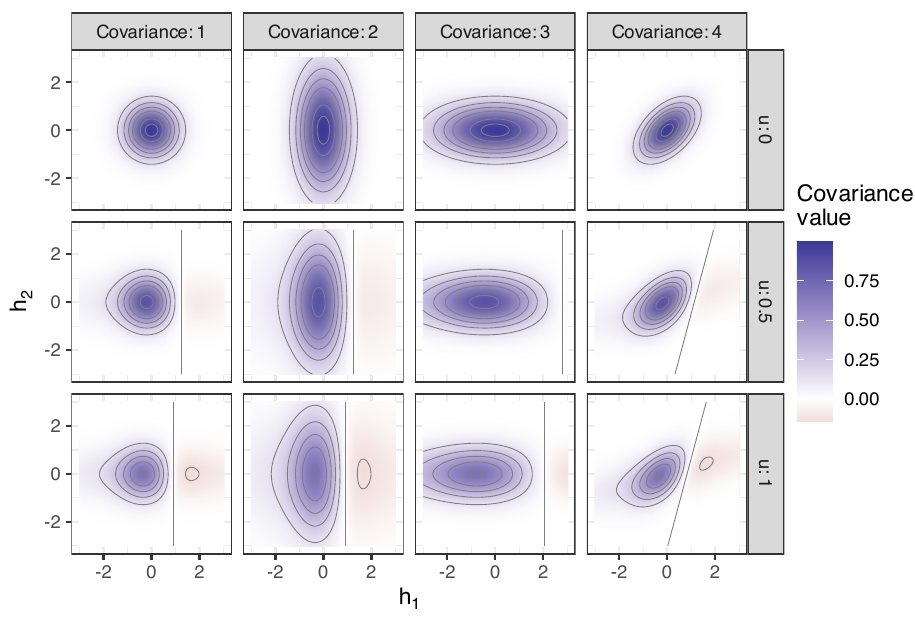}
    \caption{Space-time covariances with anisotropy. The spatial covariance is squared-exponential and the temporal is Mat\'ern with parameter $\nu = 1$ and $a_t = 1$, with $\sigma = 1$, $\xi = -0.9$, and $\tilde{\vecuse{x}} =(1,0)^\top $. The covariances 1 through 4 have $\matuse{\Sigma} = \textrm{diag}(1,1)$, $\matuse{\Sigma} = \textrm{diag}(1,5)$, $\matuse{\Sigma} = \textrm{diag}(5,1)$, $\matuse{\Sigma} =[\mathbb{I}(i = j) + 0.5\mathbb{I}(i\neq j)]_{i,j=1}^2$, respectively. }
    \label{fig:anisotropic_space_time}
\end{figure}

\subsection{Computational implementation}\label{sec:computation}

We discuss three computational aspects of the model. 
While the Lagrangian model \eqref{eq:lagrangian} requires matrix operations of size $d\times d$, the evaluation of $C(\vecuse{h},u)$ we introduce may only require special function evaluation. 
Many functions involved in the asymmetric covariances are available in C++ \citep{galassi2002gnu} with interfaces in \texttt{R}. 
In Table \ref{tab:speed_comparison}, we compare evaluation times in \texttt{R}. 
For the special cases of the Cauchy model $\alpha \in \{1/2,1\}$, evaluation of the asymmetric part takes essentially the same time as the symmetric part. 
Other evaluations require additional computation time, especially for the hypergeometric function.
We attribute the differences in computation time to the computational challenge of evaluating the functions, rather than reflective of any particular covariance property. 
We also compare matrix construction times, as well as optimization times for the multivariate simulation study presented below, in Appendix \ref{app:sim_results}.

\begin{table}[!t]
    \centering
        \caption{In seconds, the time of 50{,}000 evaluations of the cross-covariances components in $d=1$. We use versions of $\textrm{erfi}(z)$ from \cite{rcpp_faddeeva} and \cite{faddeeva}, ${}_2F_1(a,b; c; z)$ from \cite{gsl_R}, and $\textrm{Ei}(z)$ and $E_1(z)$ from \cite{expint}. }

    \label{tab:speed_comparison}
    \begin{tabular}{|ccc rr|}\hline
         Covariance & Function(s) & Implementation & $C_{jk}^{\Re}(h)$ & $C_{jk}^{\Im}(h)$ \\ \hline 
         Sq.~Exp. & $\textrm{erfi}(z)$ & \texttt{Faddeeva}/\texttt{RcppFaddeeva} & 0.069 &  0.152 \\ 
         Cauchy ($\alpha$) &${}_2F_1(a,b; c; z)$& \texttt{gsl} &  0.071 & 1.211 \\ 
         Cauchy ($\alpha = 1/2$) & $\sinh^{-1}(z)$ & \texttt{base} & 0.070 & 0.077 \\ 
        Cauchy ($\alpha = 1$) & $z$ (linear) & \texttt{base}& 0.064 &0.066\\ 
        Exponential & $\textrm{Ei}(z)$, $E_1(z)$ & \texttt{expint} & 0.065 & 0.157\\\hline
    \end{tabular}
\end{table}


There is also a contrast with the Lagrangian model in terms of the number of covariance evaluations. When data is observed on a grid of space (of size $n_{\vecuse{s}}$) and time (regularly-spaced of size $n_t$) with $n = n_{\vecuse{s}} n_t$ total data points, the Lagrangian model requires computing \eqref{eq:lagrangian} on the order of $n_tn_{\vecuse{s}}^2$ times to construct the $n\times n$ covariance matrix. 
For a separable-type model introduced here, the matrix may instead be represented using the Kronecker product $\otimes$: 
    $\sigma \matuse{\Psi}_{space}^\Re \otimes \matuse{\Psi}_{time}^\Re +  \sigma  \xi \matuse{\Psi}_{space}^\Im \otimes \matuse{\Psi}_{time}^\Im$, where the matrices $\matuse{\Psi}$ are $n_{\vecuse{s}} \times n_{\vecuse{s}}$ for the subscript $space$, $n_t\times n_t$ for the subscript $time$, symmetric for the superscript $\Re$, and skew-symmetric for the superscript $\Im$ (note $\matuse{\Psi}_{space}^\Im \otimes \matuse{\Psi}_{time}^\Im$ is symmetric). 
The order of the number of covariance evaluations is then $n_t + n_{\vecuse{s}}^2$, a substantial reduction when $n_t$ and $n_{\vecuse{s}}$ are both moderately large.

However, unlike the separable model, where the likelihood can be evaluated using matrices of size $n_{\vecuse{s}} \times n_{\vecuse{s}}$ and $n_t\times n_t$, the likelihood requires decomposing the full $n\times n$ matrix when $\xi \neq 0$. 
This can be prohibitive for large $n$. 
We use Vecchia's approximation in this setting \citep{vecchia1988estimation, katzfuss_vecchia_2020}, where conditional independence is assumed between data points based on a neighborhood structure. 
Vecchia's approximation for basic space-time models is implemented in the \texttt{GpGp} package in \texttt{R} \citep{GpGp}, as well as for Lagrangian models in \cite{ma2025asymmetric}. 
\cite{idir_improving_2025} evaluate the choice of neighbors for space-time models. 
We implement many introduced models in \texttt{GpGp}'s framework.

\section{Simulation studies}\label{sec:simulation}
We next consider simulations for multivariate spatial and univariate space-time data. 

\subsection{Multivariate spatial simulation}

We simulate bivariate spatial data according to the squared-exponential, Cauchy, and Mat\'ern models with $d=2$. 
For each, we consider $\Re(\sigma_{12}) \in \{-0.4, 0, 0.4\}$, $\Im(\sigma_{12}) \in \{0, 0.4\}$, $\sigma_{11} = \sigma_{22} = 1$, and $\tilde{\vecuse{x}} = (1,1)^\top/\sqrt{2}$. A nugget variance of $0.1$ is included. 
We vary the sample size $n \in \{100, 200, 400, 600\}$, with colocated random uniform locations on $[0,1]\times [0,1]$. 
We take $a_1 = 12$ and $a_2 = 18$, using the cross-covariance forms in Appendix \ref{app:vary_params}, as well as $\alpha_1 = \alpha_2 = 1$ for the Cauchy and $\nu_1 = \nu_2 = 1$ for the Mat\'ern. 
Simulations were replicated $100$ times for each setting.

\begin{figure}[th]
    \centering
\includegraphics[width=0.96\linewidth]{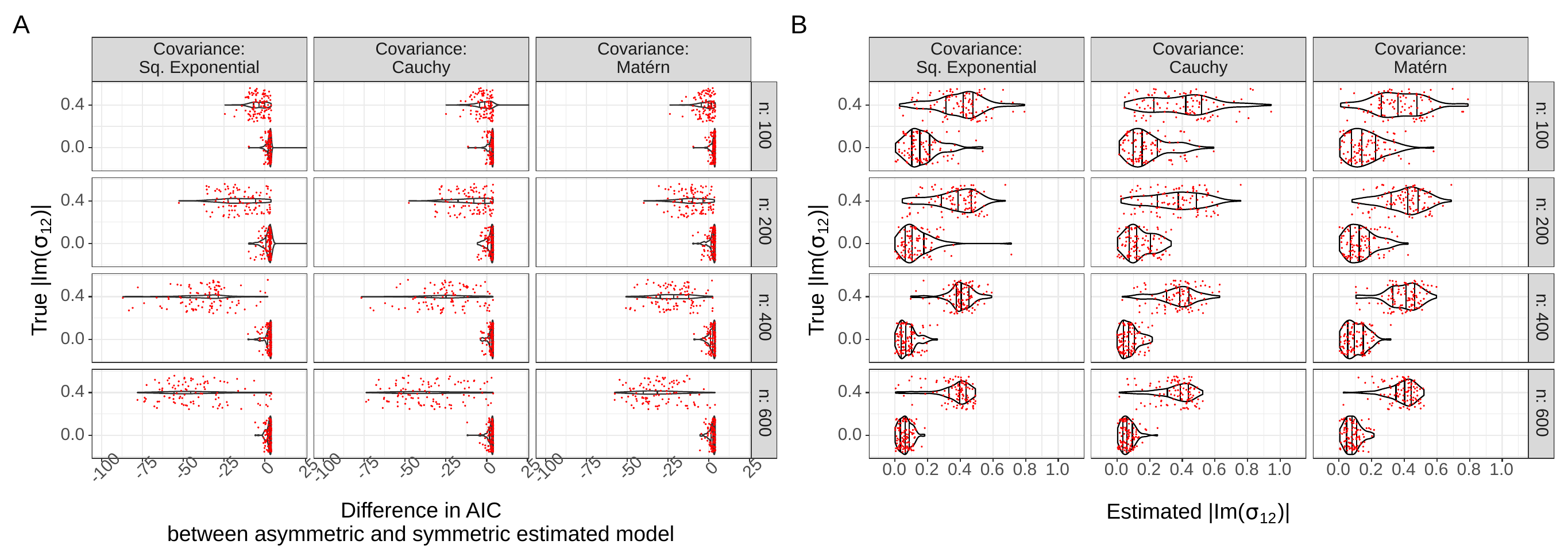}
    \caption{Multivariate spatial simulation results. Lines on violin plots represent the $0.25$, $0.50$, and $0.75$ quantiles. (A) AIC comparison of asymmetric and symmetric models when the covariance class is correctly specified and $\Re(\sigma_{12}) = 0.4$. (B) Estimates of $|\Im(\sigma_{12})|$ when the covariance class is correctly specified and $\Re(\sigma_{12}) = 0.4$. Note $\Im(\sigma_{12})$ might only be identified up to a sign. In (A) and (B), individual simulation points are plotted, and randomly jittered vertically for visualization.}
    \label{fig:mult_sim_fig}
\end{figure}

\begin{figure}[th]
    \centering
\includegraphics[width=0.96\linewidth]{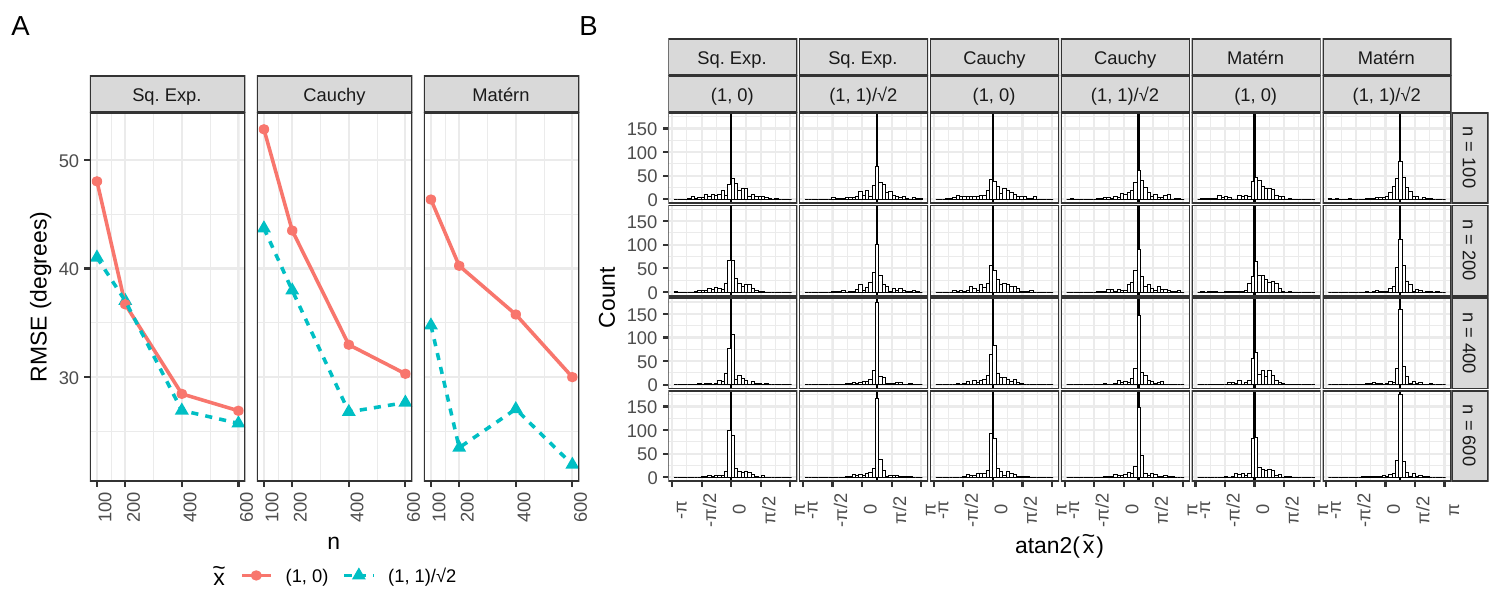}
    \caption{Asymmetry direction estimates for the multivariate spatial simulation when $\Im(\sigma_{12}) = 0.4$. (A) Root-mean-squared-error of estimates in degrees of the angle between the true $\tilde{\vecuse{x}}$ and its estimate, for a true $\tilde{\vecuse{x}} = (1,0)^\top$ or $\tilde{\vecuse{x}} = (1,1)^\top /\sqrt{2}$ and various sample sizes $n$. (B) Histogram of estimates of $\tilde{\vecuse{x}}$ converted to angle from the positive $x$ axis in radians $\textrm{atan}2(\tilde{\vecuse{x}})$. The true values of $0$ (when $\tilde{\vecuse{x}} = (1,0)^\top$) and $\pi/4$ (when $\tilde{\vecuse{x}} = (1,1)^\top/\sqrt{2}$) are plotted in vertical lines.}
    \label{fig:mult_sim_est}
\end{figure}

We estimate asymmetric and symmetric versions of each model by maximum likelihood through the L-BFGS-B algorithm \citep{byrd1995limited}.
For the Mat\'ern model, a fast Fourier transform is used on a regular grid, which is then interpolated onto the lags in the data. 
In the optimization, the asymmetry direction is estimated as $\tilde{\vecuse{x}} = \{\sin(\zeta), \cos(\zeta)\}^\top$ through the single parameter $\zeta$.
We evaluate estimation of the asymmetric cross-covariance in Fig.~\ref{fig:mult_sim_fig}. 
In general, the asymmetric model has much better fit in terms of the Akaike information criterion (AIC) when there is an asymmetric portion to the model.
For each covariance, the parameter $|\Im(\sigma_{12})|$ appears to be estimable. 
In Appendix \ref{app:sim_results}, we also plot AIC when the covariance class is misspecified and evaluate a likelihood ratio test for asymmetry. 

In Fig.~\ref{fig:mult_sim_est}, we summarize estimates of the asymmetry direction parameter. 
In particular, we evaluate estimation when $\Im(\sigma_{12}) = 0.4$ for $\tilde{\vecuse{x}} = (1,1)^\top /\sqrt{2}$ (of which all previous results are based) versus $\tilde{\vecuse{x}} = (1, 0)^\top $, with results combined from each of the $\Re(\sigma_{12}) \in \{-0.4, 0, 0.4\}$ settings.
For each estimation, we reparameterize the estimated $\tilde{\vecuse{x}}$ to $\tilde{\vecuse{x}} \textrm{sign}\{\Im(\sigma_{12})\}$. 
For most covariances and sample sizes, it is more challenging to estimate $\tilde{\vecuse{x}}$ when $\tilde{\vecuse{x}} = (1,0)^\top$. This is perhaps expected as there is more information in the direction $(1,1)^\top /\sqrt{2}$, as the domain has length $\sqrt{2}$ in this direction, but only of length $1$ in the direction of $(1,0)^\top$. 
Overall, the estimates appear reasonable and improve as $n$ increases. 

\subsection{Univariate space-time simulation}

We next consider a space-time model with spatial dimension $d=1$.
We take consider sample sizes of size $n = n_sn_t$, with $n_s \in \{20, 50, 100\}$ and $n_t =25$.
Locations in space and time are generated on $[0,1]$. 
We consider when the observations are colocated (a two-dimensional random grid of dimensions $n_s$ and $n_t$ is used) or not colocated ($n_sn_t$ observations generated uniformly and independently on $[0,1]\times [0,1]$). 

We simulate response data $Y(s,t)$ for 100 replications according to symmetric and asymmetric squared-exponential, Cauchy, exponential, and Gneiting squared-exponential (Example \ref{ex:sq_exp_gneiting_separability}) covariances. 
We take $\sigma = 1$ and $\xi \in \{0, 0.4, 0.8\}$, $a_s =10$, and $a_t = 15$. 
For the Gneiting model, we take $\tau = 1/2$ and $b = 0.7$.
In addition, we simulate according to a squared-exponential Lagrangian model \eqref{eq:lagrangian} with $\mathbb{E}({V}) = 0.5$ and $\textrm{Var}({V}) = 225$. 
This choice of $\textrm{Var}({V}) = 225$ approximately matches the choice of $a_t = 15$ considering the correspondence between Lagrangian and Gneiting covariances when $\mathbb{E}({V})=0$ \citep{salvana_daouia_lagrangian_2021}.
We consider estimation of symmetric and asymmetric versions of the models using the direct covariances without Vecchia's approximation. 

\begin{table}[th]
    \centering
        \caption{Estimation of parameters in univariate space-time simulation when the covariance is correctly-specified. Results with $n_s = 20$ are presented, with estimates presented as mean (standard deviation) of estimates over simulations; estimates for $a_s = 10$, $a_t = 15$, and $\sigma = 1$ are averaged over all $\xi \in \{0, 0.4, 0.8\}$ and colocated/not colocated settings. }
    \begin{tabular}{|c|c|c|c|c|c|c|}\hline
        Covariance & $a_s$ & $a_t$ & $\sigma$ & True $\xi$ & $\xi$ (colocated) & $\xi$ (not colocated)    \\ \hline 
        \multirow{3}{30pt}{Sq.~Exp.} & \multirow{3}{25pt}{10.0 (0.39)} & \multirow{3}{25pt}{15.0 (0.98)} & \multirow{3}{25pt}{1.01 (0.14)} & 0.0 & -0.032 (0.121) & -0.010 (0.091)\\ 
         &&&& 0.4 & 0.408 (0.111) & 0.402 (0.084) \\ 
         &&&& 0.8 &  0.784 (0.053) & 0.802 (0.045) \\ \hline
        \multirow{3}{25pt}{Cauchy} & \multirow{3}{25pt}{10.0 (0.85)} & \multirow{3}{25pt}{15.0 (1.16)} & \multirow{3}{25pt}{1.01 (0.20)} & 0.0 & 0.006 (0.159) & -0.005 (0.108)\\ 
         &&&& 0.4 & 0.393 (0.124) & 0.402 (0.090) \\ 
         &&&& 0.8 &  0.792 (0.066) & 0.793 (0.044) \\ \hline
         \multirow{3}{25pt}{Exp.} &  \multirow{3}{25pt}{10.2 (1.56)} & \multirow{3}{25pt}{15.1 (1.00)} & \multirow{3}{25pt}{1.00 (0.14)} & 0.0 & 0.005 (0.145) & -0.014 (0.097)\\ 
         &&&& 0.4 & 0.409 (0.128) & 0.387  (0.099) \\ 
         &&&& 0.8 &  0.808 (0.070) & 0.794 (0.060) \\ \hline
         \multirow{3}{25pt}{Sq.~Exp. Gnei.} & \multirow{3}{25pt}{10.2 (0.68)} & \multirow{3}{25pt}{14.4 (1.84)} & \multirow{3}{25pt}{0.963 (0.53)}& 0.0 & -0.004 (0.142)  & -0.004 (0.107)\\ 
         &&&& 0.4 &   0.438 (0.119) & 0.457 ( 0.086)\\ 
         &&&& 0.8 &   0.861 (0.048) &  0.867 ( 0.034)\\ \hline 
    \end{tabular}
    \label{tab:univ_st_sim}
\end{table}

\begin{figure}[!t]
    \centering
    \includegraphics[width=0.95\linewidth]{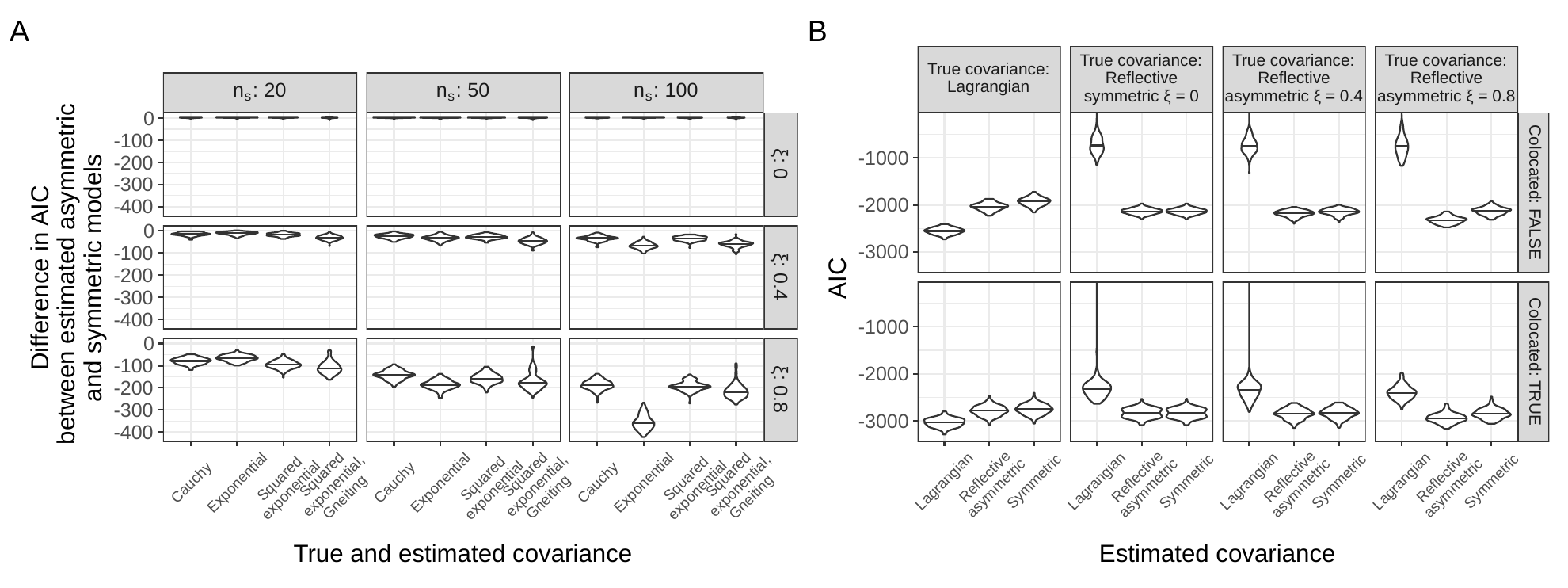}
    \caption{Univariate space-time simulation results. The lines represent medians. (A) AIC comparison of the asymmetric and symmetric models when the covariance class is correctly specified and observations are not colocated. (B) AIC comparison of the Lagrangian and reflective asymmetric squared-exponential models when $n_s = 100$. }     \label{fig:sim_uni_st}
\end{figure}

In Table~\ref{tab:univ_st_sim}, we see that in general the parameters are well-estimated. 
There is little bias for any of the parameters, including estimates of the strength of asymmetry in $\xi$. 
However, there are some notable contrasts in the standard deviation of estimates over simulations. 
First, the asymmetry parameter $\xi$ estimates have lower sampling variance when the space-time observations are observed random uniformly, not on a colocated grid. 
It is also easier to estimate $\xi$ in terms of MSE when $|\xi|$ is larger.  
The sampling variability of $\xi$ estimates is relatively consistent across covariance classes, while those of the other parameters vary more depending on the covariance structure.

In Fig.~\ref{fig:sim_uni_st}, the symmetric and asymmetric versions perform similarly in terms of AIC when the true model is symmetric, while asymmetric versions improve when there is asymmetry.
We also provide a comparison with the Lagrangian covariance. 
When the covariance is symmetric and separable ($\xi = 0$), the Lagrangian model performs substantially worse than symmetric model in terms of AIC. 
This contrasts with the reflective asymmetric model. 
When the true covariance is Lagrangian, the symmetric and reflective asymmetric covariance models perform somewhat worse than the Lagrangian model. 
These results suggest that, while incorporating asymmetry into space-time covariances is important, the type of asymmetry also plays a role in model fit. 
In Appendix \ref{app:sim_results}, we show the likelihood ratio test for asymmetry performs well when the covariance class is correctly specified.

\section{Data analysis}\label{sec:data_analysis}

\subsection{Irish wind data}

The Irish wind data, originally studied in \cite{haslett1989space}, has frequently been used as a dataset for evaluation in the space-time covariance literature \citep{gneiting_nonseparable_2002, gneiting2006geostatistical, ma2025asymmetric}. 
The data consist of 11 locations with daily wind data from 1961 to 1978.
We follow similar pre-processing as \cite{gneiting2006geostatistical}.
The years 1961--1970 are used as training data, and to create mean-zero residuals, we subtract off a location effect and 6 seasonal harmonics after applying a square root transformation. 
Locations ($d=2$) were translated from longitude and latitude to distance in kilometers.  

We aim to estimate the model by maximum likelihood using the L-BFGS-B algorithm \citep{byrd1995limited}. 
However, the dataset consists of $n_t = 3{,}652$ time points at the $n_{\vecuse{s}} = 11$ locations, with $n = n_{\vecuse{s}}n_t = 40{,}172$ not amenable to full likelihood inference. 
We thus use the Vecchia implementation mentioned in Section \ref{sec:computation}, with 30 neighbors primarily chosen using the temporal dimension. 
That is, most of $(\vecuse{s}_i,t_j)$'s neighbors consist of $\{(\vecuse{s}_{i^*}, t_j), (\vecuse{s}_{i^*}, t_j- 1), (\vecuse{s}_{i^*}, t_j + 1)\}$, for different $\vecuse{s}_{i^*}$. 
The covariances are listed in Table \ref{tab:irish_wind_compare}. 
This includes various reflective separable-type asymmetric models and their separable counterparts.
We include models in \texttt{GpGp} of \texttt{exponential\_spacetime} and \texttt{matern\_spacetime}, which take
    $C(\vecuse{h},u) = \sigma C_0\left(\sqrt{a_{\vecuse{s}}^2\lVert \vecuse{h}\rVert^2 + a_t^2u^2}\right)$, where $C_0(z)$ is an exponential or Mat\'ern covariance.
We also implement the Lagrangian model \eqref{eq:lagrangian}. 
For all models, we use the same neighbors and include a location-and-time-specific nugget with variance $\tau^2$. 

\begin{table}[!t]
    \centering
    \caption{Summaries of the estimated covariance models for the Irish wind data, with the symmetric and reflective asymmetric models presented in the first ten rows. Cauchy models had $\alpha$ either fixed or estimated. 
    $\ell(\hat{\mb{\theta}} | \vecuse{Y})$ is the maximized log-likelihood. }\label{tab:irish_wind_compare}
    \begin{tabular}{|cccrrrr|}\hline
        Cov ($\vecuse{h}$) & Cov ($u$) & Type & $\ell(\hat{\mb{\theta}} | \vecuse{Y})$ & Params  & AIC &  Time (s) \\ \hline
        Sq.~Exp. & Cauchy ($\alpha_t = 1$) & Sym &$-18256$ & 4 &36520&  8\\ 
        Sq.~Exp. & Cauchy ($\alpha_t = 1$) & Asym &$ -18092$ & 6 &36196& 40\\ 
        Sq.~Exp. & Cauchy ($\alpha_t = \frac{1}{2}$) & Sym &$-18147$ & 4 &36302 & 9\\ 
        Sq.~Exp. & Cauchy ($\alpha_t = \frac{1}{2}$) & Asym &$-17984$ & 6 &35979& 60\\  
        Sq.~Exp. & Sq.~Exp.  & Sym &$-18565$ & 4 & 37138  &  11\\ 
        Sq.~Exp. & Sq.~Exp. & Asym &$ -18417$ & 6 & 36845 & 60\\  
        Cauchy ($\alpha_{\vecuse{s}} = \frac{1}{2}$) & Cauchy ($\alpha_t = \frac{1}{2}$) & Sym &$-17375$ & 4&34758 & 7\\ 
        Cauchy ($\alpha_{\vecuse{s}} = \frac{1}{2}$) & Cauchy ($\alpha_t = \frac{1}{2}$) & Asym &$-17231$ & 6 & 34475& 42\\  
        Cauchy ($\alpha_{\vecuse{s}}$) & Cauchy ($\alpha_t$) & Sym &$-16430$ & 6 &32873&  81\\ 
        Cauchy ($\alpha_{\vecuse{s}}$) & Cauchy ($\alpha_t$) & Asym &$-16337$ & 8 &32690& 2405\\  
        \multicolumn{2}{|c}{GpGp exponential\_spacetime} & Sym &$-19323$ & 4 &38654& 20\\ 
        \multicolumn{2}{|c}{GpGp matern\_spacetime} & Sym &$-18434$ & 5 &36878& 1555\\  
        \multicolumn{2}{|c}{Lagrangian Sq.~Exp.} & Asym &$-18269$ &8& 36555& 127 \\ \hline
    \end{tabular}
\end{table}
We compare model summaries in Table \ref{tab:irish_wind_compare}. 
With the same number or fewer parameters, reflective asymmetric models with a Cauchy temporal covariance have higher log-likelihood than the Lagrangian model. 
Each reflective asymmetric model improves over its symmetric counterpart. 
The models which have both Cauchy spatial and temporal covariances outperform all other models in terms of AIC. 
In addition, many of the covariances improve over the standard \texttt{GpGp} models. 
While the asymmetric models take more time, the difference is not drastic, with exception of the Cauchy model where $\alpha_{\vecuse{s}}$ and $\alpha_t$ are estimated.
All other reflective asymmetric models are estimated more quickly than the Lagrangian model.

\begin{figure}[!t]
    \centering
    \includegraphics[width=0.96\linewidth]{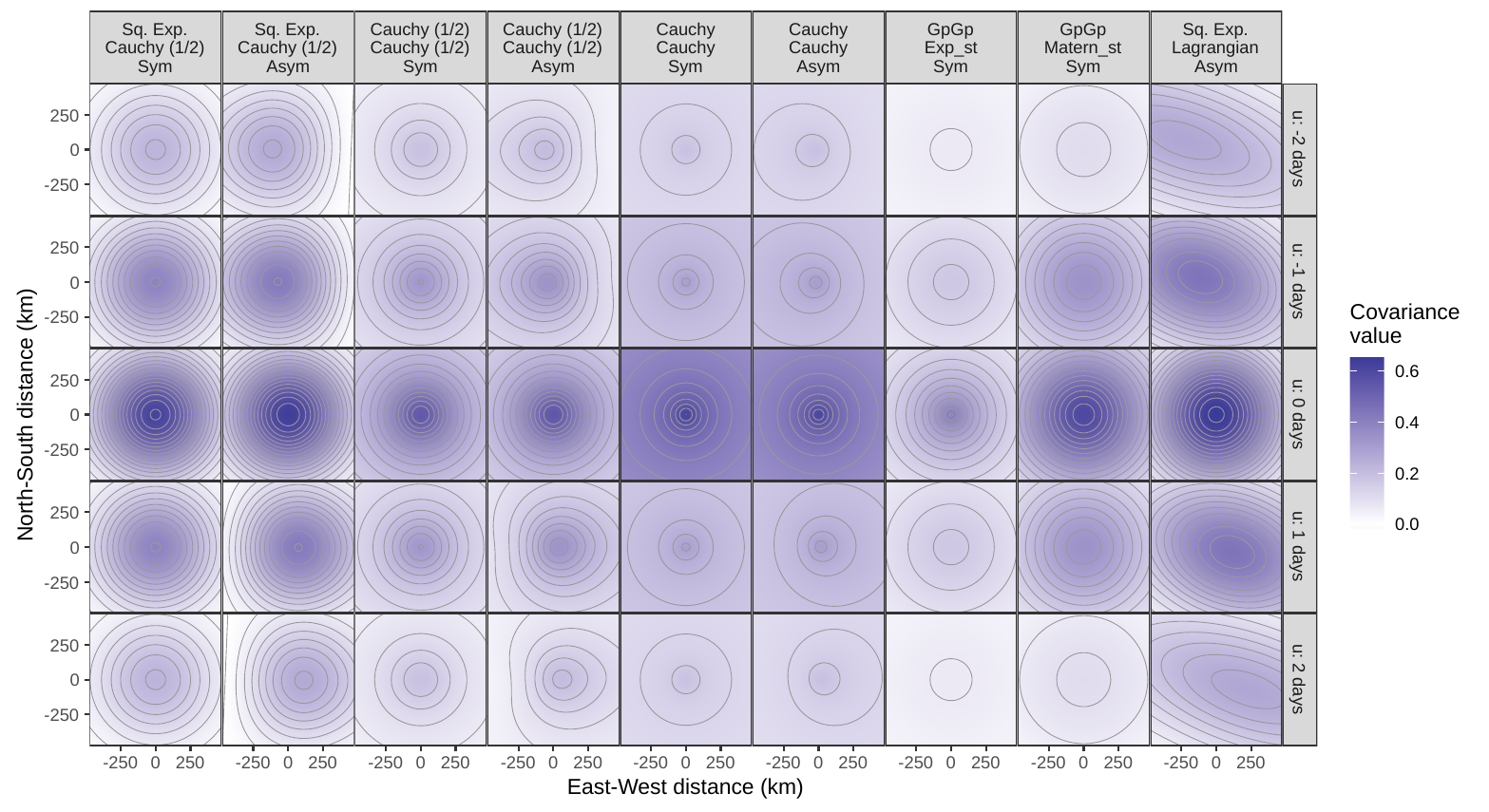}
    \caption{Selected estimated covariance functions $C(\vecuse{h},u)$ plotted for different $u$ for the Irish wind data. The top labels are the spatial covariance, temporal covariance, and if the covariance is symmetric or asymmetric, and match the order of Table \ref{tab:irish_wind_compare}. }
    \label{fig:vecchia_spat_comparison}
\end{figure}

Some selected covariance functions are compared in Fig.~\ref{fig:vecchia_spat_comparison}.
The symmetric models have $C(\vecuse{h}, u)$ maximized at $\vecuse{h} = \vecuse{0}$ for any $u$.
All asymmetric models indicate higher correlation for positive $u$ and positive East-West distance. 
This matches the expected West-to-East direction of atmospheric flow. 
The Lagrangian model has elliptical contours everywhere due to the common assumption that the advection vector $\vecuse{V}$ is multivariate normal. 
In contrast, the reflective asymmetric models have quicker decay in one direction than the other for $u\neq0$.
Flexibility to choose both the spatial and temporal covariance functions for reflective asymmetric models is illustrated by such combinations. 
The Cauchy spatial covariances decay much slower in space compared to the squared-exponential and Mat\'ern covariances. 
Quite heavy tail decay for the general Cauchy model is estimated. 
In the Appendix \ref{app:prediction}, we present the estimated parameters and a prediction comparison.

\subsection{Air pollution data}

We consider an analysis of PM$_{2.5}$ daily average surface layer estimates \citep{us_epa_equatesv10_2021}, similar to the data evaluated in \cite{park2008testing}. 
We use a subset of data of $n_{\vecuse{s}} = 221$ locations in the eastern United States from January 1st, 2009 to ($n_t = 179$) for $n = 39{,}559$ total locations.
We use a log transform and subtract a location-specific mean function of the logged data. 
In Fig.~\ref{fig:pmdata}, we give the data and residuals for the first five days of 2009. 
\begin{figure}
    \centering
    \includegraphics[width=0.95\linewidth]{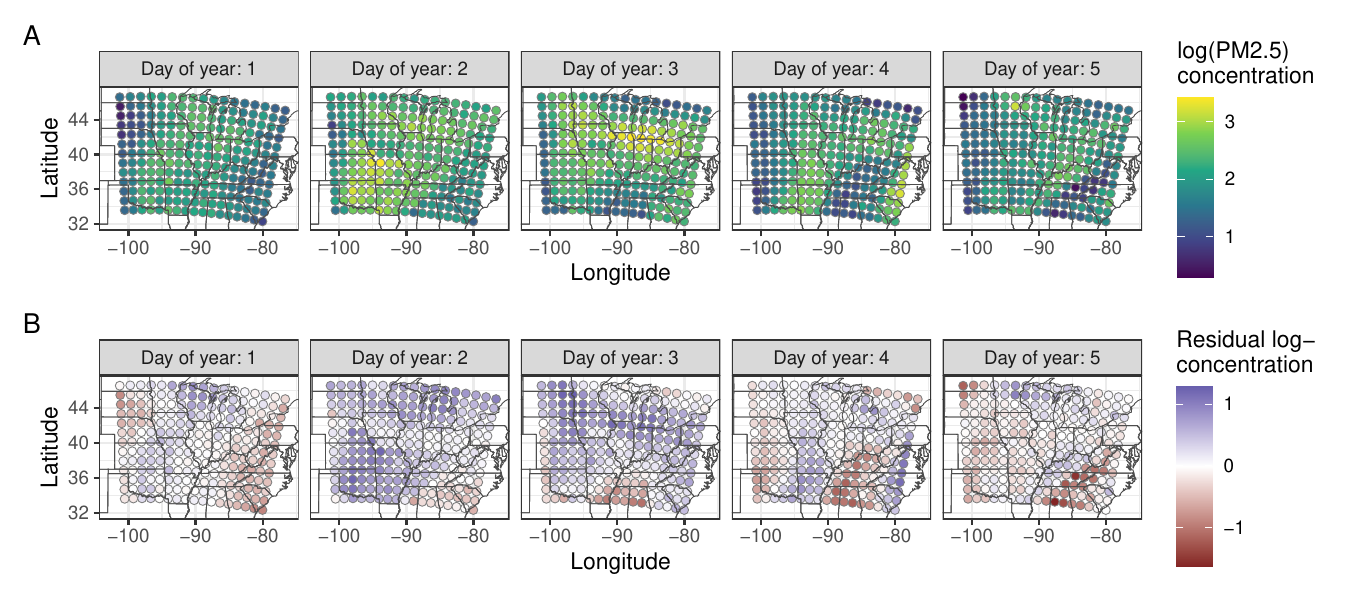}
    \caption{Example log(PM$_{2.5}$) data (A) and residuals (B) used in the first days of 2009.}
    \label{fig:pmdata}
\end{figure}
For analysis, we use $m = 15$ neighbors.
Otherwise, analysis proceeded similar to the Irish wind data. 
\begin{table}[!t]
    \centering
    \caption{Summaries of the estimated covariance models for the PM$_{2.5}$ data, with the symmetric and reflective asymmetric models presented in the first ten rows. Cauchy models had $\alpha$ either fixed or estimated. 
    $\ell(\hat{\mb{\theta}} | \vecuse{Y})$ is the maximized log-likelihood. }
    \begin{tabular}{|cccrrrr|}\hline
        Cov ($\vecuse{h}$) & Cov ($u$) & Type & $\ell(\hat{\mb{\theta}} | \vecuse{Y})$ & Params  & AIC &  Time (s)\\ \hline
        Sq.~Exp. & Cauchy ($\alpha_t = 1$) & Sym &$2342.2$ & 4 & $-4676.4 $&  31\\ 
        Sq.~Exp. & Cauchy ($\alpha_t = 1$) & Asym &$2355.3$ & 6 &$-4698.6$ & 130\\ 
        Sq.~Exp. & Cauchy ($\alpha_t = \frac{1}{2}$) & Sym &$2270.8$ & 4 &$-4533.7$ & 35\\ 
        Sq.~Exp. & Cauchy ($\alpha_t = \frac{1}{2}$) & Asym &$2282.8$ & 6 &$-4553.6$& 273\\  
        Sq.~Exp. & Sq.~Exp.  & Sym & $2321.8$ & 4 & $-4635.6$  &  42\\ 
        Sq.~Exp. & Sq.~Exp. & Asym & $2341.6$ & 6 & $-4671.7$ & 237\\  
        Cauchy ($\alpha_{\vecuse{s}} = \frac{1}{2}$) & Cauchy ($\alpha_t = \frac{1}{2}$) & Sym &$2629.2$ & 4& $-5250.4$ & 30\\ 
        Cauchy ($\alpha_{\vecuse{s}} = \frac{1}{2}$) & Cauchy ($\alpha_t = \frac{1}{2}$) & Asym &$2660.5$ & 6 & $-5309.0$& 137\\  
        Cauchy ($\alpha_{\vecuse{s}}$) & Cauchy ($\alpha_t$) & Sym &$2752.5$ & 6 &$-5493.0$&  230\\ 
        Cauchy ($\alpha_{\vecuse{s}}$) & Cauchy ($\alpha_t$) & Asym &$2777.8$ & 8 & $-5539.6$& 4502\\  
        \multicolumn{2}{|c}{GpGp exponential\_spacetime} & Sym &$1977.6$ & 4 &$-3947.2$& 21\\ 
        \multicolumn{2}{|c}{GpGp matern\_spacetime} & Sym &$2354.6$ & 5 & $-4699.3$& 4456\\  
        \multicolumn{2}{|c}{Lagrangian Sq.~Exp.} & Asym &$2320.7$ & 8 & $-4625.4$ & 385 \\ \hline
    \end{tabular}\label{tab:pm_compare}
\end{table}
We compare model summaries in Table \ref{tab:pm_compare}. 
In general, the results are similar, with lower AIC for asymmetric models and spatial Cauchy models. 
\begin{figure}[!t]
    \centering
    \includegraphics[width=0.96\linewidth]{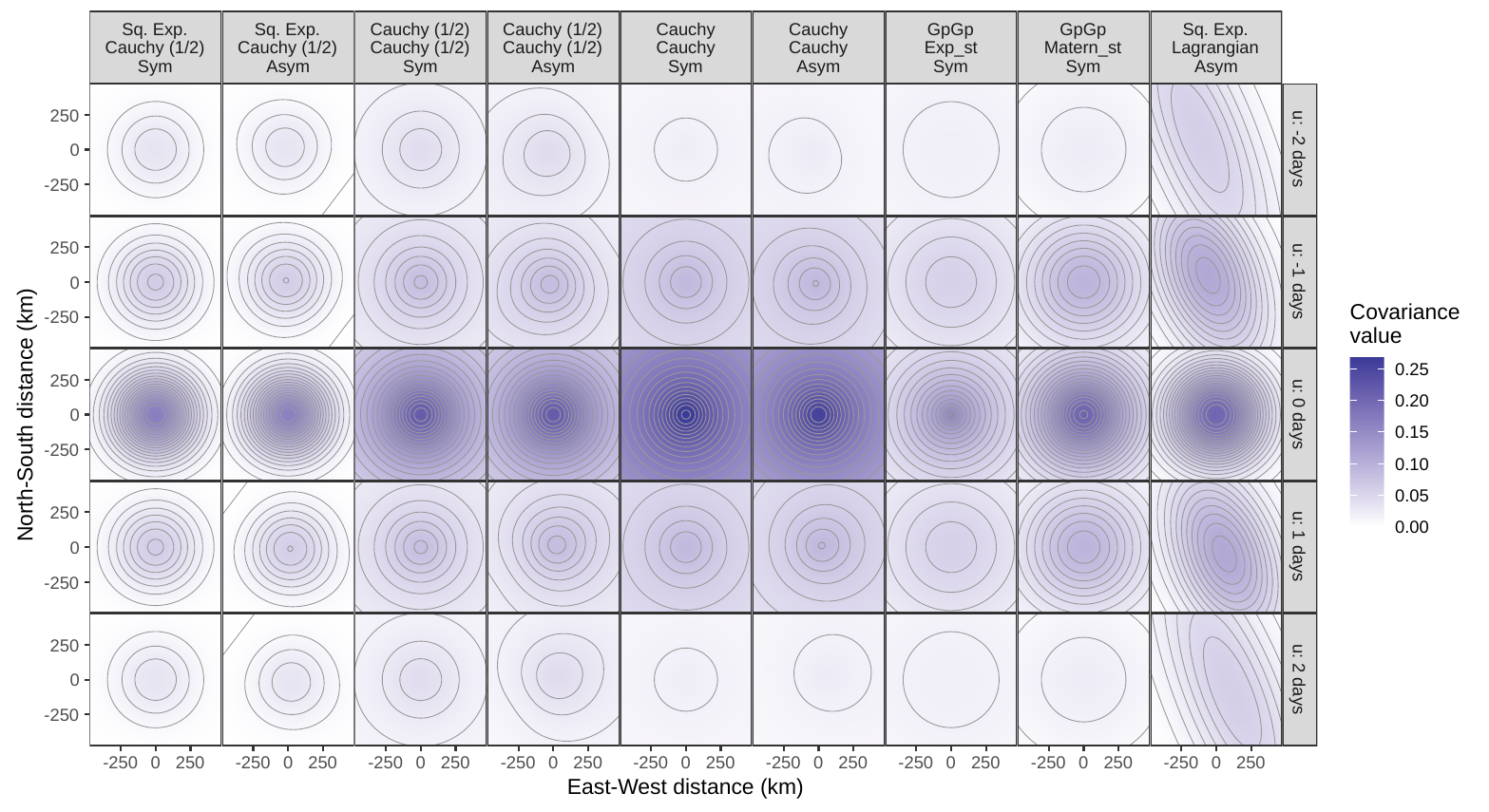}
    \caption{Selected estimated covariance functions $C(\vecuse{h},u)$ plotted for different $u$ for the PM$_{2.5}$ data. The top labels are the spatial covariance, temporal covariance, and if the covariance is symmetric or asymmetric.}
    \label{fig:vecchia_pm_comparison}
\end{figure}
In Appendix \ref{app:air}, we present the estimated parameters; in general, the asymmetry is estimated to be less strong ($\xi \approx 0.15$) compared to the Irish wind data ($\xi \approx 0.40$). 
Selected covariance functions are compared in Fig.~\ref{fig:vecchia_pm_comparison}.
While there is still general East-West direction of flow estimated, the directions are more varied. 
This may be expected when the estimated asymmetry is less strong, and the strength of asymmetry in the plots is less notable. 
Overall estimated covariance behavior is otherwise similar to those for the Irish wind data. 

\section{Discussion}\label{sec:discussion}

In this paper, we introduce a new approach that vastly increases the possible choices for asymmetric space-time covariances. 
The key features of the models are flexibility over marginal covariance functions, model parsimony with a minimal increase of complexity over separability, and improvements in application and computation. 
Among other features, the model makes a likelihood ratio test for symmetry straightforward. 
Throughout, we make contrasts with Lagrangian covariances that indicate a substantially different framework. 
The Lagrangian framework may be preferable when physical interpretation of asymmetric space-time transport is meaningful (for example, in kilometers per day). 
Furthermore, the Lagrangian framework has the opportunity for more parameters and flexibility, especially in the variance of the advection vector $\vecuse{V}$, which might be appropriate when the dissipation of the process over space and time is highly anisotropic. 
However, estimation may be more challenging in this case, as \cite{salvana2023spatio} relies on multi-step estimation. 
Reflective space-time covariances would be preferable when flexible control over the marginal temporal covariance is desired, the separable model is a submodel of interest, or one wishes to capture basic asymmetric features stably. More complex anisotropy can also be built into the reflective asymmetric models, though this will also affect the marginal spatial covariance.
We provide implementations both within and outside of the Vecchia approximation framework to make application to various space-time data feasible. 
As in \cite{yarger2023multivariate}, a limitation is the lack of closed form Mat\'ern asymmetric cross-covariances when $d > 1$. 
Addressing this and Gneiting-type asymmetric models for $d > 1$ would expand applicability of the approach.

A number of additional directions for methodological advancement are opened by this work. 
For example, the extension to where the ``temporal'' dimension is Euclidean with dimension $\vecuse{u} \in \mathbb{R}^{d^*}$ and $d^* > 1$ is straightforward in this framework by using an appropriate choice of $\left\{C^{\Re}_{time}(\vecuse{u}), C^{\Im}_{time}(\vecuse{u})\right\}$. 
In addition, one could potentially combine these models with the Lagrangian approach for further flexibility.
The extension to multivariate space-time models could introduce new asymmetries in space-time cross-covariances \citep{bourotte_flexible_2016}.

Future research can evaluate these models in a large-scale data applications where handling nonstationarity and anisotropy is necessary and scientifically relevant. 
As the Irish wind data only included 11 locations, including nonstationary flexibility in space may not be sensible. 
However, an extensive data analysis of oceanographic data \citep{kuusela_locally_2018} could make full use of the nonstationary models.

\section*{Acknowledgments}
We gratefully thank reviewers for their feedback that greatly improved the presentation and quality of the paper. 





\bibliographystyle{abbrvnat}
\bibliography{mGP}

@book{garcia1985weighted,
  title={Weighted Norm Inequalities and Related Topics},
  author={Garc{\'\i}a-Cuerva, Jos{\'e}},
  volume={116},
  year={1985},
  publisher={Elsevier}
}

@article{felsberg2002monogenic,
  title={The monogenic signal},
  author={Felsberg, Michael and Sommer, Gerald},
  journal={IEEE Transactions on Signal Processing},
  volume={49},
  number={12},
  pages={3136--3144},
  year={2002},
  publisher={IEEE}
}

@book{granlund2013signal,
  title={Signal Processing for Computer Vision},
  author={Granlund, G{\"o}sta H and Knutsson, Hans},
  year={2013},
  publisher={Springer Science \& Business Media}
}

@book{abramowitz1948handbook,
  title={Handbook of Mathematical Functions with Formulas, Graphs, and Mathematical Tables},
  author={Abramowitz, Milton and Stegun, Irene A},
  volume={55},
  year={1948},
  publisher={US Government printing office}
}

@book{gradshteyn2014table,
  title={Table of Integrals, Series, and Products},
  author={Gradshteyn, Izrail Solomonovich and Ryzhik, Iosif Moiseevich},
  year={2014},
  publisher={Academic Press}
}

@article{allard2025modeling,
  title={Modeling and simulating spatio-temporal, multivariate and nonstationary {Gaussian} {Random} {Fields}: {A} {Gaussian} mixtures perspective},
  journal={hal-05034982},
  author={Allard, Denis and Benoit, Lionel and Obakrim, Said},
  year={2025}
}

@article{cuevas_contours_2017,
	title = {Contours and dimple for the {Gneiting} class of space-time correlation functions},
	volume = {104},
	language = {en},
	number = {4},
	urldate = {2025-09-07},
	journal = {Biometrika},
	author = {Cuevas, F and Porcu, E and Bevilacqua, M},
	year = {2017},
	pages = {995--1001},
}

@article{ezzat2018spatio,
  title={Spatio-temporal asymmetry of local wind fields and its impact on short-term wind forecasting},
  author={Ezzat, Ahmed Aziz and Jun, Mikyoung and Ding, Yu},
  journal={IEEE Transactions on Sustainable Energy},
  volume={9},
  number={3},
  pages={1437--1447},
  year={2018},
  publisher={IEEE}
}

@article{park2006new,
  title={New classes of asymmetric spatial-temporal covariance models},
  author={Park, Man Sik and Fuentes, Montserrat},
  journal={North Carolina State University, Department of Statistics, Tech. Rep},
  volume={2584},
  year={2006}
}

@article{park2008testing,
  title={Testing lack of symmetry in spatial-temporal processes},
  author={Park, Man Sik and Fuentes, Montserrat},
  journal={Journal of Statistical Planning and Inference},
  volume={138},
  number={10},
  pages={2847--2866},
  year={2008},
  publisher={Elsevier}
}

@article{salvana2023spatio,
  title={Spatio-temporal cross-covariance functions under the {Lagrangian} framework with multiple advections},
  author={Salva{\~n}a, Mary Lai O and Lenzi, Amanda and Genton, Marc G},
  journal={Journal of the American Statistical Association},
  volume={118},
  number={544},
  pages={2746--2761},
  year={2023},
  publisher={Taylor \& Francis}
}

@article{stein2005space,
  title={Space-time covariance functions},
  author={Stein, Michael L},
  journal={Journal of the American Statistical Association},
  volume={100},
  number={469},
  pages={310--321},
  year={2005},
  publisher={Taylor \& Francis}
}

@Manual{expint,
    title = {expint: Exponential Integral and Incomplete Gamma Function},
    author = {Vincent Goulet},
    year = {2016},
    note = {R package version 0.1-8},
  }

@Manual{rcpp_faddeeva,
    title = {RcppFaddeeva: `Rcpp' Bindings for the `Faddeeva' Package},
    author = {Baptiste Auguié and Dirk Eddelbuettel and Steven G. Johnson},
    year = {2022},
    note = {R package version 0.2.3},
  }

@article{gsl_R,
    title = {Special functions in {R}: {Introducing} the gsl package},
    author = {Robin K. S. Hankin},
    journal = {R News},
    year = {2006},
    month = {October},
    volume = {6},
    issue = {4},
pages = {24-26}
  }

@misc{galassi2002gnu,
  title={{GNU} Scientific Library Reference Manual (3rd Ed.)},
  author={Galassi, Mark and Davies, Jim and Theiler, James and Gough, Brian and Jungman, Gerard and Alken, Patrick and Booth, Michael and Rossi, Fabrice and Ulerich, Rhys},
  year={2002}
}

@article{gneiting_strictly_2007,
	title = {Strictly Proper Scoring Rules, Prediction, and Estimation},
	volume = {102},
	doi = {10.1198/016214506000001437},
	language = {en},
	number = {477},
	journal = {Journal of the American Statistical Association},
	author = {Gneiting, Tilmann and Raftery, Adrian E},
	year = {2007},
	pages = {359--378},
}

@article{zhang_non-fully_nodate,
	title = {Non-fully symmetric space-time {Mat{\'e}rn} covariance functions},
	language = {en},
	author = {Zhang, Tonglin and Zhang, Hao},
year = {2017},
}

@article{lim_gaussian_2009,
	title = {Gaussian fields and {Gaussian} sheets with generalized {Cauchy} covariance structure},
	volume = {119},
	language = {en},
	number = {4},
	urldate = {2025-12-21},
	journal = {Stochastic Processes and their Applications},
	author = {Lim, S.C. and Teo, L.P.},
	year = {2009},
	pages = {1325--1356},
}

@article{byrd1995limited,
  title={A limited memory algorithm for bound constrained optimization},
  author={Byrd, Richard H and Lu, Peihuang and Nocedal, Jorge and Zhu, Ciyou},
  journal={SIAM Journal on Scientific Computing},
  volume={16},
  number={5},
  pages={1190--1208},
  year={1995},
  publisher={SIAM}
}

@misc{faddeeva, 
author = {Johnson, S. G.},
title = {Faddeeva Package},
year = {2025},
note = {C++},
url = {http://ab-initio.mit.edu/wiki/index.php/Faddeeva\_Package}}

@misc{GpGp, 
author = {Guinness, J and Katzfuss, M and Fahmy, Y},
year = {2021},
title = {{GpGp}: {Fast} {Gaussian} process computation using {Vecchia}’s approximation.},
note = {R package version 0.4.0}}

@article{idir_improving_2025,
	title = {Improving spatio-temporal {Gaussian} {Process} modeling with {Vecchia} approximation: {A} low-cost sensor-driven approach to urban environmental monitoring},
	doi = {10.48550/arXiv.2511.22500},
	language = {en},
	urldate = {2025-12-01},
	publisher = {arXiv},
	author = {Idir, Yacine Mohamed and Orfila, Olivier and Chatellier, Patrice and Judalet, Vincent},
	year = {2025},
}

@article{ma2025asymmetric,
  title={Asymmetric space-time covariance functions via hierarchical mixtures},
  author={Ma, Pulong},
  note={arXiv:2511.07959},
  year={2025}
}

@incollection{salvana_daouia_lagrangian_2021,
	title = {Lagrangian spatio-temporal nonstationary covariance functions},
	language = {en},
	urldate = {2025-09-16},
	booktitle = {Advances in {Contemporary} {Statistics} and {Econometrics}},
	author = {Salva{\~n}a, Mary Lai O. and Genton, Marc G.},
	editor = {Daouia, Abdelaati and Ruiz-Gazen, Anne},
	year = {2021},
publisher = {Springer}
}

@article{noauthor_corrections_2024,
    author = {Zhang, Xiran and Salvaña,  Mary Lai O. and Lenzi, Amanda and Genton, Marc G.},
	title = {Corrections to “{Spatio}-temporal cross-covariance functions under the {Lagrangian} framework with multiple advections”},
	volume = {119},
	doi = {10.1080/01621459.2024.2412190},
	language = {en},
	number = {548},
	urldate = {2025-11-04},
	journal = {Journal of the American Statistical Association},
	year = {2024},
	pages = {3189--3189},
}

@article{vu_modeling_2023,
	title = {Modeling nonstationary and asymmetric multivariate spatial covariances via deformations},
	doi = {10.5705/ss.202020.0156},
	language = {en},
	urldate = {2024-09-16},
	journal = {Statistica Sinica},
	author = {Vu, Quan and Zammit-Mangion, Andrew and Cressie, Noel},
  volume={32},
  number={4},
  pages={2071--2093},
  year={2022},
  publisher={JSTOR}
}

@article{haslett1989space,
  title={Space-time modelling with long-memory dependence: Assessing {Ireland's} wind power resource},
  author={Haslett, John and Raftery, Adrian E},
  journal={Journal of the Royal Statistical Society: Series C (Applied Statistics)},
  volume={38},
  number={1},
  pages={1--21},
  year={1989},
  publisher={Wiley Online Library}
}

@article{fonseca2011general,
  title={A general class of nonseparable space-time covariance models},
  author={Fonseca, Tha{\'\i}s CO and Steel, Mark FJ},
  journal={Environmetrics},
  volume={22},
  number={2},
  pages={224--242},
  year={2011},
  publisher={Wiley Online Library}
}

@article{horrell2017half,
  title={Half-spectral space-time covariance models},
  author={Horrell, Michael T and Stein, Michael L},
  journal={Spatial Statistics},
  volume={19},
  pages={90--100},
  year={2017},
  publisher={Elsevier}
}

@book{korotkov2020integrals,
  title={Integrals Related to the Error Function},
  author={Korotkov, Nikolai E and Korotkov, Alexander N},
  year={2020},
  publisher={Chapman and Hall/CRC}
}

@misc{us_epa_equatesv10_2021,
	title = {{EQUATESv1}.0: {Emissions}, {WRF}/{MCIP}, {CMAQv5}.3.2 {Data} -- 2002-2019 {US}\_12km and {NHEMI}\_108km},
	copyright = {, Creative Commons Zero v1.0 Universal},
	shorttitle = {{EQUATESv1}.0},
	doi = {10.15139/S3/F2KJSK},
	note = {10.15139/S3/F2KJSK},
	urldate = {2026-03-02},
	publisher = {UNC Dataverse},
	author = {{US EPA}},
	collaborator = {Foley, Kristen},
	year = {2021},
}

@article{jin_bag_2024,
	title = {Bag of {DAGs}: {Inferring} directional dependence in spatiotemporal processes},
	shorttitle = {Bag of {DAGs}},
	language = {en},
	urldate = {2026-02-24},
	journal = {Bayesian Analysis},
	author = {Jin, Bora and Peruzzi, Michele and Dunson, David},
	year = {2024},
}

@article{katzfuss_vecchia_2020,
	title = {Vecchia approximations of {Gaussian}-{Process} predictions},
	volume = {25},
	language = {en},
	number = {3},
	urldate = {2021-09-16},
	journal = {Journal of Agricultural, Biological and Environmental Statistics},
	author = {Katzfuss, Matthias and Guinness, Joseph and Gong, Wenlong and Zilber, Daniel},
	year = {2020},
	pages = {383--414},
}

@article{yarger2025multivariate,
      title={Multivariate {Confluent} {Hypergeometric} covariance functions with simultaneous flexibility over smoothness and tail decay}, 
      author={Drew Yarger and Anindya Bhadra},
      year={2025},
volume = {57},
pages = {977–1001},
issue = {5},
      journal={Mathematical Geosciences},
}

@article{mu2024gaussian,
  title={Gaussian {Processes} for time series with lead-lag effects with applications to biology data},
  author={Mu, Wancen and Chen, Jiawen and Davis, Eric S and Reed, Kathleen and Phanstiel, Douglas and Love, Michael I and Li, Didong},
  journal={Biometrics},
  volume={81},
  number={1},
  pages={ujae156},
  year={2025},
  publisher={Oxford University Press}
}

@article{genton15,
author = {Marc G. Genton and William Kleiber},
title = {Cross-covariance functions for multivariate geostatistics},
volume = {30},
journal = {Statistical Science},
number = {2},
publisher = {Institute of Mathematical Statistics},
pages = {147--163},
year = {2015},
doi = {10.1214/14-STS487},
}

@article{vecchia1988estimation,
  title={Estimation and model identification for continuous spatial processes},
  author={Vecchia, Aldo V},
  journal={Journal of the Royal Statistical Society Series B: Statistical Methodology},
  volume={50},
  number={2},
  pages={297--312},
  year={1988},
  publisher={Oxford University Press}
}

@article{arroyo_algorithm_2021,
	title = {Algorithm 1013: {An} {R} Implementation of a Continuous Spectral Algorithm for Simulating Vector {Gaussian} Random Fields in {Euclidean} Spaces},
	volume = {47},
	language = {en},
	number = {1},
	urldate = {2024-02-18},
	journal = {ACM Transactions on Mathematical Software},
	author = {Arroyo, Daisy and Emery, Xavier},
	year = {2021},
	pages = {1--25},
}

@article{emery2022new,
  title={New validity conditions for the multivariate {M}at{\'e}rn coregionalization model, with an application to exploration geochemistry},
  author={Emery, Xavier and Porcu, Emilio and White, Philip},
  journal={Mathematical Geosciences},
  volume={54},
  number={6},
  pages={1043--1068},
  year={2022},
  publisher={Springer}
}

@article{ma2022beyond,
  title={Beyond {Mat{\'e}rn}: {On} a class of interpretable {Confluent} {Hypergeometric} covariance functions},
  author={Ma, Pulong and Bhadra, Anindya},
  journal={Journal of the American Statistical Association},
  pages={2045-2058},
  year={2023},
volume = {118}, 
number = {543},
  publisher={Taylor \& Francis}
}

@article{li_approach_2011,
	title = {An approach to modeling asymmetric multivariate spatial covariance structures},
	volume = {102},
	number = {10},
	journal = {Journal of Multivariate Analysis},
	author = {Li, Bo and Zhang, Hao},
	year = {2011},
	pages = {1445--1453},
}

@article{emery2016improved,
  title={An improved spectral turning-bands algorithm for simulating stationary vector {Gaussian} random fields},
  author={Emery, Xavier and Arroyo, Daisy and Porcu, Emilio},
  journal={Stochastic Environmental Research and Risk Assessment},
  volume={30},
  pages={1863--1873},
  year={2016},
  publisher={Springer}
}

@misc{NIST:DLMF,
         key = "{\relax DLMF}",
       title = "{\it NIST Digital Library of Mathematical Functions}",
howpublished = "http://dlmf.nist.gov/, Release 1.1.3 of 2021-09-15",
year = 2021,
        author = {"F.~W.~J. Olver and A.~B. {Olde Daalhuis} and D.~W. Lozier and B.~I. Schneider and
                R.~F. Boisvert and C.~W. Clark and B.~R. Miller and B.~V. Saunders and
                H.~S. Cohl and M.~A. McClain"}}

@article{schlather2010some,
  title={Some covariance models based on normal scale mixtures},
  author={Schlather, Martin},
  year={2010}
}

@incollection{gneiting2006geostatistical,
	title = {Geostatistical {space}-{time} {models}, {stationarity}, {separability}, and {full} {symmetry}},
	volume = {107},
	language = {en},
	urldate = {2025-09-17},
	booktitle = {Statistical {Methods} for {Spatio}-{Temporal} {Systems}},
	publisher = {Chapman and Hall/CRC},
	author = {Gneiting, Tilmann and Genton, Marc and Guttorp, Peter},
	editor = {Finkenst{\"a}dt, B{\"a}rbel and Held, Leonhard and Isham, Valerie},
	year = {2006},
	pages = {151--175},
}

@article{de_iaco_positive_2013,
	title = {Positive and negative non-separability for space-time covariance models},
	volume = {143},
	language = {en},
	number = {2},
	urldate = {2025-08-12},
	journal = {Journal of Statistical Planning and Inference},
	author = {De Iaco, S. and Posa, D.},
	year = {2013},
	pages = {378--391},
}

@article{mateu_recent_2008,
	title = {Recent advances to model anisotropic space-time data},
	volume = {17},
	language = {en},
	number = {2},
	urldate = {2025-09-07},
	journal = {Statistical Methods and Applications},
	author = {Mateu, J. and Porcu, E. and Gregori, P.},
	year = {2008},
	pages = {209--223},
}

@book{chiles2012geostatistics,
  title={Geostatistics: Modeling Spatial Uncertainty},
  author={Chiles, Jean-Paul and Delfiner, Pierre},
  year={2012},
  publisher={John Wiley \& Sons}
}

@article{allard_simulating_2020,
	title = {Simulating space-time random fields with nonseparable {Gneiting}-type covariance functions},
	volume = {30},
	language = {en},
	number = {5},
	journal = {Statistics and Computing},
	author = {Allard, Denis and Emery, Xavier and Lacaux, Céline and Lantuéjoul, Christian},
	year = {2020},
	pages = {1479--1495},
}

@article{faouzi_spacetime_2025,
	title = {Space-time smoothness and parsimony in covariance functions},
	language = {en},
	urldate = {2025-06-19},
	journal = {Mathematical Methods in the Applied Sciences},
	author = {Faouzi, Tarik and Porcu, Emilio and Bevilacqua, Moreno},
	year = {2025},
}

@article{bourotte_flexible_2016,
	title = {A flexible class of non-separable cross-covariance functions for multivariate space-time data},
	volume = {18},
	doi = {10.1016/j.spasta.2016.02.004},
	language = {en},
	urldate = {2025-09-07},
	journal = {Spatial Statistics},
	author = {Bourotte, Marc and Allard, Denis and Porcu, Emilio},
	year = {2016},
	pages = {125--146},
}

@phdthesis{wu2021non,
  title={Non-fully symmetric space-time Matern-Cauchy correlation functions},
  author={Wu, Zizhuang},
  year={2021},
  school={Purdue University Graduate School}
}

@article{kent2011dimple,
  title={The dimple in Gneiting's spatial-temporal covariance model},
  author={Kent, John T and Mohammadzadeh, Mohsen and Mosammam, Ali M},
  journal={Biometrika},
  pages={489--494},
  year={2011},
  publisher={Biometrika Trust}
}

@book{king2009hilbert,
  title={Hilbert {Transforms}},
  author={King, Frederick W},
  year={2009},
  publisher={Cambridge University Press}
}

@article{gneiting_nonseparable_2002,
	title = {Nonseparable, stationary covariance functions for space-time data},
	volume = {97},
	language = {en},
	number = {458},
	urldate = {2024-09-26},
	journal = {Journal of the American Statistical Association},
	author = {Gneiting, Tilmann},
	year = {2002},
	pages = {590--600},
}

@article{bevilacqua2025parsimonious,
  title={Parsimonious compactly supported covariance models in the {Gauss} {Hypergeometric} {Class}: {Identifiability}, reparameterizations, and asymptotic properties},
  author={Bevilacqua, Moreno and Caama{\~n}o-Carrillo, Christian and Faouzi, Tarik and Emery, Xavier},
    note = {Preprint, arXiv:2506.13646},

  year={2025}
}

@book{oberhettinger2012tables,
  title={Tables of Laplace Transforms},
  author={Oberhettinger, Fritz and Badii, Larry},
  year={2012},
  publisher={Springer Science \& Business Media}
}

@article{didier_integral_2011,
	title = {Integral representations and properties of operator fractional {Brownian} motions},
	volume = {17},
	language = {EN},
	number = {1},
	urldate = {2020-12-18},
	journal = {Bernoulli},
	publisher = {Bernoulli Society for Mathematical Statistics and Probability},
	author = {Didier, Gustavo and Pipiras, Vladas},
	year = {2011},
	mrnumber = {MR2797980},
	zmnumber = {1284.60079},
	pages = {1--33},
}

@book{lucarini2005kramers,
  title={Kramers-Kronig relations in optical materials research},
  author={Lucarini, Valerio and Peiponen, Kai-Erik and Saarinen, Jarkko J and Vartiainen, Erik M},
  year={2005},
  publisher={Springer}
}

@article{gabor1946theory,
  title={Theory of communication. Part 1: The analysis of information},
  author={Gabor, Dennis},
  journal={Journal of the Institution of Electrical Engineers-part III: radio and communication engineering},
  volume={93},
  number={26},
  pages={429--441},
  year={1946},
  publisher={IET}
}

@article{allard_fully_2022,
	title = {Fully nonseparable {Gneiting} covariance functions for multivariate space-time data},
	volume = {52},
	urldate = {2023-11-06},
	journal = {Spatial Statistics},
	author = {Allard, Denis and Clarotto, Lucia and Emery, Xavier},
	year = {2022},
	pages = {100706},
}

@article{chen2021space,
  title={Space-time covariance structures and models},
  author={Chen, Wanfang and Genton, Marc G and Sun, Ying},
  journal={Annual Review of Statistics and Its Application},
  volume={8},
  pages={191--215},
  year={2021},
  publisher={Annual Reviews}
}

@article{tao2004lecture2,
  title={Lecture notes 4 for 247A},
  author={Tao, Terence},
  url = {https://www.math.ucla.edu/~tao/247a.1.06f/notes4.pdf},
  volume={1},
  year={2004}
}

@article{erfi, 
    author = {Weisstein, Eric W.},
    title = {Erfi},
    year = {2026},
    journal = {MathWorld--A Wolfram Resource},
    note = {https://mathworld.wolfram.com/Erfi.html}}

@article{porcu202130,
  title={30 Years of space-time covariance functions},
  author={Porcu, Emilio and Furrer, Reinhard and Nychka, Douglas},
  journal={Wiley Interdisciplinary Reviews: Computational Statistics},
  volume={13},
  number={2},
  pages={e1512},
  year={2021},
  publisher={Wiley Online Library}
}

@article{kuusela_locally_2018,
	title = {Locally stationary spatio-temporal interpolation of {Argo} profiling float data},
	volume = {474},
	language = {en},
	number = {2220},
	urldate = {2019-10-11},
	journal = {Proceedings of the Royal Society A: Mathematical, Physical and Engineering Sciences},
	author = {Kuusela, Mikael and Stein, Michael L.},
	month = {Dec.},
	year = {2018},
}

@article{porcu2023mat,
  title={The {Mat}\'ern model: A journey through statistics, numerical analysis and machine learning},
  author={Porcu, Emilio and Bevilacqua, Moreno and Schaback, Robert and Oates, Chris J},
  journal={Statistical Science},
  volume={39},
  number={3},
  pages={469--492},
  year={2024},
  publisher={Institute of Mathematical Statistics}
}

@article{gneiting2010matern,
  title={Mat{\'e}rn cross-covariance functions for multivariate random fields},
  author={Gneiting, Tilmann and Kleiber, William and Schlather, Martin},
  journal={Journal of the American Statistical Association},
  volume={105},
  number={491},
  pages={1167--1177},
  year={2010},
  publisher={Taylor \& Francis}
}

@article{yarger2023multivariate,
      title={Multivariate {Mat}\'ern models --- {A} spectral approach}, 
      author={Drew Yarger and Stilian Stoev and Tailen Hsing},
	year = {2026},
	pages = {69--98},
	volume = {41},
	number = {1},
	journal = {Statistical Science},
	doi = {10.1214/24-STS948}
}

@article{paciorek_spatial_2006,
    title = {Spatial modelling using a new class of nonstationary covariance functions},
    volume = {17},
    doi = {10.1002/env.785},
    language = {en},
    number = {5},
    urldate = {2025-04-15},
    journal = {Environmetrics},
    author = {Paciorek, Christopher J. and Schervish, Mark J.},
    year = {2006},
    pages = {483--506},
}

@article{alegria_dimple_2017,
  title = {The dimple problem related to space–time modeling under the {Lagrangian} framework},
  volume = {162},
  language = {en},
  urldate = {2025-10-03},
  journal = {Journal of Multivariate Analysis},
  author = {Alegría, Alfredo and Porcu, Emilio},
  year = {2017},
  pages = {110--121},
}

\appendix

\renewcommand{\thefigure}{\thesection\arabic{figure}}
\renewcommand{\thetheorem}{\thesection\arabic{theorem}}
\renewcommand{\theproposition}{\thesection\arabic{proposition}}
\renewcommand{\theexample}{\thesection\arabic{example}}
\renewcommand{\thetable}{\thesection\arabic{table}}

\setcounter{theorem}{0}
\setcounter{proposition}{0}
\setcounter{example}{0}
\setcounter{table}{0}
\setcounter{figure}{0}

\section{Additional information on models}\label{app:models}

\subsection{Cross-covariances between covariances with different parameters}\label{app:vary_params}

We present cross-covariance forms for when the spectral densities of the original densities are not proportional with $f_j(\vecuse{x}) \neq f_k(\vecuse{x})$.
We begin with the squared-exponential. 
Suppose that $C_{jj}(\vecuse{h}) = \textrm{exp}\left(-a_j^2 \lVert \vecuse{h}\rVert^2\right)$ and $C_{kk}(h) = \textrm{exp}\left(-a_{k}^2 \lVert \vecuse{h}\rVert^2\right)$ for parameters $a_j$ and $a_k$. 
Let $a_{jk} = a_j a_k/\sqrt{(a_j^2 + a_k^2)/2}$. 
Then the corresponding components of the cross-covariance are:
\begin{align*}
    C_{jk}^{\Re}(\vecuse{h}) 
    &=  \left\{\frac{a_{jk}}{(a_ja_k)^\frac{1}{2}}\right\}^{d}\textrm{exp}\left(-a_{jk}^2 \lVert \vecuse{h}\rVert^2\right);  \\ 
    C_{jk}^{\Im}(\vecuse{h})     &=  \left(\frac{a_{jk}}{(a_ja_k)^\frac{1}{2}}\right)^{d}\textrm{exp}\left(-a_{jk}^2 \lVert \vecuse{h}\rVert^2\right)\textrm{erfi}\left(a_{jk}\langle \vecuse{h}, \tilde{\vecuse{x}}_{jk}\rangle\right),
\end{align*}where $\tilde{\vecuse{x}}_{jk} = \tilde{\vecuse{x}}_{kj}$.
In general, the expression $\Re(\sigma_{jk}) (a_{jk}/\sqrt{a_ja_k})^d$ now corresponds to the marginal covariance $\textrm{Cov}\left\{Y_j(\vecuse{s}), Y_k(\vecuse{s})\right\}$ between processes $j$ and $k$.

Now, consider the Cauchy covariance, where $C_{jj}(\vecuse{h}) = \left(a_{j}^2 \lVert \vecuse{h}\rVert^2 + 1\right)^{-\alpha_j}$ and $C_{kk}(\vecuse{h}) = \left(a_{k}^2 \lVert \vecuse{h}\rVert^2 + 1\right)^{-\alpha_k}$.
Suppose that $a_{jk}$ is defined as above and  $\alpha_{jk} = (\alpha_j + \alpha_k)/2$. Then valid cross-covariance components are defined by 
\begin{align*}
    C_{jk}^{\Re}(\vecuse{h}) 
    &= \frac{a_{jk}^{2\alpha_{jk}}}{a_j^{\alpha_j}a_k^{\alpha_k}}\frac{\Gamma(\alpha_{jk})}{\sqrt{\Gamma(\alpha_j)\Gamma(\alpha_k)}}\frac{1}{\left(a_{jk}^2\lVert \vecuse{h}\rVert^2 + 1\right)^{\alpha_{jk}}},
\end{align*}and \begin{align*}
        C_{jk}^{\Im}(\vecuse{h})     &= 
 C_{jk}^{\Re}(\vecuse{h})
    \frac{2}{\sqrt{\pi}}\frac{\Gamma(\alpha_{jk}+\frac{1}{2})}{\Gamma(\alpha_{jk})}\frac{a_{jk}\langle \vecuse{h}, \tilde{\vecuse{x}}_{jk}\rangle}{\left(a_{jk}^2 \lVert \vecuse{h} \rVert^2 + 1\right)^{\frac{1}{2}}} {}_2F_1\left(\frac{1}{2}, \frac{1}{2} + \alpha_{jk}; \frac{3}{2}; \frac{a_{jk}^2\langle \vecuse{h}, \tilde{\vecuse{x}}_{jk}\rangle^2}{a_{jk}^2\lVert \vecuse{h}\rVert^2 + 1}\right). 
\end{align*}
The form of $C_{jk}^\Re(\vecuse{h})$ matches the adjustment in \cite{allard2025modeling} up to reparameterization of the inverse range parameter. 

For the Mat\'ern covariance with $d=1$, we have, with $\nu_{jk} = (\nu_j + \nu_k)/2$ and $a_{jk} = \sqrt{(a_j^2 + a_k^2)/2}$, \begin{align*}
    C_{jk}^{\Re}(h) =  \frac{a_j^{\nu_j}a_k^{\nu_k}}{a_{jk}^{2\nu_{jk}}}\frac{\Gamma(\nu_{jk})}{\sqrt{\Gamma(\nu_j)\Gamma(\nu_k)}}\frac{2^{1-\nu_{jk}}}{\Gamma(\nu_{jk})} (a_{jk}|h|)^{\nu}K_\nu\left(a_{jk}|h|\right),
\end{align*}and, for $\nu_{jk}\notin \{1/2, 3/2, \dots\}$,
\begin{align*}
    C_{jk}^{\Im}(h) &=  \frac{a_j^{\nu_j}a_k^{\nu_k}}{a_{jk}^{2\nu_{jk}}}\frac{\Gamma(\nu_{jk})}{\sqrt{\Gamma(\nu_j)\Gamma(\nu_k)}} \frac{\pi \textrm{sign}(h)}{2^{\nu_{jk}}\Gamma(\nu_{jk})\cos(\pi\nu_{jk})}
    \\
    &~~~~~~~\times(a_{jk}|h|)^{\nu_{jk}}\left\{I_{\nu_{jk}}(a_{jk}|h|) - L_{-\nu_{jk}}(a_{jk}|h|)\right\}.
\end{align*}
For the exponential, $d=1$, where $\nu_j = \nu_k = 1/2$ and $a_{jk} = \sqrt{(a_j^2 + a_k^2)/2}$, we have \begin{align*}
    C_{jk}^{\Re}(h)&=\frac{\sqrt{a_ja_k}}{a_{jk}}\textrm{exp}\left( - a_{jk}|h|\right),
\end{align*} and\begin{align*}
    C_{jk}^{\Im}(h)&=\frac{\sqrt{a_ja_k}}{a_{jk}}
     \left\{\textrm{sign}(h)/\pi\right\}\{\textrm{exp}(a_{jk}|h|)E_1(a_{jk}|h|)+ \textrm{exp}(-a_{jk}|h|)\textrm{Ei}(a_{jk}|h|)\}.
\end{align*}

\subsection{Gneiting models}\label{app:gneiting}

Support of the asymmetric squared-exponential Gneiting covariance in the paper relies on the following proposition, which is shown later in this document.  

\begin{proposition}[$d=1$ squared-exponential]\normalfont\label{prop:sq_exp_asymmetric_gneiting}

Consider the Gneiting model without separability parameter with $d=1$ and $\varphi(a_s\lVert \vecuse{h}\rVert^2) = \textrm{exp}\left(-a_{\vecuse{s}}^2 \lVert \vecuse{h}\rVert^2\right)$. The asymmetric squared-exponential covariance where $\gamma(u) = a_t^2u^2$ is:
\begin{align*}
        C(h,u)&= \sigma \frac{1}{\sqrt{a_t^2u^2 + 1}} \textrm{exp}\left(-\frac{a_s^2 h^2}{a_t^2u^2+1}\right) \left\{1 + \xi \textrm{erf}\left(\frac{a_sh}{\sqrt{a_t^2u^2+1}}a_tu\right) \right\}.
\end{align*}

\end{proposition}

We next present the Gneiting-Cauchy model for general $\alpha$.

\begin{proposition}[$d=1$ Cauchy]\normalfont\label{prop:gneiting_cauchy}

In the setting with $d=1$, the corresponding Cauchy space-time Gneiting covariance with $\gamma(u) = a_t^2u^2$ can be written in terms of the hypergeometric function for general $\alpha $:
    \begin{align*}
  C(h,u) 
     &= \frac{\sigma}{\sqrt{a_t^2u^2 + 1}}\frac{1}{\left(\frac{a_s^2h^2}{a_t^2u^2+1 } + 1\right)^\alpha} \\
 &~~~\times \left\{1 + \xi\textrm{sign}(h)\textrm{sign}(u)\frac{2}{\sqrt{\pi}}{}_2F_1\left(\frac{1}{2}, \alpha; \frac{3}{2}; -\frac{a_t^2u^2 a_s^2h^2}{a_s^2h^2+  a_t^2u^2 + 1}\right)\right\}.
    \end{align*}
\end{proposition}
Once again, simplified expressions for the special cases $\alpha = 1/2$ and $\alpha = 1$ below are derived straightforwardly from special cases of the hypergeometric function ${}_2F_1$ \citep[Section 15.4,][]{NIST:DLMF}.
These forms extend a few Gneiting covariances with only one additional parameter $\xi$. 
Explicit forms for $d > 1$ require further study. 

\begin{example}[$\alpha = 1/2$, $d=1$ Cauchy]\normalfont
Let $h_*= h/\sqrt{a_t^2u^2 + 1}$. The Cauchy-Gneiting nonseparable covariance with $\alpha = 1/2$ is 
  \begin{align*}
    C(h,u)&=\frac{\sigma}{\sqrt{a_t^2u^2 + 1}} \frac{1}{\sqrt{a_s^2h_*^2 + 1}}\left\{1 + \xi \frac{2}{\sqrt{\pi}} \frac{\sqrt{a_s^2h_*^2 +  1}}{a_tu a_sh_*}\sinh^{-1}\left(\frac{a_t^2u^2 a_s^2h_*^2}{a_s^2h_*^2+ 1}\right) \right\}
    \end{align*}if $ h\neq 0$ and $u \neq 0$. If $h =0$ or $u = 0$, we have  $C(h,u) =\sigma/\left(\sqrt{a_t^2u^2 + 1}\sqrt{a_s^2h_*^2 + 1}\right)$.
\end{example}

\begin{example}[$\alpha = 1$, $d=1$, Cauchy]\normalfont
The special case with $\alpha = 1$ is 
        \begin{align*}
  C(h,u) 
  &= \frac{\sigma}{\sqrt{a_t^2u^2 + 1}} \frac{1}{\left(\frac{a_s^2h^2}{a_t^2 u^2 +1} + 1\right)} \\
  &~~~\times\left\{1 + \xi \frac{2}{\sqrt{\pi}}
  \frac{\sqrt{a_s^2h^2 + a_t^2u^2 + 1}}{a_tu a_s h}
  \tan^{-1}\left(\frac{a_t^2u^2 a_s^2h^2}{a_s^2h^2+  a_t^2u^2 + 1}\right)\right\}. 
    \end{align*}
\end{example}

\begin{example}[$\alpha = 1/2$, $d=1$ Cauchy with separability parameter]\normalfont
Let $b\in[0,1]$, $\delta > 0$, and $\tau = b/2 + \delta > b/2$. 
Then we have the following valid space-time covariance, where now $h_* = h / (a_t^2u^2+1)^{b/2}$: 
        \begin{align*}
  C(h,u)
  &= \frac{\sigma}{\left(a_t^2u^2 + 1\right)^{\tau}} \frac{1}{\sqrt{a_s^2h_*^2 + 1}}\left\{1 + \xi \frac{2}{\sqrt{\pi}} \frac{\sqrt{a_s^2h_*^2 +  1}}{a_tu a_sh_*}\sinh^{-1}\left(\frac{a_t^2u^2 a_s^2h_*^2}{a_s^2h_*^2+ 1}\right) \right\}
    \end{align*}if $ h_*\neq 0$ and $u \neq 0$. Otherwise, $C(h,u) =\sigma/\left\{\left(a_t^2u^2 + 1\right)^{\tau}\sqrt{a_s^2h_*^2 + 1}\right\}$.
\end{example}
Similar forms may be constructed for other $\alpha$.

\subsection{Squared-exponential Gneiting half-spectral density}\label{app:half_spec}

We present half-spectral \citep{horrell2017half} densities, in which the time dimension of the spectral density is integrated out. 
This proposition is similar to Lemma 1 of \cite{allard_simulating_2020} for the Gneiting covariance and now includes an asymmetric part. 

\begin{proposition}[Half-spectral representation of Gneiting squared-exponential covariance]\normalfont\label{prop:half_spectral}

Consider the squared-exponential covariance $\varphi\left(a_{\vecuse{s}}^2\lVert \vecuse{h}\rVert^2\right) = \textrm{exp}\left(-a_{\vecuse{s}}^2\lVert \vecuse{h} \rVert^2\right)$.
For $\gamma(u) = a_t^2 u^2$, the Gneiting space-time covariance with additional asymmetric part for the squared-exponential has representation
\begin{align*}
C(\vecuse{h},u)&= \sigma\int_{\mathbb{R}^d} \frac{\textrm{exp}\left(\I \langle \vecuse{h}, \vecuse{x}\rangle\right)}{\left(2a_{\vecuse{s}}\sqrt{\pi}\right)^d} \textrm{exp}\left\{- \frac{\lVert \vecuse{x}\rVert^2}{4a_{\vecuse{s}}^2}\left(a_t^2 u^2 + 1\right)\right\} \\
&~~~~~~\times \left\{1+ \xi \textrm{sign}(\langle \vecuse{x}, \tilde{\vecuse{x}}\rangle)\textrm{erfi}\left(\frac{\lVert \vecuse{x}\rVert}{2a_{\vecuse{s}}}a_t u\right)\right\}d\vecuse{x}.
\end{align*}

For $\gamma(u) = a_t|u|$, it instead has representation\begin{align*}
        C(\vecuse{h},u)&= \sigma\int_{\mathbb{R}^d}\frac{\textrm{exp}\left(\I \langle \vecuse{h}, \vecuse{x}\rangle\right)}{\left(2a_{\vecuse{s}}\sqrt{\pi}\right)^d} \textrm{exp}\left\{- \frac{\lVert \vecuse{x}\rVert^2}{4a_{\vecuse{s}}^2}\left(a_t |u| + 1\right)\right\}\left\{1+ \xi \textrm{sign}(\langle \vecuse{x}, \tilde{\vecuse{x}}\rangle)Z(\vecuse{x}, u)\right\}d\vecuse{x},\\
           Z(\vecuse{x}, u) &= \frac{\textrm{sign}(u)}{\pi }\Bigg\{\textrm{exp}\left(-\frac{\lVert \vecuse{x}\rVert^2 }{4a_{\vecuse{s}}^2}a_t|u|\right)\textrm{Ei}\left(\frac{\lVert \vecuse{x}\rVert^2}{4a_{\vecuse{s}}^2}a_t|u|\right) \\
           &~~~~~~~~~~~~~~~~~~- \textrm{exp}\left(\frac{\lVert \vecuse{x}\rVert^2}{4a_{\vecuse{s}}^2}a_t|u|\right)E_1\left(\frac{\lVert \vecuse{x}\rVert^2}{4a_{\vecuse{s}}^2}a_t|u|\right)\Bigg\},
\end{align*}where parts of $Z(\vecuse{x}, u)$ correspond to the asymmetric exponential cross-covariance.

\end{proposition}

In general, for any $\gamma(u) = (a_t|u|)^{\beta}$ for $0 < \beta \leq 2$, the new asymmetric portion of the half-spectral representation will be related to the Hilbert transform of the powered-exponential class $\textrm{exp}\left\{-(a_t|u|)^\beta\right\}$, with the above corresponding to $\beta = 2$ and $\beta = 1$, respectively. 
The half-spectral representation makes computation of the covariance straightforward through fast Fourier transforms even if a closed-form representation is not available. 
Recently, constructions have often focused on this form of $\gamma(u) = (a_t|u|)^{\beta}$ \citep{allard2025modeling}. 
\subsection{Additional model plot}\label{app:plots}

In Fig.~\ref{fig:vary_u_h}, we plot cross-sections of separable-type asymmetric covariance functions $C(h,u)$ in $d=1$.
\begin{figure}[th]
    \centering
    \includegraphics[width=0.98\linewidth]{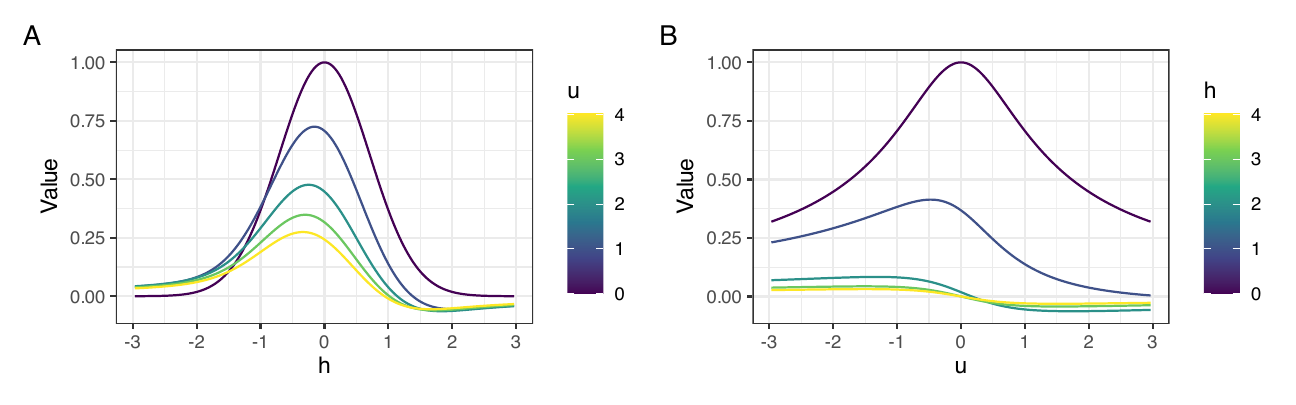}
    \caption{Cross-sections of the separable-type covariances for $d=1$ using the squared-exponential (spatial) and Cauchy (temporal) model with $a_s = a_t = 1$, $\alpha = 1$, $\sigma = 1$, and $\xi = -0.9$. (A) Covariance as a function of $h$ for varying $u$; (B) covariance as a function of $u$ for varying $h$.}
    \label{fig:vary_u_h}
\end{figure}

\section{Additional results}

\subsection{Spectral simulation derivation}\label{app:spec_sim}

We provide a derivation for the spectral simulation algorithm, where we construct
\begin{align*}
    Y(\vecuse{s},t) = \sqrt{\frac{\sigma}{L}} \sum_{\ell = 1}^L c_1 X_\ell^{\Re}(\vecuse{s}) W_\ell^{\Re}(t) + c_2 X_\ell^{\Im}(\vecuse{s}) W_\ell^{\Im}(t)\textrm{sign}(\langle \vecuse{x}_\ell, \tilde{\vecuse{x}}\rangle)\textrm{sign}(\eta_\ell).
\end{align*}
For each $\ell$, arguments from \cite{emery2016improved} show the covariance can be represented in expectation. 
First, we see that the $2L$-sized collection of $\{X_{\ell}^{\Re}(\vecuse{s})\}_{\ell =1}^L$ and $\{W_{\ell}^{\Re}(t)\}_{\ell =1}^L$ are all jointly independent by the joint independence of $\{\phi_\ell, \psi_\ell, \vecuse{x}_\ell, \eta_\ell\}_{\ell=1}^L$. 
Similarly, the $2L$-sized collection of $\{X_{\ell}^{\Im}(\vecuse{s})\}_{\ell =1}^L$ and $\{W_{\ell}^{\Im}(t)\}_{\ell =1}^L$ are all jointly independent. 
We look at the covariance $\mathbb{E}\{Y(\vecuse{s}_1, t_1) Y(\vecuse{s}_2, t_2)\}$, breaking it down in the four cross-product terms. 
Due to the independence relations just noted, start with the first terms' product \begin{align*}
    \mathbb{E}\left\{X_{\ell}^\Re(\vecuse{s}_1) X_{\ell}^\Re(\vecuse{s}_2)W_{\ell}^{\Re}(t_1)W_{\ell}^{\Re}(t_2)\right\} &=     \mathbb{E}\left\{X_{\ell}^\Re(\vecuse{s}_1) X_{\ell}^\Re(\vecuse{s}_2)\right\}\mathbb{E}\left\{W_{\ell}^{\Re}(t_1)W_{\ell}^{\Re}(t_2)\right\}.
\end{align*}
Focusing on the first part and applying a trigonometric identity ($\cos(a)\cos(b) = (1/2)\cos(a + b) + \cos(a - b)$) to the defined $X_{\ell}^\Re(\vecuse{s})$: \begin{align*}
    \mathbb{E}\left\{X_{\ell}^\Re(\vecuse{s}_1) X_{\ell}^\Re(\vecuse{s}_2)\right\} &= \mathbb{E}\left\{\cos(\langle \vecuse{s}_1, \vecuse{x}_\ell\rangle  + \phi_\ell )\cos(\langle \vecuse{s}_2, \vecuse{x}_\ell\rangle  + \phi_\ell )\allowbreak \frac{f_{ space}(\lVert \vecuse{x}_\ell\rVert)}{ g_{ space}(\lVert \vecuse{x}_\ell\rVert)}\right\} \\
    &=\mathbb{E}\left[\frac{1}{2}\left\{ \cos(\langle \vecuse{s}_1 + \vecuse{s}_2, \vecuse{x}_\ell\rangle  + 2\phi_\ell) + \cos(\langle \vecuse{s}_1 - \vecuse{s}_2, \vecuse{x}_\ell\rangle) \right\}\allowbreak \frac{f_{ space}(\lVert \vecuse{x}_\ell\rVert)}{ g_{ space}(\lVert \vecuse{x}_\ell\rVert)}\right].
\end{align*}
Since $2\phi_\ell$ represents a random phase change, the expectation of the first term is $0$. Therefore, when $\vecuse{x}_\ell \sim g$, \begin{align*}
    \mathbb{E}\left\{X_{\ell}^\Re(\vecuse{s}_1) X_{\ell}^\Re(\vecuse{s}_2)\right\}     &=\frac{1}{2}\mathbb{E}\left\{\cos(\langle \vecuse{s}_1 - \vecuse{s}_2, \vecuse{x}_\ell\rangle) \allowbreak \frac{f_{ space}(\lVert \vecuse{x}_\ell\rVert)}{ g_{ space}(\lVert \vecuse{x}_\ell\rVert)}\right\} \\
    &= \frac{1}{2}\int_{\mathbb{R}^d}\cos(\langle \vecuse{s}_1 - \vecuse{s}_2, \vecuse{x}\rangle) \allowbreak \frac{f_{ space}(\lVert \vecuse{x}\rVert)}{ g_{ space}(\lVert \vecuse{x}\rVert)} g_{space}(\lVert \vecuse{x}\rVert)  d\vecuse{x} \\
    &= \frac{1}{2}\int_{\mathbb{R}^d}\cos(\langle \vecuse{s}_1 - \vecuse{s}_2, \vecuse{x}\rangle) \allowbreak f_{ space}(\lVert \vecuse{x}\rVert)  d\vecuse{x} = \frac{1}{2} C^{\Re}_{space}(\vecuse{s}_1 - \vecuse{s}_2).
\end{align*}
The same arguments apply for $W_{\ell}^\Re(t)$: \begin{align*}
    \mathbb{E}\left\{W_{\ell}^{\Re}(t_1)W_{\ell}^{\Re}(t_2)\right\} &= \frac{1}{2}C^{\Re}_{time}(t_1 - t_2). 
\end{align*}
Therefore, the component from the first terms has expectation \begin{align*}
\frac{ \sigma c_1^2}{4} C^{\Re}_{space}(\vecuse{s}_1 - \vecuse{s}_2)C^{\Re}_{time}(t_1 - t_2). 
\end{align*}

We next turn to the second terms multiplied, rather than the cross-product terms. 
We now have \begin{align*}
        &\mathbb{E}\left\{X_{\ell}^\Im(\vecuse{s}_1) X_{\ell}^\Im(\vecuse{s}_2)W_{\ell}^{\Im}(t_1)W_{\ell}^{\Im}(t_2)\textrm{sign}(\langle \vecuse{x}_\ell, \tilde{\vecuse{x}}\rangle)^2\textrm{sign}(\eta_\ell)^2\right\} \\
        &~~~~~~=     \mathbb{E}\left\{X_{\ell}^\Im(\vecuse{s}_1) X_{\ell}^\Im(\vecuse{s}_2)\right\}\mathbb{E}\left\{W_{\ell}^{\Im}(t_1)W_{\ell}^{\Im}(t_2)\right\},
\end{align*}
as $\textrm{sign}(z)^2 = 1$ for $z \neq 0$, and using now $\sin(a)\sin(b) = (1/2) (\cos(a - b) - \cos(a+b))$, we have \begin{align*}
        \mathbb{E}\left\{X_{\ell}^\Im(\vecuse{s}_1) X_{\ell}^\Im(\vecuse{s}_2)\right\} &= \mathbb{E}\left\{\sin(\langle \vecuse{s}_1, \vecuse{x}_\ell\rangle  + \phi_\ell )\sin(\langle \vecuse{s}_2, \vecuse{x}_\ell\rangle  + \phi_\ell )\allowbreak \frac{f_{ space}(\lVert \vecuse{x}_\ell\rVert)}{ g_{ space}(\lVert \vecuse{x}_\ell\rVert)}\right\} \\
    &=\frac{1}{2}\mathbb{E}\left[\left\{ \cos(\langle \vecuse{s}_1 - \vecuse{s}_2, \vecuse{x}_\ell\rangle ) -\cos(\langle \vecuse{s}_1 + \vecuse{s}_2, \vecuse{x}_\ell + 2\phi_\ell\rangle) \right\}\allowbreak \frac{f_{ space}(\lVert \vecuse{x}_\ell\rVert)}{ g_{ space}(\lVert \vecuse{x}_\ell\rVert)}\right].
\end{align*}
The same arguments apply, with the random uniform phase change, and the expression is equivalent to before: \begin{align*}
     \mathbb{E}\left\{X_{\ell}^\Im(\vecuse{s}_1) X_{\ell}^\Im(\vecuse{s}_2)\right\} &= \frac{1}{2}C_{space}^\Re(\vecuse{s}_1 - \vecuse{s}_2),
\end{align*}and the component from the second terms has expectation \begin{align*}
        \frac{\sigma c_2^2 }{4}C^{\Re}_{space}(\vecuse{s}_1 - \vecuse{s}_2)C^{\Re}_{time}(t_1 - t_2). 
\end{align*}

We now turn to the cross-product terms with: \begin{align*}
        &\mathbb{E}\left\{X_{\ell}^\Re(\vecuse{s}_1) X_{\ell}^\Im(\vecuse{s}_2)W_{\ell}^{\Re}(t_1)W_{\ell}^{\Im}(t_2)\textrm{sign}(\langle \vecuse{x}_\ell, \tilde{\vecuse{x}}\rangle)\textrm{sign}(\eta_\ell)\right\} \\
        &~~~~~~=     \mathbb{E}\left\{X_{\ell}^\Re(\vecuse{s}_1) X_{\ell}^\Im(\vecuse{s}_2) \textrm{sign}(\langle \vecuse{x}_\ell, \tilde{\vecuse{x}}\rangle)\right\}\mathbb{E}\left\{W_{\ell}^{\Re}(t_1)W_{\ell}^{\Im}(t_2) \textrm{sign}(\eta_\ell)\right\},
\end{align*}
We follow again with a similar argument, this time with $\cos(a)\sin(b) =(1/2) (\sin(a +b) - \sin(a - b))$: \begin{align*}
      &\mathbb{E}\left\{X_{\ell}^\Re(\vecuse{s}_1) X_{\ell}^\Im(\vecuse{s}_2) \textrm{sign}(\langle \vecuse{x}_\ell, \tilde{\vecuse{x}}\rangle)\right\} \\
      &~~~~~= \mathbb{E}\left\{ \cos(\langle \vecuse{s}_1, \vecuse{x}_\ell\rangle  + \phi_\ell )\sin(\langle \vecuse{s}_2, \vecuse{x}_\ell\rangle  + \phi_\ell )\allowbreak \frac{f_{ space}(\lVert \vecuse{x}_\ell\rVert)}{ g_{ space}(\lVert \vecuse{x}_\ell\rVert)}\textrm{sign}(\langle \vecuse{x}_\ell, \tilde{\vecuse{x}}\rangle)\right\} \\
        &=\frac{1}{2}\mathbb{E}\left[\left\{ \sin(\langle \vecuse{s}_1 + \vecuse{s}_2, \vecuse{x}_\ell + 2 \phi_\ell \rangle ) -\sin(\langle \vecuse{s}_1 - \vecuse{s}_2, \vecuse{x}_\ell\rangle) \right\}\allowbreak \frac{f_{ space}(\lVert \vecuse{x}_\ell\rVert)}{ g_{ space}(\lVert \vecuse{x}_\ell\rVert)} \textrm{sign}(\langle \vecuse{x}_\ell, \tilde{\vecuse{x}}\rangle)\right]\\
        &=-\frac{1}{2}\mathbb{E}\left[\left\{\sin(\langle \vecuse{s}_1 - \vecuse{s}_2, \vecuse{x}_\ell\rangle) \right\}\allowbreak \frac{f_{ space}(\lVert \vecuse{x}_\ell\rVert)}{ g_{ space}(\lVert \vecuse{x}_\ell\rVert)} \textrm{sign}(\langle \vecuse{x}_\ell, \tilde{\vecuse{x}}\rangle)\right]\\
        &=-\frac{1}{2} \int_{\mathbb{R}^d }\sin(\langle \vecuse{s}_1 - \vecuse{s}_2, \vecuse{x}\rangle) \allowbreak \frac{f_{ space}(\lVert \vecuse{x}\rVert)}{ g_{ space}(\lVert \vecuse{x}\rVert)} \textrm{sign}(\langle \vecuse{x}, \tilde{\vecuse{x}}\rangle) d\vecuse{x} = -\frac{1}{2} C_{space}^{\Im}(\vecuse{s}_1 - \vecuse{s_2}).
\end{align*}
Again the argument for temporal portions is the same. For this cross-product term, we have \begin{align*}
    \sigma \frac{c_1c_2}{4} \left\{-C^{\Im}_{space}(\vecuse{s}_1 - \vecuse{s}_2)\right\}\left\{-C^{\Im}_{time}(t_1 - t_2)\right\}.
\end{align*}
We next handle the final cross-product term:  \begin{align*}
        &\mathbb{E}\left\{X_{\ell}^\Im(\vecuse{s}_1) X_{\ell}^\Re(\vecuse{s}_2)W_{\ell}^{\Im}(t_1)W_{\ell}^{\Re}(t_2)\textrm{sign}(\langle \vecuse{x}_\ell, \tilde{\vecuse{x}}\rangle)\textrm{sign}(\eta_\ell)\right\} \\
        &~~~~~~=     \mathbb{E}\left\{X_{\ell}^\Im(\vecuse{s}_1) X_{\ell}^\Re(\vecuse{s}_2) \textrm{sign}(\langle \vecuse{x}_\ell, \tilde{\vecuse{x}}\rangle)\right\}\mathbb{E}\left\{W_{\ell}^{\Im}(t_1)W_{\ell}^{\Re}(t_2) \textrm{sign}(\eta_\ell)\right\},
\end{align*}
We follow again with a similar argument, this time with $\sin(a)\cos(b) = (1/2)(\sin(a +b) + \sin(a - b)) $: \begin{align*}
      &\mathbb{E}\left\{X_{\ell}^\Im(\vecuse{s}_1) X_{\ell}^\Re(\vecuse{s}_2) \textrm{sign}(\langle \vecuse{x}_\ell, \tilde{\vecuse{x}}\rangle)\right\} \\
      &~~~~~= \mathbb{E}\left\{ \sin(\langle \vecuse{s}_1, \vecuse{x}_\ell\rangle  + \phi_\ell )\cos(\langle \vecuse{s}_2, \vecuse{x}_\ell\rangle  + \phi_\ell )\allowbreak \frac{f_{ space}(\lVert \vecuse{x}_\ell\rVert)}{ g_{ space}(\lVert \vecuse{x}_\ell\rVert)}\textrm{sign}(\langle \vecuse{x}_\ell, \tilde{\vecuse{x}}\rangle)\right\} \\
        &=\frac{1}{2}\mathbb{E}\left[\left\{ \sin(\langle \vecuse{s}_1 + \vecuse{s}_2, \vecuse{x}_\ell + 2 \phi_\ell \rangle ) +\sin(\langle \vecuse{s}_1 - \vecuse{s}_2, \vecuse{x}_\ell\rangle) \right\}\allowbreak \frac{f_{ space}(\lVert \vecuse{x}_\ell\rVert)}{ g_{ space}(\lVert \vecuse{x}_\ell\rVert)} \textrm{sign}(\langle \vecuse{x}_\ell, \tilde{\vecuse{x}}\rangle)\right]\\
        &=\frac{1}{2}\mathbb{E}\left[\left\{\sin(\langle \vecuse{s}_1 - \vecuse{s}_2, \vecuse{x}_\ell\rangle) \right\}\allowbreak \frac{f_{ space}(\lVert \vecuse{x}_\ell\rVert)}{ g_{ space}(\lVert \vecuse{x}_\ell\rVert)} \textrm{sign}(\langle \vecuse{x}_\ell, \tilde{\vecuse{x}}\rangle)\right]\\
        &=\frac{1}{2} \int_{\mathbb{R}^d }\sin(\langle \vecuse{s}_1 - \vecuse{s}_2, \vecuse{x}\rangle) \allowbreak \frac{f_{ space}(\lVert \vecuse{x}\rVert)}{ g_{ space}(\lVert \vecuse{x}\rVert)} \textrm{sign}(\langle \vecuse{x}, \tilde{\vecuse{x}}\rangle) d\vecuse{x} = \frac{1}{2} C_{space}^{\Im}(\vecuse{s}_1 - \vecuse{s_2}).
\end{align*}
Then, again the argument for temporal portions is the same. For this cross-product term, we have \begin{align*}
    \frac{\sigma c_1c_2}{4} C^{\Im}_{space}(\vecuse{s}_1 - \vecuse{s}_2)C^{\Im}_{time}(t_1 - t_2).
\end{align*}

Combining all four terms and averaging over the $L$ independent features with identical expectation, we have \begin{align*}
    \textrm{Cov}\left\{Y(\vecuse{s}_1, t_1) ,Y(\vecuse{s}_2, t_2)\right\} &= \sigma \left(\frac{c_1^2 +c_2^2}{4}\right) C^{\Re}_{space}(\vecuse{s}_1 - \vecuse{s}_2)C^{\Re}_{time}(t_1 - t_2) \\ 
    & ~~~~~+ \sigma \frac{c_1 c_2}{2} C^{\Im}_{space}(\vecuse{s}_1 - \vecuse{s}_2)C^{\Im}_{time}(t_1 - t_2). 
\end{align*}
It remains to choose appropriate $c_1$ and $c_2$. In particular, we want $(c_1^2+ c_2^2)/4 = 1$ and $c_1c_2/2 = \xi$. 
By adding and subtracting these equations, we obtain $(c_1 + c_2)^2/4 = 1 + \xi$ and $(c_1 - c_2)^2/4 = 1- \xi$, or $c_1 + c_2 = 2\sqrt{1 + \xi}$ and $c_1 - c_2 = 2\sqrt{1- \xi}$. Therefore, the values $c_1 =\sqrt{1 +\xi} + \sqrt{1- \xi}$ and $c_2 = \sqrt{1+\xi} - \sqrt{1- \xi}$ will satisfy this. 
By summing $L$ independent terms with the same expectation and finite variance, as argued in \cite{emery2016improved}, the resulting field should be approximately Gaussian.

\subsection{Additional spectral simulation results}

In Fig.~\ref{fig:spec_sim_rmse}, we give additional results on the spectral simulation, plotting an example simulated dataset, the bias of the estimated covariances for the true covariance, the standard deviation of the estimates across simulations, and the number of data point combinations used for estimation. 
In general we see relatively low bias values, with more deviations for $u \neq 0$. 
Similarly, the standard deviation of estimates is low, while larger for $|u|$ and $\lVert \vecuse{h}\rVert$ as there are fewer data point combination to estimate the covariance. 

\begin{figure}
    \centering
    \includegraphics[width=0.98\linewidth]{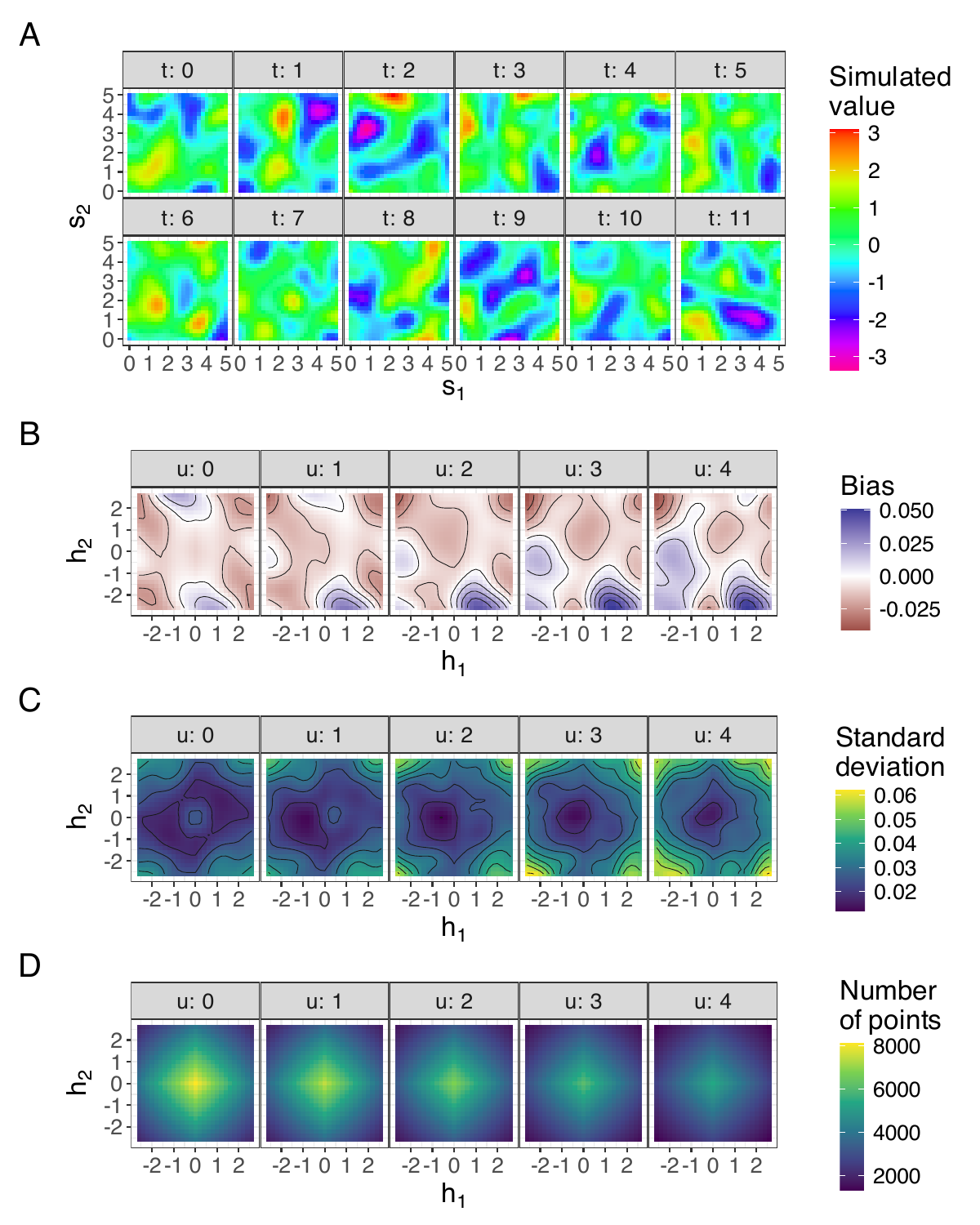}
    \caption{(A) Example simulated data (B) Bias across 200 simulations for the covariance estimate $\hat{C}(\vecuse{h},u)$ for the true covariance $\hat{C}(\vecuse{h},u)$ (C) Standard deviation of $\hat{C}(\vecuse{h},u)$ across $200$ simulations (D) Number of data point combinations used to estimate the covariance at each lag. }
    \label{fig:spec_sim_rmse}
\end{figure}

\subsection{Additional simulation results}
\label{app:sim_results}

We plot AIC information for the multivariate spatial simulation including when the covariance function is misspecified in Fig.~\ref{fig:sim_misspecified}. 
In general, including asymmetry in the model can still improve performance when the true model is asymmetric with a different covariance class. 
For example, the asymmetric Mat\'ern provides performance near that of the squared-exponential when the true covariance is squared-exponential and $\Im(\sigma_{12}) = 0.4$. 
The Mat\'ern and Cauchy perform similarly across most settings, while the squared-exponential has considerably higher AIC when the true covariance is Mat\'ern. 

In Figs.~\ref{fig:sym_est} and \ref{fig:params}, we give additional estimates of the multivariate simulation study. 
In Fig.~\ref{fig:sym_est}, we see that $\Re(\sigma_{12})$ is estimated well when there is negative, zero, or positive marginal correlation. Estimates generally improve as the sample size increases, and are relatively similar regardless of the covariance function. 
In Fig.~\ref{fig:params}, we plot the other parameter estimates. The squared-exponential has much less spread in the estimates. The Cauchy and Mat\'ern have two additional parameters to estimate, and identifiability challenges between $a_1$ and $\sigma_{11}$, and on, could contribute to the spread. However, there appears to be little bias even with the larger spread. 
Overall, these suggest the parameters are well estimated, including asymmetry in the main text and measures of dependence between the processes. 

\begin{figure}
    \centering
    \includegraphics[width=0.98\linewidth]{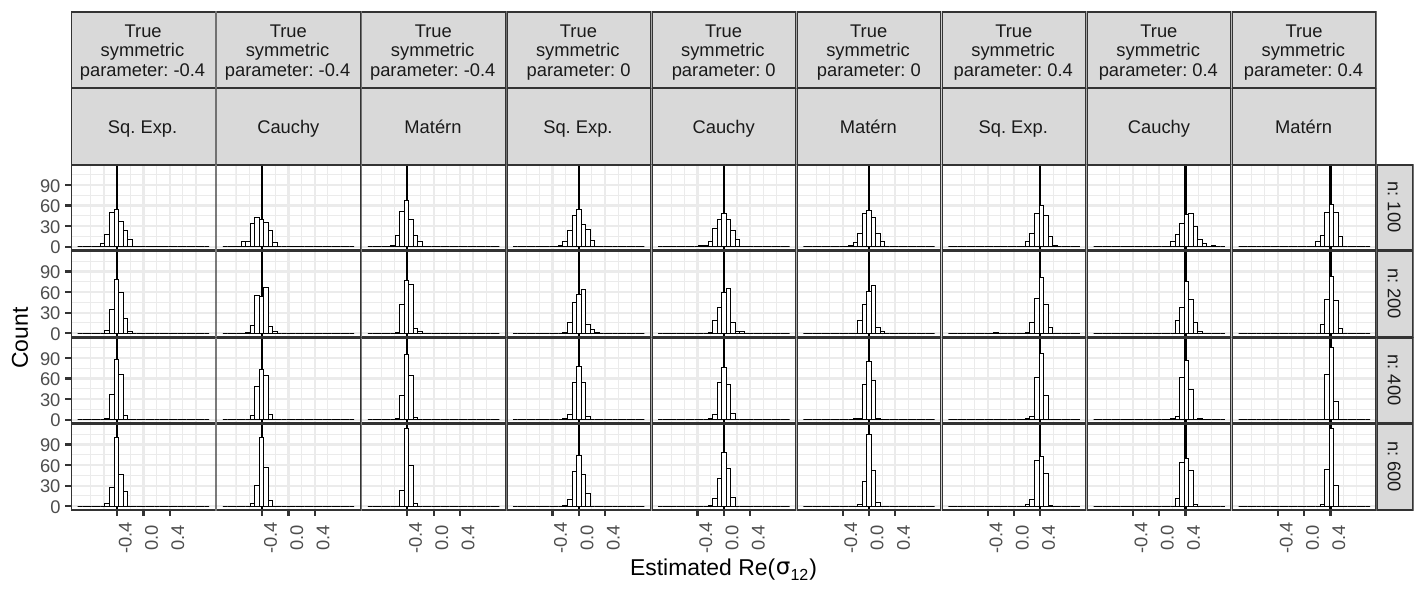}
    \caption{Histogram of estimates of $\Re(\sigma_{12})$ for various sample sizes, true $\Re(\sigma_{12})$, and covariances. Each histogram has 200 points combining when $\Im(\sigma_{12}) \in \{0.0, 0.4\}$. The true parameter is plotted with a vertical line in each panel. }
    \label{fig:sym_est}
\end{figure}

\begin{figure}
    \centering
    \includegraphics[width=0.98\linewidth]{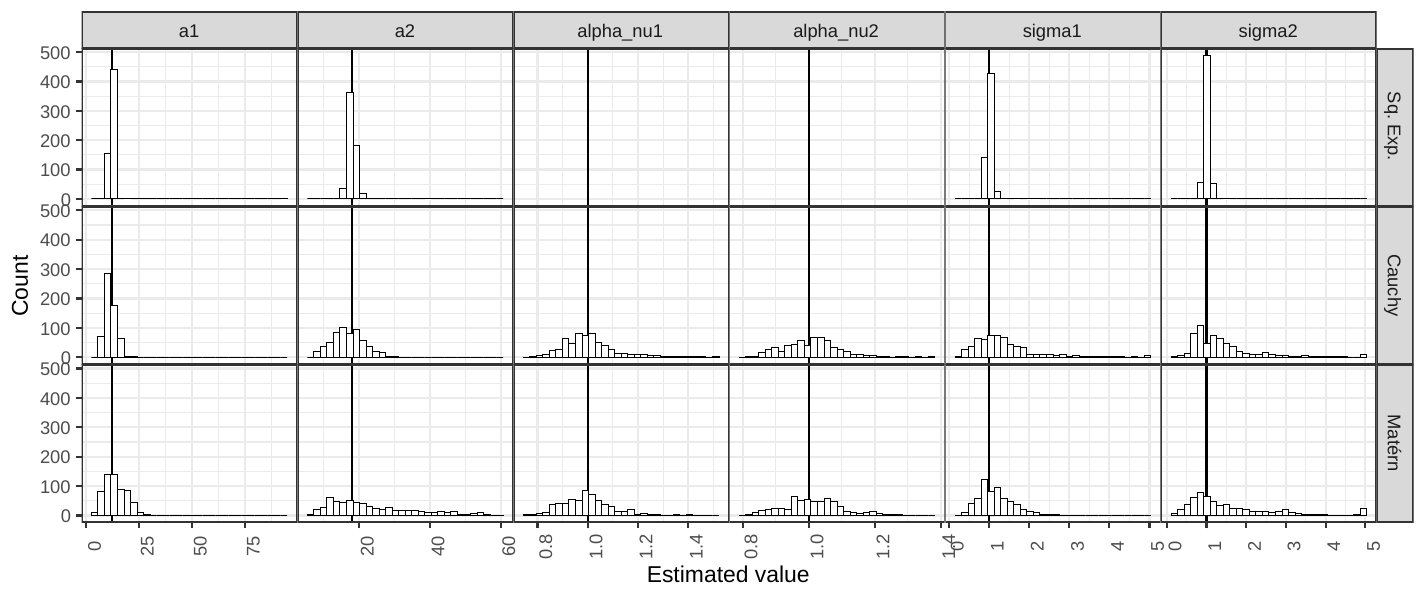}
    \caption{Histogram of estimates of other parameters $a_1$, $a_2$, $\alpha_1$, $\alpha_2$ (Cauchy), $\nu_1$, $\nu_2$ (Mat\'ern), $\sigma_{11}$, and $\sigma_{22}$ in the multivariate simulation study for $n=200$ and various covariances. Each histogram has 600 points combining when $\Re(\sigma_{12}) \in \{-0.4, 0.0, 0.4\}$ and  $\Im(\sigma_{12}) \in \{0.0, 0.4\}$. The true parameter is plotted with a vertical line in each panel. }
    \label{fig:params}
\end{figure}

We present computation time summaries in Fig.~\ref{fig:time_comparison}. 
For (A), we compute the $2n\times 2n$ covariance matrix for symmetric and asymmetric models $100$ times. 
In this case, we compute both the symmetric and asymmetric Mat\'ern by bilinearly interpolating the fast Fourier transform, so their results are somewhat different from the rest, requiring a fixed cost of computing this transform. 
The squared-exponential and Cauchy symmetric covariances can be computed the quickest, while the others require additional time. 
Panel (B) presents the total optimization times for the simulation, which is also affected by the number of parameters estimated. Estimating additional Cauchy, Mat\'ern, and asymmetry parameters increase the optimization time compared the symmetric squared-exponential covariance.

\begin{figure}
    \centering
    \includegraphics[width=0.98\linewidth]{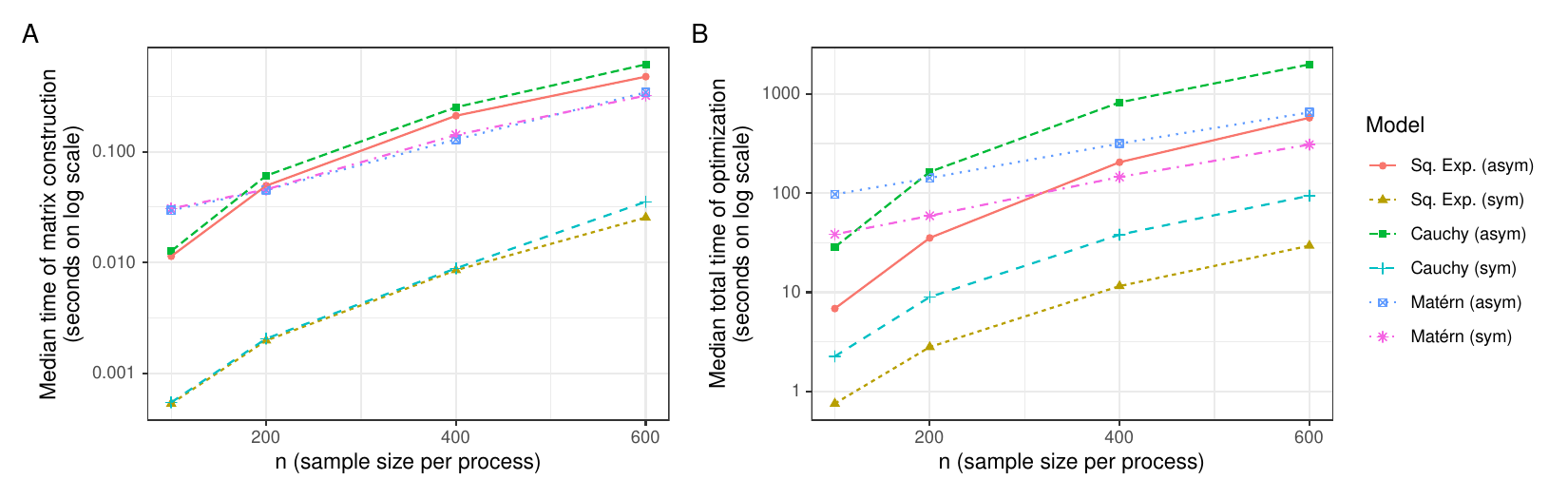}
    \caption{Median time comparison of (A) $2n \times 2n$ covariance matrix construction for the multivariate simulation study and (B) total optimization time summary for the multivariate simulation study.}
    \label{fig:time_comparison}
\end{figure}

In Tables \ref{tab:mult_sim_hypothesis}, \ref{tab:st_sim_hypothesis}, and \ref{tab:st_sim_hypothesis2}, we provide empirical size and power analyses for a likelihood ratio test in the multivariate and space-time simulation studies. 
In particular, we consider a chi-squared likelihood ratio test of level $\alpha = 0.05$. 
In general, we find that the test controls Type I error rate generally well when the covariance class is correctly specified, while it may be somewhat inflated if the covariance class is misspecified. 
Under the simulation settings, the test has more power in the space-time simulation, perhaps due to the larger total sample size $n = n_s\times n_t$.

\begin{figure}[th]
    \centering
    \includegraphics[width=0.96\linewidth]{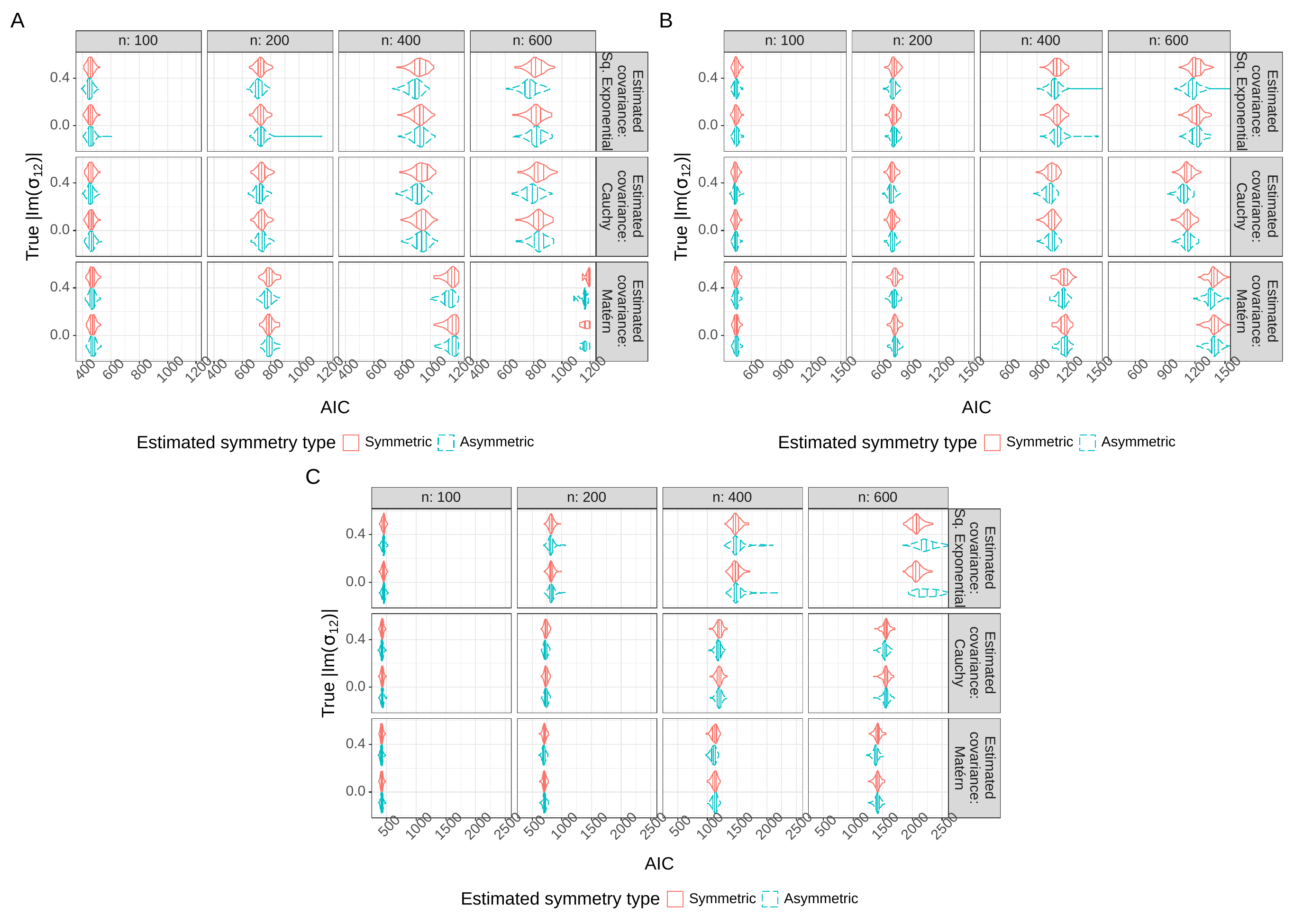}

    \caption{AIC for estimated multivariate spatial models when the true covariance is squared-exponential (A), Cauchy (B), and Mat\'ern (C). }
    \label{fig:sim_misspecified}
\end{figure}

\begin{table}[th] \centering 
\begin{tabular}{|crrrrrr|} \hline 
\multicolumn{7}{|c|}{True covariance}  \\ \hline 
Class & \multicolumn{2}{c}{Sq.~Exp.} & \multicolumn{2}{c}{Cauchy} & \multicolumn{2}{c|}{Mat\'ern} \\ \hline 
$\Im(\sigma_{12})$ & 0.00 & 0.40 & 0.00 & 0.40  & 0.00 & 0.40 \\ \hline \hline 
LRT Sq.~Exp. & 0.02 & 0.56 & 0.04 & 0.44 & 0.05 & 0.45 \\ 
LRT Cauchy & 0.03 & 0.47  & 0.03 & 0.42  & 0.03 & 0.48   \\ 
LRT Mat\'ern & 0.01 & 0.48  & 0.03 & 0.38 & 0.02 & 0.38   \\  \hline 
\end{tabular} 
  \caption{For the multivariate spatial simulation study, the proportion of rejections for the chi-squared likelihood ratio tests at the $0.05$ level when $n = 100$ and $\Re(\sigma_{jk}) = 0.4$. The likelihood ratio tests (LRTs) are based on the difference of log-likelihoods for the symmetric and reflective asymmetric multivariate covariances of the specified class.} 
  \label{tab:mult_sim_hypothesis} 
\end{table}

\begin{table}[th] \centering 
\begin{tabular}{|crrrrrrrrr|} \hline 
\multicolumn{10}{|c|}{True covariance}  \\ \hline 
Class & \multicolumn{3}{c}{Sq.~Exp.} & \multicolumn{3}{c}{Cauchy} & \multicolumn{3}{c|}{Exponential} \\ \hline 
$\xi$ & 0.00 & 0.40 & 0.80 & 0.00 & 0.40 & 0.80 & 0.00 & 0.40 & 0.80 \\ \hline \hline 
LRT Sq.~Exp. & 0.05 & 0.98 & 1.00 & 0.12 & 0.98 & 1.00 & 0.06 & 0.76 & 1.00 \\ 
LRT Cauchy & 0.04 & 0.77 & 1.00& 0.03 & 1.00 & 1.00 & 0.05 & 0.81 & 1.00  \\ 
LRT Exponential & 0.08 & 0.97 & 1.00 & 0.08 & 0.99 & 1.00 & 0.04 & 0.98 & 1.00  \\ 
LRT Gneiting $b$ & 0.07 & 0.76 & 1.00 & 0.09 & 0.51 & 1.00 & 0.07 & 0.38 & 0.87  \\ 
LRT Gneiting & 0.09 & 0.84 & 1.00 & 0.37 & 0.74 & 0.97 & 0.27 & 0.46 & 0.89 \\ \hline
\end{tabular} 
  \caption{For the space-time simulation study, the proportion of rejections for the chi-squared likelihood ratio tests at the $0.05$ level when $n_s = 20$ and locations are colocated. The likelihood ratio tests (LRTs) are based on the difference of log-likelihoods for the symmetric and reflective asymmetric covariances of the specified class, where Gneiting $b$ refers to with the nonseparability parameter and Gneiting refers to without the nonseparability parameter.} 
  \label{tab:st_sim_hypothesis} 
\end{table}

\begin{table}[th] \centering 
\begin{tabular}{|crrrrrrr|} \hline 
\multicolumn{8}{|c|}{True covariance}  \\ \hline 
Class & Sq.~Exp.~Lagrangian & \multicolumn{3}{c}{Gneiting $b$} & \multicolumn{3}{c|}{Gneiting} \\ \hline 
$\xi$ &  -  & 0.00 & 0.40 & 0.80 & 0.00 & 0.40 & 0.80 \\ \hline \hline 
LRT Sq.~Exp. & 1.00  & 0.07 & 1.00 & 1.00 & 0.10 & 0.97 & 1.00 \\ 
LRT Cauchy & 1.00  & 0.10 & 1.00 & 1.00 & 0.11 & 0.97 & 1.00  \\ 
LRT Exponential & 0.99 & 0.15 & 0.98 & 1.00 & 0.12 & 0.96 & 1.00  \\ 
LRT Gneiting $b$ & 1.00 & 0.12 & 1.00 & 1.00 & 0.04 & 1.00 & 0.99  \\ 
LRT Gneiting & 1.00 & 0.09 & 1.00 & 1.00 & 0.04 & 1.00 & 1.00 \\ \hline
\end{tabular} 
  \caption{For the space-time simulation study, the proportion of rejections for the chi-squared likelihood ratio tests at the $0.05$ level when $n_s = 20$ and locations are colocated. The likelihood ratio tests (LRTs) are based on the difference of log-likelihoods for the symmetric and reflective asymmetric covariances of the specified class, where Gneiting $b$ refers to with the nonseparability parameter and Gneiting refers to without the nonseparability parameter.} 
  \label{tab:st_sim_hypothesis2} 
\end{table} 

Finally, we consider additional parameter estimation for the univariate space-time study. 
In Table \ref{tab:lag_est}, we give parameter estimation results for the Lagrangian model. 
We find that the estimates for smaller sample size $n_s = 20$ are not estimated well overall for true parameters $a_s = 10$, $\sigma = 1$, $\mathbb{E}(V) = 0.5$, and $\textrm{Var}(V) = 225$. 
However, the parameters are estimated very well for $n_s = 50$ and $n_s = 100$, with little bias in the estimates.

\begin{table}[th]
    \centering
    \begin{tabular}{|c|c|c|c|c|c|}\hline 
    $n_s$ & $a_s$ & $\sigma$& $\mathbb{E}(V)$ & \textrm{Var}(V) \\  \hline
    20 & 22.5 (45.6) &  1.05 (2.07) &  0.67 (2.07) & 281 (200) \\
    50 &  10.0 (0.26) &  1.02 (0.15) & 0.50 (0.06)  &    224 (23.3) \\
    100 & 10.0 (0.22) &   1.00 (0.13) &0.50 (0.06) &   225 (20.3) \\ \hline
    \end{tabular}
    \caption{Lagrangian parameter estimates when correctly specified over both colocated and not colocated settings. Estimates are presented as mean (standard deviation) over the simulations. }
    \label{tab:lag_est}
\end{table}

\subsection{Irish wind prediction study} \label{app:prediction}

\begin{table}[ht]
    \centering
    \footnotesize
    \begin{tabular}{|cccrrrrrr|}\hline
        Cov ($\vecuse{h}$) & Cov ($u$) & Type& $\sigma$ & $a_{\vecuse{s}}$ & $a_t$  & $\tau^2$ & $\xi$ & $\textrm{atan2}(\tilde{\vecuse{x}})$  \\ \hline\hline
        Sq.~Exp. & Cauchy ($\alpha_t = 1$) & Sym& 0.57 &0.0024 & 0.84 & 0.07 & - & - \\ 
        Sq.~Exp. & Cauchy ($\alpha_t = 1$) & Asym&0.59 &0.0024 & 0.83 & 0.06 & 0.50 & $-1.6$ \\  
        Sq.~Exp. & Cauchy ($\alpha_t = \frac{1}{2}$)&Sym& 0.61 & 0.0024 &1.23 &  0.07& - & - \\ 
        Sq.~Exp. & Cauchy ($\alpha_t = \frac{1}{2}$) &Asym& 0.63 & 0.0024 &1.20 &  0.06  & 0.50 & $-2.2$  \\  
        Sq.~Exp. & Sq.~Exp.  & Sym&0.54 &0.0024 & 0.84 &  0.06& - & - \\ 
        Sq.~Exp. & Sq.~Exp. &Asym& 0.55 & 0.0024 & 0.83 & 0.07 & 0.48 & $-3.3$ \\  
        Cauchy ($\alpha_{\vecuse{s}} = \frac{1}{2}$) & Cauchy ($\alpha_t = \frac{1}{2}$) & Sym&0.54 &0.0051 & 1.33 & 0.04& - & - \\ 
        Cauchy ($\alpha_{\vecuse{s}} = \frac{1}{2}$) & Cauchy ($\alpha_t = \frac{1}{2}$) & Asym&0.54 &0.0051 & 1.31& 0.04 & 0.46 & ~~3.0  \\ 
        Cauchy ($\alpha_{\vecuse{s}}$) & Cauchy ($\alpha_t$) & Sym&0.61 &0.0182 & 2.56 & 0.01& - & - \\ 
        Cauchy ($\alpha_{\vecuse{s}}$) & Cauchy ($\alpha_t$) & Asym&0.60 &0.0170 & 2.38& 0.02 & 0.42 & ~~8.1 \\ 
        \multicolumn{2}{|c}{GpGp exponential\_spacetime} &Sym& 0.43 &0.0025 & 0.93 & 0.00 & - & - \\ 
        \multicolumn{2}{|c}{GpGp matern\_spacetime} &Sym& 0.60 &0.0055 & 2.00 & 0.07 & -&-\\ 
        \multicolumn{2}{|c}{Lagrangian Sq.~Exp.} & Asym & 0.65 & 0.0022 & - & 0.07 & -& - \\ \hline
    \end{tabular}
    \caption{Estimated covariance parameters for the Irish wind data. The column $\textrm{atan2}(\tilde{\vecuse{x}})$ represents angle from the vector $(1,0)^\top$ (directly East) converted to degrees. The symmetric and asymmetric Cauchy with general exponent had estimated $\alpha_{\vecuse{s}} = 0.11$ and $\alpha_{\vecuse{s}} =0.11$, respectively, and $\alpha_t = 0.36$ and $\alpha_t = 0.37$, respectively. The additional parameters for the Lagrangian model were $\vecuse{\mu}_{\vecuse{V}} = ( 114.0, -34.5)^\top$ in kilometers (corresponding to $\textrm{atan2}(\vecuse{\mu}_{\vecuse{V}}) = -16.9$ degrees), and $(\matuse{\Sigma}_{\vecuse{V}})_{11}= 111{,}921$, $(\matuse{\Sigma}_{\vecuse{V}})_{22} = 12{,}701$, and $(\matuse{\Sigma}_{\vecuse{V}})_{12} = (
\matuse{\Sigma}_{\vecuse{V}})_{21} = -37{,}666$. The GpGp matern\_spacetime model had estimated parameter $\nu = 2.35$. }
    \label{tab:irish_parameters}
\end{table}

For the Irish wind analysis, we present the estimated parameters in Table \ref{tab:irish_parameters}. 
The symmetric and asymmetric versions of the models have comparable variance, inverse range, and nugget variance parameters. 
Furthermore, when the spatial covariance is of the same form, the inverse range parameters are similarly estimated regardless of the temporal covariance used. 
The parameter $\xi$ is estimated to be similarly strong for each of the new asymmetric models. 
The estimated parameter $\tilde{\vecuse{x}}$ points East in general as expected, matching the plot of covariance functions. 
The Cauchy models with $\alpha_{\vecuse{s}}$ and $\alpha_t$ estimated have these parameters less than $1/2$, suggesting heavier tail decay than this special case.

We next present the prediction study. 
For a particular time $t_{pred}$, we use data from all $t < t_{pred}$ to predict $Y(\vecuse{s},t)$ at times $Y(\vecuse{s},t_{pred})$, $Y(\vecuse{s}, t_{pred} + 1)$, $Y(\vecuse{s}, t_{pred} + 2)$, $\dots$. 
We consider two settings. 
The first uses all locations $\{Y(\vecuse{s},t)\}_{t < t_{pred},\vecuse{s}}$ as prediction data for $Y(\vecuse{s}_{pred},t_{pred})$.
The second uses only data from other locations, using $\{Y(\vecuse{s},t)\}_{t < t_{pred},\vecuse{s}\neq \vecuse{s}_{pred}}$ for $Y(\vecuse{s}_{pred},t_{pred})$. 
We refer to this as the leave-one-out setting, which is fundamentally more challenging than the first setting. 
We take $t_{pred}$ to be each day from January 1st, 1971 to December 30th, 1978, where the mean function and covariance parameters were estimated from data before January 1st, 1971. 
This gives $2{,}922$ total $t_{pred}$ for each predictions are formed for the subsequent days. 
Predictions for different $t_{pred}$ may be computed in parallel. 

To form predictions, we again use Vecchia's approximation through \cite{GpGp}. 
We use the same $m=60$ neighbors for each method to compute conditional expectations of $Y(\vecuse{s}_{pred}, t_{pred})$. 
In addition, we use \cite{GpGp}'s conditional simulation with $m=8$ to simulate $50$ realizations from the conditional distribution of $Y(\vecuse{s}_{pred}, t_{pred})$, whose variance is used as a conditional variance. 
All distributions are taken to be Gaussian.

To evaluate the models, we use the continuous ranked probability score (CRPS) that takes into account the entire distributional prediction \citep{gneiting_strictly_2007}, with smaller values representing better prediction. 
We provide summaries of this score for the two settings, as well as for a different number of days in advance of prediction, plotted in Fig.~\ref{fig:pred}. 
In general, asymmetric versions of the model have improved prediction than their symmetric versions, especially when predicting just one day in advance. 
When all locations are used, the models with Cauchy spatial covariances perform the best.
There are fewer differences in the leave-one-out setting, but the models with both squared-exponential spatial and temporal covariances lag behind other models. 
The asymmetric Lagrangian model performs well in the leave-one-out setting, but is worse than some reflective asymmetric models when using all locations for prediction. 
In general, prediction performance drops off quickly from one day in advance to two or three days in advance. 
As expected, the leave-one-out-setting has worse predictions than the all locations setting.

\begin{figure}[!th]
    \centering
    \includegraphics[width=0.6\linewidth]{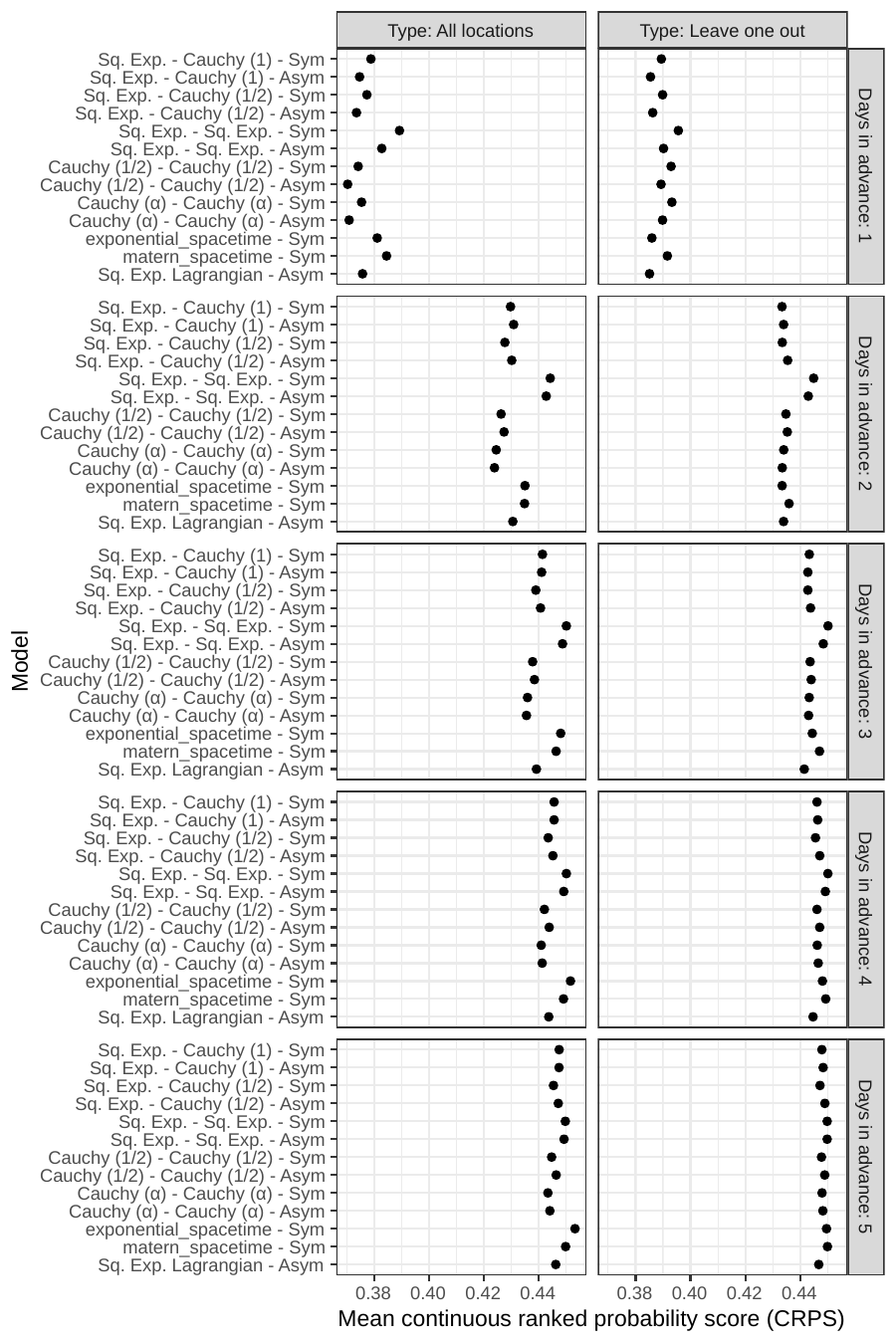}
    \caption{Prediction summaries for the Irish wind data across the study design and the number of days in advance the prediction was formed. The models are in the same order as Table 2. }
    \label{fig:pred}
\end{figure}

\subsection{Air pollution data}\label{app:air}

In Table \ref{tab:pm_parameters}, we give the estimated parameters for the PM$_{2.5}$ data. Generally, the directions are more varied than the Irish wind data and the estimated $\xi$ is smaller, indicating that the asymmetry is less strong. The estimated Cauchy $\alpha$ parameters small for the spatial dimension but larger for the temporal dimension. Once again, the symmetric and reflective asymmetric estimated model parameters are comparable to each other.  
\begin{table}[ht]
    \centering
    \footnotesize
    \begin{tabular}{|cccrrrrrr|}\hline
        Cov ($\vecuse{h}$) & Cov ($u$) & Type& $\sigma$ & $a_{\vecuse{s}}$ & $a_t$  & $\tau^2$ & $\xi$ & $\textrm{atan2}(\tilde{\vecuse{x}})$  \\ \hline\hline
        Sq.~Exp. & Cauchy ($\alpha_t = 1$) & Sym& 0.16 &0.0034 & 1.27 & 0.021 & - & - \\ 
        Sq.~Exp. & Cauchy ($\alpha_t = 1$) & Asym&0.16 &0.0035 & 1.28 & 0.021 & 0.10 & $148.0$ \\  
        Sq.~Exp. & Cauchy ($\alpha_t = \frac{1}{2}$)&Sym& 0.17 & 0.0034 &2.44 &  0.022& - & - \\ 
        Sq.~Exp. & Cauchy ($\alpha_t = \frac{1}{2}$) &Asym& 0.16 & 0.0034 &2.44 &  0.021  & 0.11 & $142.8$  \\  
        Sq.~Exp. & Sq.~Exp.  & Sym&0.16&0.0034 & 1.03 &  0.021& - & - \\ 
        Sq.~Exp. & Sq.~Exp. &Asym& 0.15 & 0.0035 & 1.04 & 0.021 & 0.13 & $168.6$ \\  
        Cauchy ($\alpha_{\vecuse{s}} = \frac{1}{2}$) & Cauchy ($\alpha_t = \frac{1}{2}$) & Sym&0.22 &0.0046 & 2.50 & 0.015& - & - \\ 
        Cauchy ($\alpha_{\vecuse{s}} = \frac{1}{2}$) & Cauchy ($\alpha_t = \frac{1}{2}$) & Asym&0.22 &0.0046 & 2.44& 0.018 & 0.20 & $-146.5$ \\ 
        Cauchy ($\alpha_{\vecuse{s}}$) & Cauchy ($\alpha_t$) & Sym&0.27 &0.0074 & 1.12 & 0.014& - & - \\ 
        Cauchy ($\alpha_{\vecuse{s}}$) & Cauchy ($\alpha_t$) & Asym&0.26 &0.0068 & 1.11& 0.015 & 0.14 & $-154.7$ \\ 
        \multicolumn{2}{|c}{GpGp exponential\_spacetime} &Sym& 0.16 &0.0025 & 1.02 & 0.000 & - & - \\ 
        \multicolumn{2}{|c}{GpGp matern\_spacetime} &Sym& 0.20 &0.0060 & 2.00 & 0.018 & -&-\\ 
        \multicolumn{2}{|c}{Lagrangian Sq.~Exp.} & Asym & 0.20 & 0.0029 & - & 0.022 & -& - \\ \hline
    \end{tabular}
    \caption{Estimated covariance parameters for the PM$_{2.5}$ data. The column $\textrm{atan2}(\tilde{\vecuse{x}})$ represents angle from the vector $(1,0)^\top$ (directly East) converted to degrees. The symmetric and asymmetric Cauchy with general exponent had estimated $\alpha_{\vecuse{s}} = 0.20$ and $\alpha_{\vecuse{s}} =0.23$, respectively, and $\alpha_t = 1.36$ and $\alpha_t = 1.32$, respectively. The additional parameters for the Lagrangian model were $\vecuse{\mu}_{\vecuse{V}} = (57.9, -48.3)^\top$ in kilometers (corresponding to $\textrm{atan2}(\vecuse{\mu}_{\vecuse{V}}) = -39.8$ degrees), and $(\matuse{\Sigma}_{\vecuse{V}})_{11}= 20{,}256$, $(\matuse{\Sigma}_{\vecuse{V}})_{22} = 117{,}780$, and $(\matuse{\Sigma}_{\vecuse{V}})_{12} = (
\matuse{\Sigma}_{\vecuse{V}})_{21} = -47{,}393$. The GpGp matern\_spacetime model had estimated parameter $\nu = 1.76$. }
    \label{tab:pm_parameters}
\end{table}

\section{Proofs}\label{app:proofs}

\begin{proof}[Proof of Proposition 1.]
    We aim to compute \begin{align*}
    C_{jk}^{\Im}(\vecuse{h}) &= \int_{\mathbb{R}^d} \textrm{exp}\left(\I\langle \vecuse{h}, \vecuse{x}\rangle\right) \left\{-\I\textrm{sign}(\langle \vecuse{x}, \tilde{\vecuse{x}}\rangle)\right\} \frac{1}{(2a\sqrt{\pi})^{d}}\textrm{exp}\left(-\frac{\lVert \vecuse{x}\rVert^2}{4a^2}\right) d\vecuse{x}.
\end{align*}
Consider the variable transformation $\vecuse{w} = (w_1, \dots, w_d)^\top = O\vecuse{x}$ where $O$ is an orthonormal rotation matrix such that  $\textrm{sign}(\langle \vecuse{x}, \tilde{\vecuse{x}}\rangle) = \textrm{sign}(\langle O^{\top}\vecuse{w}, \tilde{\vecuse{x}}\rangle) = \textrm{sign}(w_1)$. That is, we choose $O$ such that the direction of $w_1$ aligns with $\tilde{\vecuse{x}}$. 
We have $\lVert \vecuse{x}\rVert^2 = \vecuse{x}^\top \vecuse{x} = \vecuse{w}^\top O O^\top \vecuse{w} = \vecuse{w}^\top \vecuse{w}  = \lVert \vecuse{w}\rVert^2$, $\langle \vecuse{h}, O^\top \vecuse{w}\rangle$ = $\vecuse{h}^\top O^\top \vecuse{w} = (O\vecuse{h})^{\top}\vecuse{w}$, and $d\vecuse{w} = |O|d\vecuse{x} = d\vecuse{x}$, so the integral is written \begin{align*}
        C_{jk}^{\Im}(\vecuse{h}) &= \int_{\mathbb{R}^d} \textrm{exp}\left(\I\langle O\vecuse{h}, \vecuse{w}\rangle\right) \left\{-\I\textrm{sign}(w_1)\right\} \frac{1}{(2a\sqrt{\pi})^{d}} \textrm{exp}\left(-\frac{\lVert \vecuse{w}\rVert^2}{4a^2}\right) d\vecuse{w}.
\end{align*}

Now, break up the integral into $d-1$ and $1$ dimensions, letting the Hadamard product of $\vecuse{w}$ with the first standard basis vector be $\vecuse{w} \odot \vecuse{e}_1 = (w_1, 0, \dots, 0)^\top$: \begin{align*}
        &C_{jk}^{\Im}(\vecuse{h}) = \left[\int_{\mathbb{R}}\textrm{exp}\left(\I\langle O\vecuse{h}, \vecuse{w}\odot \vecuse{e}_1\rangle\right) \left\{-\I \textrm{sign}(w_1) \right\}\frac{1}{2a\sqrt{\pi}}\textrm{exp}\left(-\frac{w_1^2}{4a^2}\right) dw_1\right]\\
        &~~~~\times \left[\int_{\mathbb{R}^{d-1}} \textrm{exp}\left\{\I\langle O\vecuse{h}, \vecuse{w} - \vecuse{w}\odot \vecuse{e}_1\rangle\right\}  \frac{1}{(2a\sqrt{\pi})^{d-1}}\textrm{exp}\left(-\frac{\sum_{j=2}^d w_j^2}{4a^2}\right) dw_2\dots dw_{d}\right].
\end{align*}
The second integral (which is separated as it does not technically depend on $w_1$) just corresponds with the Gaussian spectral density in $d-1$ dimensions with the first entry of $O\vecuse{h}$ zeroed out, resulting in $
    \textrm{exp}\left\{-a^2 \lVert O\vecuse{h}  - (O\vecuse{h} \odot \vecuse{e}_1)\rVert^2\right\}$.
The first integral corresponds to the $d=1$ asymmetric case applied to the first entry $O\vecuse{h} \odot\vecuse{e}_1$: $   \textrm{exp}\left(- a^2\lVert O\vecuse{h} \odot \vecuse{e}_1\rVert^2\right)\textrm{erfi}\left\{a \left(\vecuse{h}^\top O^\top \vecuse{e}_1\right)\right\}.$
As $O\vecuse{h} \odot \vecuse{e}_1$ and $O\vecuse{h} - O\vecuse{h} \odot \vecuse{e}_1$ are orthogonal, the product of the two integrals is $
    C_{jk}^{\Im}(\vecuse{h}) = \textrm{exp}\left(-a^2 \lVert O\vecuse{h}\rVert^2\right)\textrm{erfi}\left(a  \vecuse{h}^\top O^\top \vecuse{e}_1\right)$. 
Converting back into the original orientation we have $C_{jk}^{\Im}(\vecuse{h})=\textrm{exp}\left(-a^2 \lVert \vecuse{h}\rVert^2\right)\textrm{erfi}\left(a\langle \vecuse{h}, \tilde{\vecuse{x}}\rangle\right)$.\end{proof}

\begin{proof}[Proof of Proposition 2.]

By the appropriate assumption, the one-dimensional Hilbert transform with respect to the direction of $\tilde{\vecuse{x}}$ exists and is well defined \cite{king2009hilbert}. 
We will show that they correspond. 
Without loss of generality, suppose that $\tilde{\vecuse{x}} = \vecuse{e}_1$, the first standard basis vector. 
The Hilbert transform in this direction is, where $\vecuse{h}_{\tau_1}$ is defined as $(\tau_1, h_2, \dots, h_d)^\top$,  \begin{align*}
\frac{1}{\pi} {\rm p.v.} \int_{-\infty}^\infty \frac{C_{jk}^{\Re}(\vecuse{h}_{\tau_1})}{h_1 - \tau_1} d\tau_1&= \frac{1}{\pi} {\rm p.v.} \int_{-\infty}^\infty \frac{\int_0^\infty \textrm{exp}\left\{- v \lVert \vecuse{h}_{\tau_1}\rVert^2\right\} g(v) dv}{h_1 - \tau_1} d\tau_1,
\end{align*}where ${\rm p.v.}$ denotes the principal value. 
Assuming $g(v)$ is integrable on $(0, \infty)$, suppose we can exchange the order of integration so that this is \begin{align*}
    &\frac{1}{\pi}\int_0^\infty {\rm p.v.} \int_{-\infty}^\infty \frac{ \textrm{exp}\left\{- v \lVert \vecuse{h}_{\tau_1}\rVert^2\right\} g(v) }{h_1 - \tau_1} d\tau_1dv \\
    &~~~~~= \frac{1}{\pi}\int_0^\infty  {\rm p.v.} \int_{-\infty}^\infty \frac{ \textrm{exp}\left(- v \tau_1^2\right) }{h_1 - \tau_1} d\tau_1 g(v) \textrm{exp}\left(- v \sum_{k = 2}^d h_k^2\right)dv
\end{align*}The inner integral now corresponds to the one-dimensional Hilbert transform of the squared-exponential covariance, to give \begin{align*}
    &\int_0^\infty \textrm{exp}\left(-vh_1^2\right) \textrm{erfi}(\sqrt{v} h_1) g(v) \textrm{exp}\left(- v \sum_{k = 2}^d h_k^2\right)dv \\
    &~~~~ = \int_0^\infty \textrm{exp}\left(- v \lVert \vecuse{h}\rVert^2\right) \textrm{erfi}(\sqrt{v}h_1) g(v) dv.
\end{align*}
The result for general $\tilde{\vecuse{x}}$ can be found with a change of variable.

It can be shown similarly by exchanging the order of integration with the spectral density. We are interested in \begin{align*}
C_{jk}^{\Im}(\vecuse{h}) &=\int_{\mathbb{R}^d} \textrm{exp}\left(\I\langle \vecuse{h}, \vecuse{x}\rangle\right) \{ -\I \textrm{sign}(\langle \vecuse{x}, \tilde{\vecuse{x}}\rangle)\}f(\vecuse{x}) d\vecuse{x}
\end{align*}where $f(\vecuse{x})$ is the covariance's spectral density. 
From Proposition 1 of \cite{allard_fully_2022}, the spectral density (when it exists) can be represented through the mixture density $g(v)$\begin{align*}
    f(\vecuse{x}) &= \int_0^\infty \frac{1}{(2\sqrt{\pi})^{d}}\textrm{exp}\left(-\frac{\lVert \vecuse{x}\rVert^2}{4v}\right)v^{-\frac{d}{2}} g(v) dv,
\end{align*}so that \begin{align*}
    C_{jk}^{\Im}(\vecuse{h}) &=  \int_{\mathbb{R}^d} \textrm{exp}\left(\I\langle \vecuse{h}, \vecuse{x}\rangle\right) \{-\I\textrm{sign}(\langle \vecuse{x}, \tilde{\vecuse{x}}\rangle)\}\frac{1}{(2\sqrt{\pi})^d}\int_0^\infty \textrm{exp}\left(-\frac{\lVert \vecuse{x}\rVert^2}{4v}\right)v^{-\frac{d}{2}} g(v) dv d\vecuse{x} \\ 
    &= \int_0^\infty \int_{\mathbb{R}^d} \textrm{exp}\left(\I\langle \vecuse{h}, \vecuse{x}\rangle\right)\left\{-\I\textrm{sign}(\langle \vecuse{x}, \tilde{\vecuse{x}}\rangle)\right\}\frac{1}{(2\sqrt{\pi})^d}\textrm{exp}\left(-\frac{\lVert \vecuse{x}\rVert^2}{4v}\right) d\vecuse{x}v^{-\frac{d}{2}} g(v) dv
\end{align*}by Fubini's Theorem. The result then follows directly from Proposition 1, as the inner integral satisfies  \begin{align*}
    &\int_{\mathbb{R}^d} \textrm{exp}\left(\I\langle \vecuse{h}, \vecuse{x}\rangle\right) \frac{\left\{-\I\textrm{sign}(\langle \vecuse{x}, \tilde{\vecuse{x}}\rangle)\right\}}{(2\sqrt{\pi})^d}\textrm{exp}\left(-\frac{\lVert \vecuse{x}\rVert^2}{4v}\right) d\vecuse{x} \\
    &~~~~~~= v^{\frac{d}{2}} \textrm{exp}\left(-v\lVert \vecuse{h}\rVert^2 \right)\textrm{erfi}\left(\sqrt{v}\langle \vecuse{h}, \tilde{\vecuse{x}}\rangle\right).
\end{align*}\end{proof}

    \begin{proof}[Proof of Proposition 3.]
    From \cite{allard2025modeling}, the corresponding Cauchy mixing density is Gamma-distributed: $
        g(v) = a^{-2\alpha}/\{\Gamma(\alpha)\}v^{\alpha - 1} \textrm{exp}\left(-v/a^2\right)$, such that \begin{align*}
        C_{jk}^{\Re}(\vecuse{h})&= \frac{1}{\left(1 + a^2 \lVert \vecuse{h}\rVert^2\right)^\alpha} = \int_0^\infty e^{-v\lVert \vecuse{h}\rVert^2}\frac{a^{-2\alpha}}{\Gamma(\alpha)} v^{\alpha - 1} \textrm{exp}\left(-\frac{v}{a^2}\right) dv.
    \end{align*}
    Therefore, we are interested in  \begin{align*}
            C_{jk}^{\Im}(\vecuse{h})
            &=  \frac{a^{-2\alpha}}{\Gamma(\alpha)}\int_0^\infty \textrm{exp}\left\{-v\left(\lVert \vecuse{h}\rVert^2 + \frac{1}{a^2}\right)\right\}\textrm{erfi}\left(\sqrt{v} \langle \vecuse{h}, \tilde{\vecuse{x}}\rangle\right) v^{\alpha - 1} dv\\
            &=  \frac{a^{-2\alpha}}{\Gamma(\alpha)}\int_0^\infty \textrm{exp}\left\{-v\left(\lVert \vecuse{h}\rVert^2 + \frac{1}{a^2}\right)\right\} \frac{2\sqrt{v} \langle \vecuse{h}, \tilde{\vecuse{x}}\rangle}{\sqrt{\pi}}{}_1F_1\left(\frac{1}{2}; \frac{3}{2};  v\langle \vecuse{h}, \tilde{\vecuse{x}}\rangle^2\right)v^{\alpha - 1} dv,
        \end{align*}where the last line follows from 13.6.19 of \cite{abramowitz1948handbook}. 
        We consider the change of variables $w = v\langle \vecuse{h}, \tilde{\vecuse{x}}\rangle^2$, so that \begin{align*}
            C_{jk}^{\Im}(\vecuse{h})
            &=  \frac{a^{-2\alpha}}{\Gamma(\alpha)}\frac{2}{\sqrt{\pi}}\langle \vecuse{h}, \tilde{\vecuse{x}}\rangle^{-2\alpha } \\
            &~~~~~\times\int_0^\infty \textrm{exp}\left\{-w\langle \vecuse{h}, \tilde{\vecuse{x}}\rangle^{-2}\left(\lVert \vecuse{h}\rVert^2 + \frac{1}{a^2}\right)\right\}{}_1F_1\left(\frac{1}{2}; \frac{3}{2};  w\right)w^{\alpha - \frac{1}{2}} dw.
        \end{align*}
        
        The result follows from 7.621.4 of \cite{gradshteyn2014table}, which specifies \begin{align*}
            \int_0^\infty w^{\alpha^*-1} e^{-c^*t} {}_1F_1(a^*; b^*; w) dw &=(c^*)^{-\alpha^*}\Gamma(\alpha^*) {}_2F_1(a^*, \alpha^*; b^*; 1/c^*). 
        \end{align*}
        We set $\alpha^* = \alpha +1/2$, $c^* = \langle \vecuse{h}, \tilde{\vecuse{x}}\rangle^{-2}\left(\lVert \vecuse{h}\rVert^2 + a^{-2}\right)$, $a^* = 1/2$, and $b^* = 3/2$ to get
        \begin{align*}
            C_{jk}^{\Im}(\vecuse{h}) &= \frac{a^{-2\alpha}}{\Gamma(\alpha)}\frac{2}{\sqrt{\pi}}\langle \vecuse{h}, \tilde{\vecuse{x}}\rangle^{-2\alpha} 
           \left(\frac{\lVert \vecuse{h}\rVert^2 + a^{-2}}{\langle \vecuse{h}, \tilde{\vecuse{x}}\rangle^2}\right)^{-\alpha - \frac{1}{2}}\Gamma\left(\alpha + \frac{1}{2}\right) \\
           &~~~~~~~~~~~~~~~~\times {}_2F_1\left(\frac{1}{2}, \frac{1}{2} + \alpha; \frac{3}{2}; \frac{\langle \vecuse{h}, \tilde{\vecuse{x}}\rangle^2}{\lVert \vecuse{h}\rVert^2 + a^{-2}}\right) \\
&= \frac{1}{\left(a^2\lVert \vecuse{h}\rVert^2 + 1\right)^\alpha} \frac{2\Gamma(\alpha + \frac{1}{2})}{\sqrt{\pi}\Gamma(\alpha)}\frac{a\langle \vecuse{h}, \tilde{\vecuse{x}}\rangle}{\left(a^2\lVert \vecuse{h}\rVert^2 + 1\right)^{ \frac{1}{2}}} {}_2F_1\left(\frac{1}{2}, \frac{1}{2} + \alpha; \frac{3}{2}; \frac{a^2\langle \vecuse{h}, \tilde{\vecuse{x}}\rangle^2}{a^2\lVert \vecuse{h}\rVert^2 + 1}\right).
        \end{align*}\end{proof}

For Propositions \ref{prop:half_spectral}, \ref{prop:sq_exp_asymmetric_gneiting}, and \ref{prop:gneiting_cauchy}, we first introduce and prove Lemma \ref{lemma:gneiting_sq_exp_spectral_density}. 
Note that Proposition \ref{prop:sq_exp_asymmetric_gneiting} requires Proposition \ref{prop:half_spectral}, and Proposition  \ref{prop:gneiting_cauchy} requires Proposition \ref{prop:sq_exp_asymmetric_gneiting}, so we follow this order. 
\begin{lemma}[Spectral density of squared-exponential]\normalfont\label{lemma:gneiting_sq_exp_spectral_density}
Consider a version of the Gneiting class with a squared-exponential spatial covariance:
\begin{align}
    C(\vecuse{h},u) &= \frac{\sigma}{(\gamma(u) + 1)^{\frac{d}{2}}} \textrm{exp}\left(- \frac{a_{\vecuse{s}}^2 \lVert \vecuse{h}\rVert^2}{\gamma(u) +1}\right).\label{eq:sq_exp_gne}
\end{align}
The spectral density of this $C(\vecuse{h},u)$ in \eqref{eq:sq_exp_gne} when $\gamma(u) = a_t^2u^2$ is \begin{align*}
     f^*(\lVert\vecuse{x}\rVert,|\eta|) &=  \frac{a_{\vecuse{s}}}{(2\sqrt{\pi} a_{\vecuse{s}})^{d} \sqrt{\pi}a_t\lVert \vecuse{x}\rVert}\textrm{exp}\left(-\frac{\lVert \vecuse{x}\rVert^2}{4a_{\vecuse{s}}^2} - \frac{a_{\vecuse{s}}^2\eta^2}{a_t^2\lVert \vecuse{x}\rVert^2}\right).
\end{align*}
Furthermore, when $\gamma(u) = a_t |u|$, then the spectral density is \begin{align*}
     f^*(\lVert\vecuse{x}\rVert,|\eta|) &= \frac{1}{(2\sqrt{\pi}a_{\vecuse{s}})^d a_t \pi } \textrm{exp}\left(-\frac{\lVert \vecuse{x}\rVert^2}{4a_{\vecuse{s}}^2}\right)\left(\frac{\lVert \vecuse{x}\rVert^2}{4a_{\vecuse{s}}^2} + \frac{4a_{\vecuse{s}}^2\eta^2}{a_t^2\lVert \vecuse{x}\rVert^2}\right)^{-1}.
\end{align*}
\end{lemma}

        \begin{proof}[Proof of Lemma \ref{lemma:gneiting_sq_exp_spectral_density}.]
    First consider the case $\gamma(u) = a_t^2u^2$. We are interested in the inverse formula \begin{align*}
     f^*(\lVert\vecuse{x}\rVert,|\eta|)&=(2\pi)^{-d-1} \int_{\mathbb{R}}\int_{\mathbb{R}^d} \textrm{exp}\left(-\I\langle \vecuse{h}, \vecuse{x}\rangle - \I \eta u\right)C(\vecuse{h},u) d\vecuse{h}du  \\
    &= (2\pi)^{-d-1} \int_{\mathbb{R}}\int_{\mathbb{R}^d}  \frac{\textrm{exp}\left(-\I\langle \vecuse{h}, \vecuse{x}\rangle  - \I \eta u\right)}{\left(a_t^2 u^2 + 1\right)^{\frac{d}{2}}}\textrm{exp}\left(-\frac {a_{\vecuse{s}}^2\lVert \vecuse{h}\rVert^2}{a_t^2u^2 + 1}\right) d\vecuse{h}du \\
    &= (2\pi)^{-d-1} \int_{\mathbb{R}} \frac{\textrm{exp}\left(- \I \eta u\right)}{\left(a_t^2u^2 + 1\right)^{\frac{d}{2}}}\int_{\mathbb{R}^d} \textrm{exp}\left(-\I\langle \vecuse{h}, \vecuse{x}\rangle \right)\textrm{exp}\left(-\frac {a_{\vecuse{s}}^2\lVert \vecuse{h}\rVert^2}{a_t^2u^2 + 1}\right) d\vecuse{h}du.
    \end{align*}
    From the inverse Fourier transform of the squared-exponential function, we have
\begin{align*}
     f^*(\lVert\vecuse{x}\rVert,|\eta|)&= \frac{1}{(2\pi)^{d+1}} \int_{\mathbb{R}} \frac{\textrm{exp}\left(- \I \eta u\right)}{(a_t^2u^2 + 1)^{\frac{d}{2}}}\textrm{exp}\left\{-\frac{\lVert \vecuse{x}\rVert^2}{4a_{\vecuse{s}}^2}\left(1+a_t^2u^2\right)\right\}\frac{\pi^{\frac{d}{2}}}{a_{\vecuse{s}}^{d}}(a_t^2u^2 + 1)^{\frac{d}{2}}du\\
    &= \frac{1}{2\pi (2a_{\vecuse{s}}\sqrt{\pi})^d}  \int_{\mathbb{R}} \textrm{exp}\left(- \I \eta u\right)\textrm{exp}\left\{-\frac{\lVert \vecuse{x}\rVert^2}{4a_{\vecuse{s}}^2}\left(1+a_t^2u^2\right)\right\}du\\
    &=  \frac{a_{\vecuse{s}}}{\sqrt{\pi} (2a_{\vecuse{s}}\sqrt{\pi})^d a_t\lVert \vecuse{x}\rVert} 
    \textrm{exp}\left(-\frac{\lVert \vecuse{x}\rVert^2}{4a_{\vecuse{s}}^2} - \frac{a_{\vecuse{s}}^2\eta^2}{a_t^2\lVert \vecuse{x}\rVert^2}\right).
    \end{align*}

    For the case $\gamma(u) = a_t |u|$, the integration with respect to $h$ is similar, leaving \begin{align*}
     f^*(\lVert\vecuse{x}\rVert,|\eta|) &= \frac{1}{2\pi (2a_{\vecuse{s}}\sqrt{\pi})^d}  \int_{\mathbb{R}} \textrm{exp}\left(- \I \eta u\right)\textrm{exp}\left(-\frac{\lVert \vecuse{x}\rVert^2}{4a_{\vecuse{s}}^2}\left(1+a_t|u|\right)\right)du\\
        &= \frac{1}{2\pi (2a_{\vecuse{s}}\sqrt{\pi})^d} \textrm{exp}\left(-\frac{\lVert \vecuse{x}\rVert^2}{4a_{\vecuse{s}}^2}\right) \int_{\mathbb{R}} \textrm{exp}\left(- \I \eta u\right)\textrm{exp}\left(-\frac{\lVert \vecuse{x}\rVert^2}{4a_{\vecuse{s}}^2}a_t|u|\right)du.
\end{align*}
The integral corresponds to the spectral density of the exponential covariance function, resulting in: \begin{align*}
     f^*(\lVert\vecuse{x}\rVert,|\eta|) &=\frac{1}{2\pi (2a_{\vecuse{s}}\sqrt{\pi})^d} \textrm{exp}\left(-\frac{\lVert \vecuse{x}\rVert^2}{4a_{\vecuse{s}}^2}\right) \frac{2\frac{\lVert \vecuse{x}\rVert^2a_t}{4a_{\vecuse{s}}^2}}{\frac{\lVert \vecuse{x}\rVert^4a_t^2}{16a_{\vecuse{s}}^4} + \eta^2}.
\end{align*}
The result then follows from simplifying the fraction. \end{proof}

\begin{proof}[Proof of Proposition \ref{prop:half_spectral}.]
The symmetric part follows from Lemma \ref{lemma:gneiting_sq_exp_spectral_density}.
The asymmetric part for $\gamma(u) = a_t^2u^2$ is
\begin{align*}
    \sigma\xi     & \int_{\mathbb{R}^d} \int_{\mathbb{R}}\textrm{exp}\left(\I\langle \vecuse{h}, \vecuse{x}\rangle + \I u\eta \right)\textrm{sign}(\langle \vecuse{x}, \tilde{\vecuse{x}}\rangle)\textrm{sign}(\eta) f^*(\lVert \vecuse{x}\rVert, |\eta|) d\eta d\vecuse{x}, \end{align*}with $f^*(\lVert\vecuse{x}\rVert, |\eta|)$ as in Lemma \ref{lemma:gneiting_sq_exp_spectral_density}.
    This is 
    \begin{align*}
    &\sigma\xi \int_{\mathbb{R}^d} \int_{\mathbb{R}}\exp\left(\I\langle \vecuse{h}, \vecuse{x}\rangle  + \I u\eta \right)\frac{\textrm{sign}(\langle \vecuse{x}, \tilde{\vecuse{x}}\rangle)\textrm{sign}(\eta)a_{\vecuse{s}}}{\left(2\sqrt{\pi}a_{\vecuse{s}}\right)^d\sqrt{\pi} a_t \lVert \vecuse{x}\rVert}\textrm{exp}\left(-\frac{ \lVert \vecuse{x}\rVert^2}{4a_{\vecuse{s}}^2} - \frac{a_{\vecuse{s}}^2\eta^2}{a_t^2 \lVert \vecuse{x}\rVert^2}\right)  d\eta d\vecuse{x}\\
        &~~= 
    \sigma\xi \int_{\mathbb{R}^d} \frac{\exp\left(\I \langle \vecuse{h}, \vecuse{x}\rangle\right)\textrm{sign}(\langle \vecuse{x}, \tilde{\vecuse{x}}\rangle)a_{\vecuse{s}}}{\left(2\sqrt{\pi}a_{\vecuse{s}}\right)^d\sqrt{\pi} a_t \lVert \vecuse{x}\rVert}\textrm{exp}\left(-\frac{ \lVert \vecuse{x}\rVert^2}{4a_{\vecuse{s}}^2}\right) 2\int_0^\infty \sin(u\eta )\textrm{exp}\left(- \frac{a_{\vecuse{s}}^2\eta^2}{a_t^2 \lVert \vecuse{x}\rVert^2}\right)  d\eta d\vecuse{x}, 
\end{align*}as the inner integrand is an odd function of $\eta$. 
The inner integral is related to Dawson's function \citep[3.896.3,][]{gradshteyn2014table}: $
   (1/2)\int_0^\infty \textrm{exp}\left(-t^2/4\right) \sin(zt) dt =  D_+(z)$.
Applied here with the change of variables $w = 2a_{\vecuse{s}}\eta/(a_t\lVert \vecuse{x}\rVert)$, we have 
   \begin{align*}
   2\int_0^\infty \sin(\eta u)\textrm{exp}\left(-\frac{a_{\vecuse{s}}^2\eta^2}{a_t^2  \lVert \vecuse{x}\rVert^2}\right)d\eta&=2 \frac{a_t \lVert \vecuse{x}\rVert}{2a_{\vecuse{s}}}\int_0^\infty \sin\left(w\frac{ \lVert  \vecuse{x}\rVert}{2a_{\vecuse{s}}} a_tu\right)\textrm{exp}\left(-\frac{w^2}{4}\right)dw \\
   &=\frac{2a_t \lVert \vecuse{x}\rVert}{a_{\vecuse{s}}}D_+\left(\frac{\lVert \vecuse{x}\rVert}{2a_{\vecuse{s}}}a_t u\right),
\end{align*} so that we have
\begin{align*}
    \sigma \xi \int_{\mathbb{R}^d} \textrm{exp}\left(\I \langle \vecuse{h},\vecuse{x} \rangle\right)\frac{2\textrm{sign}(\langle \vecuse{x}, \tilde{\vecuse{x}}\rangle)}{\left(2\sqrt{\pi}a_{\vecuse{s}}\right)^d\sqrt{\pi} } \textrm{exp}\left(-\frac{ \lVert \vecuse{x}\rVert^2}{4a_{\vecuse{s}}^2}\right) D_+\left(\frac{\lVert \vecuse{x}\rVert}{2a_{\vecuse{s}}}a_tu\right) d\vecuse{x}.
\end{align*}
Now, Dawson's function is related to the imaginary error function \citep[7.5.1,][]{NIST:DLMF}: \begin{align*}
    D_+\left(\frac{\lVert \vecuse{x}\rVert}{2a_\vecuse{s}}a_tu\right)&=\frac{\sqrt{\pi}}{2} \textrm{exp}\left(- \frac{ \lVert \vecuse{x}\rVert^2}{4a_{\vecuse{s}}^2}a_t^2u^2\right)\textrm{erfi}\left(\frac{\lVert \vecuse{x}\rVert}{2a_{\vecuse{s}}} a_tu\right).
\end{align*}
Therefore, we can write the remaining integral as \begin{align*}
    & \sigma\xi \int_{\mathbb{R}^d}\textrm{exp}\left(\I \langle \vecuse{h},\vecuse{x} \rangle\right)\frac{\textrm{sign}(\langle \vecuse{x}, \tilde{\vecuse{x}}\rangle)}{(2\sqrt{\pi}a_{\vecuse{s}})^d}\textrm{exp}\left\{-\frac{ \lVert \vecuse{x}\rVert^2}{4a_{\vecuse{s}}^2}\left(a_t^2u^2 + 1\right)\right\}\textrm{erfi}\left(\frac{\lVert  \vecuse{x}\rVert }{2a_{\vecuse{s}}}a_tu\right) d\vecuse{x}
\end{align*} and \begin{align*}
        &\sigma \xi \int_{\mathbb{R}^d}\textrm{exp}\left(\I \langle \vecuse{h},\vecuse{x} \rangle\right)\textrm{sign}(\langle \vecuse{x}, \tilde{\vecuse{x}}\rangle) f\left(\vecuse{x}\sqrt{a_t^2u^2 + 1}\right)   \textrm{erfi}\left(\frac{\lVert  \vecuse{x}\rVert }{2a_{\vecuse{s}}}a_tu\right) d\vecuse{x},
\end{align*}where $f(\vecuse{x})$ is the spectral density of the squared-exponential covariance with parameter $a_{\vecuse{s}}$.

We now consider $\gamma(u) = a_t|u|$. 
The symmetric half-spectral representation is \begin{align*}
    \sigma \int_{\mathbb{R}^d} \textrm{exp}\left(\I \langle \vecuse{h}, \vecuse{x}\rangle \right)\frac{1}{(2\sqrt{\pi}a_{\vecuse{s}})^d}\textrm{exp}\left\{-\frac{\lVert \vecuse{x}\rVert^2}{4a_{\vecuse{s}}^2}(a_t|u| + 1)\right\} dx.
\end{align*}
The asymmetric representation can be found by applying the Hilbert transform of $u$, which, for $\textrm{exp}(-a|u|)$ is demonstrated in \cite{yarger2023multivariate}; as well as multiplying a $\textrm{sign}(\langle \vecuse{x}, \tilde{\vecuse{x}}\rangle)$ term.
This can also be found by relating the result in Lemma \ref{lemma:gneiting_sq_exp_spectral_density} to the spectral density of the exponential covariance function.
\end{proof}

\begin{proof}[Proof of Proposition \ref{prop:sq_exp_asymmetric_gneiting}.]

The symmetric part follows directly from Lemma \ref{lemma:gneiting_sq_exp_spectral_density}.

For the asymmetric part when $d=1$, we use Proposition \ref{prop:half_spectral}. The integrand is again an odd function, and thus reduces to:
\begin{align*}
    \sigma\xi \frac{1}{a_s\sqrt{\pi}}\int_0^\infty \textrm{sin}(hx)\textrm{exp}\left\{-\frac{ x^2}{4a_s^2} \left(a_t^2u^2 + 1\right)\right\}   \textrm{erfi}\left(\frac{ x }{2a_s}a_tu\right) dx.
\end{align*}The integral 2.13.1.5 of \cite{korotkov2020integrals} states that \begin{align*}
    \int_0^\infty \sin(bz)e^{-\gamma z^2}\textrm{erfi}(cz)  dz = \sqrt{\frac{\pi}{4\gamma}}e^{-\frac{b^2}{4\gamma}}\textrm{erf}\left(\frac{cb}{2\sqrt{\gamma^2 - c^2\gamma}}\right)
\end{align*}for $\gamma > c^2$. Setting $c = a_tu/(2a_s)$, $\gamma = \left(a_t^2u^2+1\right)/\left(4a_s^2\right)$ and $b=h$, the condition $\gamma > c^2$ is satisfied, and we have that the integral is \begin{align*}
    &\sigma\xi\frac{1}{a_s\sqrt{\pi}} \sqrt{\frac{\pi}{4} \frac{4a_s^2}{a_t^2u^2 + 1}} \textrm{exp}\left(- \frac{a_s^2h^2}{a_t^2u^2 + 1}\right) \textrm{erf}\left(\frac{a_tu h}{4a_s\sqrt{\frac{( a_t^2u^2 + 1)^2}{(4a_s^2)^2} - \frac{a_t u^2  (a_t^2u^2 + 1)}{ (4a_s^2)^2}}}\right)\\ 
    &~~~~~~~=  \frac{\sigma\xi}{\sqrt{a_t^2u^2 + 1}} \textrm{exp}\left(- \frac{a_s^2h^2}{a_t^2u^2 + 1}\right) \textrm{erf}\left(\frac{a_s ha_tu }{\sqrt{a_t^2u^2 + 1}}\right),
    \end{align*}
    which is the result. \end{proof}

        \begin{proof}[Proof of Proposition \ref{prop:gneiting_cauchy}.]
            We mix the asymmetric part of the squared-exponential covariance: 
\begin{align*}
        \int_0^\infty \frac{\sigma\xi}{\sqrt{a_t^2u^2 + 1}} \textrm{exp}\left(-\frac{v h^2}{a_t^2u^2+1}\right) \textrm{erf}\left(\frac{h\sqrt{v}a_tu}{\sqrt{a_t^2u^2+1}}\right) g(v) dv,
\end{align*}
where $g(v)$ is the probability density function of a Gamma random variable: \begin{align*}
    g(v) &= \frac{a_s^{-2\alpha}}{\Gamma(\alpha)} v^{\alpha - 1} \textrm{exp}\left(-\frac{v}{a_s^2}\right).
\end{align*}This results in the integral \begin{align*}
       \frac{\sigma\xi}{\sqrt{a_t^2u^2 + 1}} \frac{a_s^{-2\alpha}}{\Gamma(\alpha)}    \int_0^\infty  \textrm{exp}\left\{-v\left(\frac{ h^2}{a_t^2u^2+1} + \frac{1}{a_s^2}\right)\right\} \textrm{erf}\left(\sqrt{v}\frac{ha_tu}{\sqrt{a_t^2u^2+1}}\right) v^{\alpha - 1} dv.
\end{align*}
This becomes \begin{align*}
            C^*(h, u)
            &=  \frac{\sigma \xi a_s^{-2\alpha}}{\Gamma(\alpha)}\frac{1}{\sqrt{a_t^2u^2 + 1}}\frac{2}{\sqrt{\pi}}\frac{ha_tu}{\sqrt{a_t^2u^2+1}}\\
            &~~~
            \times\int_0^\infty \textrm{exp}\left\{-v\left(\frac{ h^2}{a_t^2u^2+1} + \frac{1}{a_s^2}\right)\right\}{}_1F_1\left(\frac{1}{2}; \frac{3}{2}; - v\frac{h^2a_t^2u^2}{a_t^2u^2+1}\right)v^{\alpha - \frac{1}{2}} dv,
        \end{align*}which follows from 13.6.19 of \cite{abramowitz1948handbook}. 
        We consider the change of variables $w = v(h^2a_t^2u^2)/(a_t^2u^2+1)$, so that \begin{align*}
            C^*(h,u)
            &=                \frac{\sigma\xi a_s^{-2\alpha}}{\Gamma(\alpha)}\frac{\textrm{sign}(h)\textrm{sign}(u)}{\sqrt{a_t^2u^2 + 1}}\frac{2}{\sqrt{\pi}}\left(\frac{a_t^2u^2 + 1}{h^2a_t^2u^2}\right)^{\alpha }  \\
            &~~~\times \int_0^\infty \textrm{exp}\left[-w\left\{\frac{1}{a_t^2u^2}  + \frac{a_t^2u^2+1}{(a_sha_tu)^2}\right\}\right]{}_1F_1\left(\frac{1}{2}; \frac{3}{2}; - w\right)w^{\alpha - \frac{1}{2}} dw.
        \end{align*}
        
        The result follows from 7.621.4 of \cite{gradshteyn2014table}, which specifies \begin{align*}
            \int_0^\infty w^{\alpha^*-1} e^{-c^*w} {}_1F_1(a^*; b^*; w) dw &=(c^*)^{-\alpha^*}\Gamma(\alpha^*) {}_2F_1(a^*, \alpha^*; b^*; 1/c^*). 
        \end{align*}
        We set $\alpha^* = \alpha$, $c^* = \left(a_t^2u^2\right)^{-1} + \left(a_sha_tu\right)^{-2} \left(a_t^2u^2 + 1\right)$, $a^* = 1/2$, and $b^* = 3/2$: \begin{align*}
            C^*(h,u)&=                \frac{\sigma\xi }{\Gamma(\alpha)}\frac{\textrm{sign}(h)\textrm{sign}(u)}{\sqrt{a_t^2u^2 + 1}}\frac{2}{\sqrt{\pi}}\left(\frac{a_t^2u^2 + 1}{a_s^2h^2a_t^2u^2}\right)^{\alpha} \left(\frac{1}{a_t^2u^2}  + \frac{a_t^2u^2+1}{(a_sha_tu)^2}\right)^{-\alpha}\\
&~~~~~~~\times\Gamma(\alpha){}_2F_1\left\{\frac{1}{2}, \alpha; \frac{3}{2}; -\left(\frac{1}{a_t^2u^2}  + \frac{a_t^2u^2+1}{(a_sha_tu)^2}\right)^{-1}\right\}.
        \end{align*}
        We have \begin{align*}
            \left(\frac{1}{a_t^2u^2}  + \frac{a_t^2u^2+1}{(a_sha_tu)^2}\right)^{-\alpha} 
            = \left(\frac{a_s^2h^2a_t^2u^2 }{a_s^2h^2+  a_t^2u^2 + 1}\right)^{\alpha},
        \end{align*}so \begin{align*}
            C^*(h,u) 
            &= \sigma \xi \frac{\textrm{sign}(h)\textrm{sign}(u)}{\sqrt{a_t^2u^2 + 1}}\frac{2}{\sqrt{\pi}}\left(\frac{a_t^2u^2+1}{a_s^2h^2+  a_t^2u^2 + 1}\right)^{\alpha} \\
            &~~~~~~\times {}_2F_1\left(\frac{1}{2}, \alpha; \frac{3}{2}; -\frac{a_t^2u^2 a_s^2h^2}{a_s^2h^2+  a_t^2u^2 + 1}\right).
        \end{align*}

        Note that as\begin{align*}
        \frac{a_tu^2 + 1}{a_s^2h^2 + a_t^2u^2 + 1} = \left(\frac{a_s^2h^2}{a_t^2u^2+1 } + 1\right)^{-1},
    \end{align*} we can write this \begin{align*}
        C^*(h,u) &=     \sigma \xi\frac{\textrm{sign}(h)\textrm{sign}(u)}{\sqrt{a_t^2u^2 + 1}}\frac{2}{\sqrt{\pi}}\frac{1}{\left(\frac{a_s^2h^2}{a_t^2u^2+1 } + 1\right)^\alpha}{}_2F_1\left(\frac{1}{2}, \alpha; \frac{3}{2}; -\frac{a_t^2u^2 a_s^2h^2}{a_s^2h^2+  a_t^2u^2 + 1}\right).
    \end{align*}\end{proof}

\end{document}